\documentclass[a4paper,UKenglish,cleveref, autoref, thm-restate,authorcolumns]{lipics-v2019}

\title{Strategy Complexity of Parity Objectives in Countable MDPs} %

\titlerunning{Strategy Complexity of Parity Objectives in Countable MDPs} %

\author{Stefan Kiefer}{Department of Computer Science, University of Oxford, UK}{}{}{}

\author{Richard Mayr}{School of Informatics, University of Edinburgh, UK}{}{}{} 

\author{Mahsa Shirmohammadi}{CNRS \& IRIF,  Universit\'e de Paris, FR}{}{}{}

\author{Patrick Totzke}{Department of Computer Science, University of Liverpool, UK}{}{}{}

\authorrunning{S. Kiefer,  R. Mayr, M. Shirmohammadi and P. Totzke} %

\Copyright{Stefan Kiefer,  Richard Mayr, Mahsa Shirmohammadi and Patrick Totzke} %



\relatedversion{Full version of a paper presented at CONCUR 2020.}

\nolinenumbers %

\hideLIPIcs  %

\EventEditors{John Q. Open and Joan R. Access}
\EventNoEds{2}
\EventLongTitle{42nd Conference on Very Important Topics (CVIT 2016)}
\EventShortTitle{CVIT 2016}
\EventAcronym{CVIT}
\EventYear{2016}
\EventDate{December 24--27, 2016}
\EventLocation{Little Whinging, United Kingdom}
\EventLogo{}
\SeriesVolume{42}
\ArticleNo{23}
\usepackage{amsmath}

\usepackage[textwidth=2.5cm]{todonotes}
\newcounter{todocounter}

\usepackage{stmaryrd}

\usepackage{mathtools}
\usepackage{thm-restate}
\usepackage{xcolor}

\usepackage{tikz}
\usetikzlibrary{patterns}
\usetikzlibrary{arrows}
\usetikzlibrary{fadings}

\usetikzlibrary{calc}
\usetikzlibrary{fadings}
\usetikzlibrary{decorations.text}
\usetikzlibrary{angles}
\usetikzlibrary{arrows}
\usetikzlibrary{arrows.meta}
\usetikzlibrary{automata}
\usetikzlibrary{positioning}
\usetikzlibrary{shapes}

\tikzset{every state/.style={
            anchor=center,
            line width=0.6mm,
            fill=black!10,
            rounded corners=0.5mm,
            minimum size=0.7cm,inner sep=2pt,
}}
\tikzset{cstate/.style={state,rectangle}}
\tikzset{rstate/.style={state,circle}}
\tikzset{ucstate/.style={cstate, 
            minimum size=0.3cm,
}}

\tikzset{layer1/.style={
        fill=red!50!white,
        draw=red!80!black,
}}
\tikzset{layer2/.style={
        fill=blue!10,
        draw=blue!80!black,
}}

\tikzset{every edge/.append style={
  shorten <=2pt,
  shorten >=2pt,
  -{Triangle}, 
}}

\newcommand{\+}[1]{\mathbb{#1}}

\newcommand{\N}{\+{N}}

\newcommand{\x}{\times}

\usepackage{wasysym} %
\newcommand{\rsymbol}{\ocircle}
\newcommand{\zsymbol}{\Box}
\newcommand{\zstates}{\states_\zsymbol}
\newcommand{\rstates}{\states_\rsymbol}
\newcommand{\reachset}{T}

\newcommand{\eqby}[2][=]{\stackrel{\text{{\tiny{#2}}}}{#1}}
\newcommand{\eqdef}{\eqby{def}}

\newcommand{\eps}{\varepsilon}
\newcommand{\step}[2][]{\xrightarrow[#1]{#2}}

\newcommand{\problemx}[3]{
\par\noindent\underline{\sc#1}\par\nobreak\vskip.2\baselineskip
\begingroup\clubpenalty10000\widowpenalty10000
\setbox0\hbox{\bf INPUT:\ }\setbox1\hbox{\bf QUESTION:\ }
\dimen0=\wd0\ifnum\wd1>\dimen0\dimen0=\wd1\fi
\vskip-\parskip\noindent
\hbox to\dimen0{\box0\hfil}\hangindent\dimen0\hangafter1\ignorespaces#2\par
\vskip-\parskip\noindent
\hbox to\dimen0{\box1\hfil}\hangindent\dimen0\hangafter1\ignorespaces#3\par
\endgroup}

\newcommand{\Prob}[2][]{\+{P}^{#1}_{#2}}

\newcommand{\dist}{\mathcal{D}}
\newcommand{\supp}{{\sf supp}}
\newcommand{\cParity}[1]{{#1}\text{-}\mathtt{Parity}}
\newcommand{\Parity}[1]{\mathtt{Parity}(#1)}

\newcommand{\reach}[1]{\mathtt{Reach}(#1)}

\newcommand{\safety}[1]{\mathtt{Safety}(#1)}
\newcommand{\always}{{\sf G}}
\newcommand{\eventually}{{\sf F}}

\renewcommand{\next}{{\sf X}}

\newcommand{\hide}[1]{}
\newcommand{\ignore}[1]{}

\newcommand{\lrc}[1]{(#1)}
\newcommand{\states}{S}
\newcommand{\state}{s}
\newcommand{\transition}{{\longrightarrow}}
\newcommand{\probp}{P}
\newcommand{\mdp}{{\mathcal M}}

\newcommand{\tuple}[1]{\lrc{#1}}
\newcommand{\mdptuple}{\tuple{\states,\zstates,\rstates,\transition,\probp}}
\newcommand{\play}{\rho}
\newcommand{\partialplay}{\rho}
\newcommand{\zstrat}{\sigma}

\newcommand{\zstratset}{\Sigma}
\newcommand{\probm}{{\+{P}}}

\newcommand{\expectval}{\operatorname*{\+E}}
\newcommand{\expectation}[1][]{ \operatorname*{\+E}_{#1}}

\newcommand{\valueof}[2]{{\mathtt{val}_{#1}(#2)}}

\newcommand{\formula}{{\varphi}}

\newcommand{\lrd}[1]{\{#1\}}
\newcommand{\nat}{\mathbb N}
\newcommand{\setcomp}[2]{\lrd{{#1} \mid {#2}}}
\newcommand{\bigdenotationof}[2]{\Big\llbracket #1\Big\rrbracket^{#2}}
\newcommand{\denotationof}[2]{\llbracket #1\rrbracket^{#2}}
\newcommand{\stateset}{Q}
\newcommand{\playset}{{\?R}}
\newcommand{\updatefun}{u}
\newcommand{\memory}{{\sf M}}
\newcommand{\memconf}{{\sf m}}
\newcommand{\colorset}[3]{[#1]^{\coloring#2#3}}
\newcommand{\coloring}{{\mathit{C}ol}}
\newcommand{\colorof}[1]{\coloring\lrc{{#1}}}
\newcommand{\cset}{{\mathcal C}}
\newcommand{\constraint}{\rhd}
\newcommand{\safe}[1]{{\it Safe}(#1)}
\newcommand{\safesub}[2]{{\it Safe_{#1}}(#2)}
\mathchardef\mhyphen="2D %
\newcommand{\optav}{\zstrat_{\mathit{opt\mhyphen av}}}
\newcommand{\epsoptav}{\zstrat_{\eps}}
\newcommand{\M}{\mathcal{M}}

\newcommand{\even}{{\mathit{even}}}
\newcommand{\core}{{\mathit{core}}}
\newcommand{\fix}{{\mathit{fix}}}
\newcommand{\Fix}[1]{\mathsf{Fix}_{#1}}

\newcommand{\ini}{R}

\newcommand{\closure}[1]{\mathit{Cl}(#1)}
\newcommand{\fixin}[2]{{#1}[#2]}
\newcommand{\bubblearound}[3]{\mathsf{bubble}(#1,#2,#3)} %
\newcommand{\fixinbubble}[4]{{#1}[#2,#3,#4]} %
\newcommand{\obs}[1]{\tau_{#1}}  %
\newcommand{\ors}[1]{\rho_{#1}}  %
\newcommand{\err}[2]{{\eps_{#2}}}  %
\newcommand{\previ}{{i}}
\newcommand{\nexti}{{i+1}}
\newcommand{\Alphaset}[1]{\mathsf{ALPHA}_{#1}}
\newcommand{\Betaset}[1]{\mathsf{BETA}_{#1}}
\newcommand{\Gammaset}[1]{\mathsf{GAMMA}_{#1}}
\newcommand{\Fixx}[1]{\mathsf{FIX}_{#1}}

\newcommand{\chain}{{\mathcal C}}
\newcommand{\bubble}[2]{{\sf bubble}_{#1}(#2)}
\newcommand{\F}{{\mathcal F}}
\newcommand{\E}{{\mathcal E}}
\newcommand{\indicatorfun}{{\boldsymbol{1}}}

\newcommand{\poststar}[1]{\mathit{Post}^*(#1)}
\newcommand{\ferr}[1]{L(#1)}
\newcommand{\RVNotInT}[1]{[\next^{#1} \neg \reachset]}

\newcommand{\pmdp}{\mdp_{*}}
\newcommand{\pstates}{\states_{*}}
\newcommand{\pzstates}{\states_{*\zsymbol}}
\newcommand{\prstates}{\states_{*\rsymbol}}
\newcommand{\ptransition}{\transition_{*}}
\newcommand{\pprobp}{\probp_{*}}

\begin{document}

\maketitle

\begin{abstract}
We study countably infinite MDPs with parity objectives.
Unlike in finite MDPs, optimal strategies need not exist,
and may
require infinite memory if they do.
We provide a complete picture of the exact strategy complexity
of $\eps$-optimal strategies (and optimal strategies, where they exist)
for all subclasses of parity objectives in the Mostowski hierarchy.
Either MD-strategies, Markov strategies, or 1-bit Markov strategies
are necessary and sufficient, depending on the number of colors,
the branching degree of the MDP, and whether one considers
$\eps$-optimal or optimal strategies.
In particular, 1-bit Markov strategies
are necessary and sufficient for $\eps$-optimal (resp.~optimal) strategies
for general parity objectives.

\end{abstract}

\section{Introduction}\label{sec:intro}
{\bf\noindent Background.}
Markov decision processes (MDPs) are a standard model for dynamic systems that
exhibit both stochastic and controlled behavior \cite{Puterman:book}.
MDPs play a prominent role in numerous domains, including artificial intelligence and machine learning~\cite{sutton2018reinforcement,sigaud2013markov}, control theory~\cite{blondel2000survey,NIPS2004_2569}, operations research and finance~\cite{bauerle2011finance,schal2002markov}, and formal verification~\cite{ModCheckHB18,ModCheckPrinciples08}.

An MDP is a directed graph where states are either random or controlled.
Its observed behavior is described by runs, which are infinite paths that are, in part, determined by the choices of a controller.
If the current state is random then the next state is chosen according to a fixed probability distribution.
Otherwise, if the current state is controlled, the controller can choose a distribution over all possible successor states.
By fixing a strategy for the controller (and initial state), one obtains a probability space
of runs of the MDP. The goal of the controller is to optimize the expected value of
some objective function on the runs.

The type of strategy necessary to achieve an optimal (resp.\ $\eps$-optimal)
value for a given objective is called its \emph{strategy complexity}.
There are different types of strategies, depending on whether one can take
the whole history of the run into account (history-dependent; (H)),
or whether one is limited to a finite amount of memory (finite memory; (F))
or whether decisions are based only on the current state (memoryless; (M)).
Moreover, the strategy type depends on whether the controller can
randomize (R) or is limited to deterministic choices~(D).
The simplest type, MD, refers to memoryless deterministic strategies.
\emph{Markov strategies} are strategies that base their decisions only on
the current state and the number of steps in the history of the run.
Thus they do use infinite memory, but only in a very restricted form
by maintaining an unbounded step-counter.
Slightly more general are \emph{1-bit Markov strategies} that use 1 bit of
extra memory in addition to a step-counter. 

\medskip
{\bf\noindent Parity objectives.}
We study countably infinite MDPs with parity objectives.
Parity conditions are widely used in temporal
logic and formal verification, e.g., they can express $\omega$-regular
languages and modal $\mu$-calculus \cite{GTW:2002}.
Every state has a \emph{color}, out of a finite set of colors encoded as natural
numbers. A run is winning iff the highest color that is seen
infinitely often is even. The controller wants to maximize 
the probability of winning runs.
The Mostowski hierarchy \cite{Mostowski:84} is a classification of
parity conditions based on restricting the set of allowed colors.
For instance, $\cParity{\{1,2,3\}}$ objectives only use colors $1$, $2$, and $3$.
This includes B\"uchi ($\cParity{\{1,2\}}$) and co-B\"uchi objectives ($\cParity{\{0,1\}}$),
both of which further subsume reachability and safety objectives.

\begin{figure*}[t]
    \centering
    \begin{subfigure}[t]{0.48\textwidth}
        \centering
        \centering
\tikzstyle{dashdotted}=[dash pattern=on 3pt off 2pt on \the\pgflinewidth off 2pt]
\tikzstyle{densely dashdotted}=      [dash pattern=on 3pt off 1pt on \the\pgflinewidth off 1pt]
\tikzstyle{loosely dashdotted}=      [dash pattern=on 3pt off 4pt on \the\pgflinewidth off 4pt]
\begin{tikzpicture}[>=latex',shorten >=1pt,node distance=1.9cm,on grid,auto]

\draw[thick,blue!70!white, pattern=north west lines, pattern color=blue!40!white, fill opacity=.5] (-.7,1.25)  -- (2.5,.4)--(2.5,-1.5) --(-.7,-1.5)--(-.7,1.25);

\draw[thick,red!70!white, pattern=north east lines, pattern
color=red!40!white,fill opacity=0.5]  (5.7,4)-- (5.7,-.4) -- (-.7,1.35) --(-.7,4);

\draw[thick,green, pattern=north east lines, pattern color=green, fill
opacity=.5] (2.6,.35)  -- (5.7,-.5)--(5.7,-1.5) --(2.6,-1.5) -- (2.6,.35);

\node (safe)  at(.8,-1) [draw=none] {$\mathtt{Safety}$};
\node (reach) at(4.2,-1)[draw=none] {$\mathtt{Reach}$};
\node (s01)   at(.8,.5) [draw=none] {$\cParity{\{0,1\}}$};
\node (s12)   at(4.2,.5)[draw=none] {$\cParity{\{1,2\}}$};
\node (s012)  at(.8,2) [draw=none] {$\cParity{\{0,1,2\}}$};
\node (s123)  at(4.2,2) [draw=none] {$\cParity{\{1,2,3\}}$};
\node (s012)  at(.8,3.5) [draw=none] {$\cParity{\{0,1,2,3\}}$};
\node (s123)  at(4.2,3.5) [draw=none] {$\cParity{\{1,2,3,4\}}$};

\node[label=below:\rotatebox{-15}{\textbf{\textcolor{red}{1-bit Markov}}}] at (2.4,1.8) {};
\node[label=below:\rotatebox{-15}{\textbf{\textcolor{blue}{Markov}}}] at (1.1,.4) {};
\node[label=below:\rotatebox{-15}{\textbf{\textcolor{green!80!black}{MD}}}] at (4.3,0) {};

\end{tikzpicture}
        \caption{$\eps$-optimal strategies for infinitely branching MDPs.}
    \end{subfigure}%
    \;\; 
    \begin{subfigure}[t]{0.48\textwidth}
        \centering
        \centering
\tikzstyle{dashdotted}=[dash pattern=on 3pt off 2pt on \the\pgflinewidth off 2pt]
\tikzstyle{densely dashdotted}=      [dash pattern=on 3pt off 1pt on \the\pgflinewidth off 1pt]
\tikzstyle{loosely dashdotted}=      [dash pattern=on 3pt off 4pt on \the\pgflinewidth off 4pt]
\begin{tikzpicture}[>=latex',shorten >=1pt,node distance=1.9cm,on grid,auto]

\draw[thick,green, pattern=north east lines, pattern color=green, fill
opacity=.5] (-.7,.1)  -- (2.5,-.8)--(2.5,1.9) --(5.7,0.8) -- (5.7, -1.6) -- (-.7,-1.6)--(-.7,.1);
\draw[thick,blue!70!white, pattern=north west lines, pattern color=blue!40!white,fill opacity=0.5] (-.7,3)  -- (2.4,1.95)--(2.4,-.7) --(-.7,.2)--(-.7,3);

\draw[thick,red!70!white, pattern=north east lines, pattern color=red!40!white,fill opacity=0.5] (5.7,4) -- (5.7,0.9)--(-.7,3.1)  --(-.7,4);

\node (safe)  at(.8,-1) [draw=none] {$\mathtt{Safety}$};
\node (reach) at(4.2,-1)[draw=none] {$\mathtt{Reach}$};
\node (s01)   at(.8,.5) [draw=none] {$\cParity{\{0,1\}}$};
\node (s12)   at(4.2,.5)[draw=none] {$\cParity{\{1,2\}}$};
\node (s012)  at(.8,2) [draw=none] {$\cParity{\{0,1,2\}}$};
\node (s123)  at(4.2,2) [draw=none] {$\cParity{\{1,2,3\}}$};
\node (s012)  at(.8,3.5) [draw=none] {$\cParity{\{0,1,2,3\}}$};
\node (s123)  at(4.2,3.5) [draw=none] {$\cParity{\{1,2,3,4\}}$};

\node[label=below:\rotatebox{-15}{\textbf{\textcolor{blue!90!black}{Markov}}}] at (1.1,1.9) {};
\node[label=below:\rotatebox{-15}{\textbf{\textcolor{green!80!black}{MD}}}] at (3.7,1.5) {};
\node[label=below:\rotatebox{-15}{\textbf{\textcolor{red}{1-bit Markov}}}] at (3.2,3.5) {};

\end{tikzpicture}
        \caption{Optimal strategies for infinitely branching MDPs.}
    \end{subfigure}
    \caption{These diagrams show the strategy complexity of
      $\eps$-optimal strategies and optimal strategies
      (where they exist) for parity objectives.
      Depending on the position in the
      Mostowski hierarchy, either MD-strategies (green), deterministic Markov-strategies (blue)
      or deterministic 1-bit Markov strategies (red) are necessary and
      sufficient (and randomization does not help \cite{KMST:ICALP2019}). 
      If the MDPs are finitely branching then the Markov strategies can
      be replaced by MD-strategies (i.e., the blue parts turn
      green), but the deterministic 1-bit Markov part (red) remains unchanged.
    }
    \label{fig:overview}
\end{figure*}
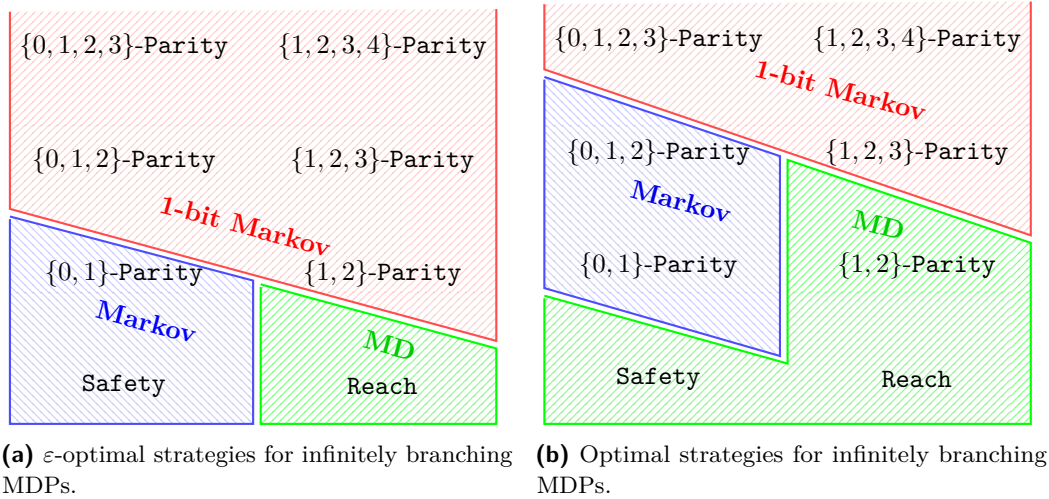

\medskip
{\bf\noindent Related work.}
In \emph{finite} MDPs, there always exist optimal MD-strategies for parity
objectives. In fact, this holds even for finite turn-based 2-player stochastic parity games
\cite{Chatterjee:2004:QSP:982792.982808,Zielonka04}.
Similarly, there always exist optimal MD-strategies in countably infinite
\emph{non-stochastic} turn-based 2-player parity games \cite{Zielonka:1998}.

The picture is more complex for countably infinite MDPs.
Optimal strategies need not exist (not even for reachability objectives
\cite{Puterman:book,Ornstein:AMS1969}),
and $\eps$-optimal strategies for B\"uchi objectives \cite{Hill:79}
and optimal strategies for parity objectives \cite{KMSW2017} 
require infinite memory.

The paper \cite{KMSW2017} gave a complete classification whether MD-strategies
suffice or whether infinite memory is required for
$\eps$-optimal (resp.\ optimal) strategies for all subclasses of parity objectives in the
Mostowski-hierarchy.

However, the mere fact that infinite memory is required for (a subclass of)
parity does not establish the precise strategy complexity.
E.g., are Markov strategies (or Markov strategies with finite extra memory)
sufficient?

In \cite{KMST:ICALP2019} we showed that deterministic 1-bit Markov strategies
are both necessary and sufficient for $\eps$-optimal strategies for B\"uchi
objectives.
I.e., deterministic 1-bit Markov strategies are sufficient, but neither
randomized Markov strategies nor randomized finite-memory strategies are
sufficient.
This solved a 40-year old problem in gambling theory from \cite{Hill:79,Hill:99}. 
The same paper \cite{KMST:ICALP2019} showed that even for finitely branching MDPs with $\cParity{\{1,2,3\}}$ objectives,
optimal strategies (where they exist) need to be \emph{at least} deterministic
1-bit Markov in general, i.e., neither randomized Markov nor randomized finite-memory strategies are sufficient.

While the lower bounds for $\eps$-optimal strategies for B\"uchi objectives
(resp.\ for optimal strategies for $\cParity{\{1,2,3\}}$ objectives)
carry over to general parity objectives, the upper bounds on the strategy complexity of
$\eps$-optimal (resp.\ optimal) parity remained open.

{\bf\noindent A basic upper bound and related conjecture.}
A basic upper bound on the complexity of $\eps$-optimal strategies for parity
can be obtained by using a combination of the results of
\cite{KMST:ICALP2019} on B\"uchi objectives (1-bit Markov) and 
L\'evy's zero-one law as follows.
(However, note that the following argument does not work directly for optimal strategies.) 

Informally speaking, L\'evy's zero-one law implies that, for a tail objective
(like parity) and any strategy, the level of attainment from the current state
almost surely converges to either zero or one. I.e., the runs that always stay
in states where the strategy attains something in $(0,1)$ is a null-set
(cf.\ \cref{app:levy01}).
A consequence for parity is that 
almost all winning runs must eventually, with ever higher
probability, commit to winning by some particular color.
Thus, with minimal losses (e.g., $\eps/2$),
after a sufficiently long finite prefix (depending on~$\eps$),
one can switch to a strategy that aims to visit some \emph{particular} color $x$
infinitely often.
The latter objective is like a B\"uchi objective where the states of color~$x$ are accepting and states of color~$>x$ are considered losing sinks.
By \cite{KMST:ICALP2019}, an $\eps/2$-optimal strategy for such a B\"uchi
objective can be chosen 1-bit Markov. However, one would also need to remember
which color~$x$ one is supposed to win by and \emph{stick to that color}.
The latter is critical, since strategies that switch focus between winning
colors infinitely often
(e.g., if they follow some local criteria based on the value of the current state
wrt.\ various colors) can end up losing.
Overall, the memory needed for such an
$\eps$-optimal strategy for parity is: $\lceil\log_2(c)\rceil$
bits for $c$ even colors to remember which color $x$ one
is supposed to win by and Markov plus 1 bit for the B\"uchi strategy (see
above), where the Markov step-counter also determines whether one still plays in the
prefix. Thus Markov plus $(1+\lceil\log_2(c)\rceil)$ bits are sufficient.
This argument would suggest that more memory is required for more colors.
However, our result shows that this is \emph{not} the case.

\medskip
{\bf\noindent Our contributions.}
We show \emph{tight} upper bounds on the strategy complexity of $\eps$-optimal
(resp.\ optimal) strategies for parity objectives:
They can be chosen as deterministic 1-bit Markov, regardless of the number of
colors.
I.e., we provide matching upper bounds to the lower bounds from \cite{KMST:ICALP2019}.

In \cref{epsParity} we prove \cref{theo-eps-opt-full}. An iterative plastering
construction (i.e., fixing player choices on larger and larger subspaces)
builds an $\eps$-optimal 1-bit Markov strategy where
the probability of never switching between winning even colors is $\geq 1-\eps$.
Its correctness relies heavily on L\'evy's zero-one law.
The number of iterations is finite and proportional to the number of even colors.
It eliminates the need to remember the winning color $x$ and the
$\lceil\log_2(c)\rceil$ part of the memory.

\begin{theorem}\label{theo-eps-opt-full}
Consider an MDP $\mdp$, a parity objective and
a finite set $\states_0$ of initial states.

For every $\eps >0$ there exists a deterministic 1-bit Markov strategy
that is $\eps$-optimal from every state $\state \in \states_0$.
\end{theorem}

In \cref{as-par} we prove \cref{theo:opt-par-main}.
If an optimal strategy exists, then an optimal 1-bit Markov strategy
can be constructed by the so-called
\emph{sea urchin} construction. It is a very complex plastering
construction with infinitely many iterations that uses the results
of \cref{theo-eps-opt-full} and L\'evy's zero-one law as building blocks.
Its name comes from the shape of the subspace
in which player choices get fixed: a growing finite body
(around a start set $\states_0$) with
a finite, but increasing, number of spikes, where each spike is of infinite
size; cf.~\cref{fig:as-par:sea-urchin}.
E.g., if the initial states are almost surely winning
then, at the stage with $i$ spikes, this strategy attains parity with
some probability $\ge 1-2^{-i}$ already \emph{inside} this subspace,
and in the limit of $i \rightarrow \infty$ it attains parity almost surely.
A further step even yields a single deterministic 1-bit Markov strategy
that is optimal from every state that has an optimal strategy. 

\begin{theorem}\label{theo:opt-par-main}
Consider an MDP $\mdp$ with a parity objective and
let $\states_{\mathit{opt}}$ be
the subset of states that have an optimal strategy.

There exists a deterministic 1-bit Markov strategy
that is optimal from every $\state \in \states_{\mathit{opt}}$.
\end{theorem}

In \cref{theo-eps-opt-full} and \cref{theo:opt-par-main} the
initial content of the 1-bit memory is
irrelevant (cf.\ \cref{lem-eps-opt}, \cref{thm:opt-par-acyclic} and
\cref{rem:as-layered-memconf}).

Moreover, we show
in \cref{as012} and \cref{eps01}
that in certain subcases deterministic Markov strategies are
necessary and sufficient
(i.e., these require a Markov step-counter, but not the extra bit):
optimal strategies for co-B\"uchi and $\cParity{\{0,1,2\}}$,
and $\eps$-optimal strategies for safety and co-B\"uchi.
In the special case of finitely branching MDPs, these Markov strategies
(but not the 1-bit Markov strategies) can be replaced by MD-strategies.

Together with the previously established lower bounds,
this yields a complete picture of the \emph{exact} strategy complexity
of parity objectives at all levels of the Mostowski hierarchy, for countable MDPs.
\cref{fig:overview} gives a complete overview.

\section{Preliminaries}\label{sec:prelim}
A \emph{probability distribution} over a countable set $S$ is a function
$f:\states\to[0,1]$ with $\sum_{\state\in \states}f(\state)=1$.
We write 
$\dist(\states)$ for the set of all probability distributions over $\states$.

We study 
\emph{Markov decision processes}
(MDPs)
over 
countably infinite
state spaces. Formally,
an MDP $\mdp=\mdptuple$ consists of
a countable set~$\states$ of \emph{states}, 
which is partitioned into a set~$\zstates$ of \emph{controlled states} 
and  a set~$\rstates$ of \emph{random states},
a  \emph{transition relation}~$\transition\subseteq\states\x\states$,
and a  \emph{probability function}~$\probp:\rstates \to \dist(\states)$. 
We  write $\state\transition{}\state'$ if $\tuple{\state,\state'}\in \transition$,
and  refer to~$s'$ as a \emph{successor} of~$s$. 
We assume that every state has at least one successor.  
The probability function~$P$  assigns to each random state~$\state\in \rstates$
a probability distribution~$P(\state)$ over its set of successors.
A \emph{sink} is a subset $T \subseteq \states$ closed under the $\transition$ relation.
An MDP is \emph{acyclic} if the underlying  graph~$(S,\transition)$ is acyclic.
It is  \emph{finitely branching} 
if every state has finitely many successors
and \emph{infinitely branching} otherwise.
An MDP without controlled states
($\zstates=\emptyset$) is
a \emph{Markov chain}.

\medskip
\noindent{\bf Strategies and Probability Measures.}
A \emph{run}~$\play$ is an  infinite sequence $\state_0\state_1\cdots$ of states
such that $\state_i\transition{}\state_{i+1}$ for all~$i\in \mathbb{N}$;
write~$\play(i)\eqdef\state_i$ for the $i$-th state along~$\play$.
A \emph{partial run} is a finite prefix of a run.
We say that (partial) run $\play$ \emph{visits} $\state$ if
$\state=\play(i)$ for some $i$, and that~$\play$ \emph{starts in} $s$ if $\state=\play(0)$. 

A \emph{strategy} %
is a function $\zstrat:\states^*\zstates \to \dist(S)$ that 
assigns to partial runs $\partialplay\state \in \states^*\zstates$ 
a distribution over the successors of $\state$. 
A (partial) run~$\state_0\state_1\cdots$ is \emph{induced by} strategy~$\zstrat$
if for all~$i$
either $\state_i \in \zstates$ and $\zstrat(\state_0\state_1\cdots\state_i)(\state_{i+1})>0$,
or
$\state_i \in \rstates$ and $\probp(\state_i)(\state_{i+1})>0$.

 A strategy~$\zstrat$ and an initial state $\state_0\in \states$
induce a standard probability measure on sets of infinite plays. We write $\probm_{\mdp,\state_0,\zstrat}({\playset})$ for the probability of a 
measurable set $\playset \subseteq \state_0 \states^\omega$ of runs starting from~$\state_0$.
As usual, it is  first defined on the \emph{cylinders} $s_0 s_1 \ldots s_n \states^\omega$, where $s_1, \ldots, s_n \in \states$:
if $s_0 s_1 \ldots s_n$ is not a partial run induced by~$\zstrat$ then
$\probm_{\mdp,\state_0,\zstrat}(s_0 s_1 \ldots s_n \states^\omega) \eqdef 0$.
Otherwise, $\probm_{\mdp,\state_0,\zstrat}(s_0 s_1 \ldots s_n \states^\omega) \eqdef \prod_{i=0}^{n-1} \bar{\zstrat}(s_0 s_1 \ldots s_i)(s_{i+1})$, where $\bar{\zstrat}$ is the map that extends~$\zstrat$ by $\bar{\zstrat}(w s) = \probp(s)$ for all $w s \in \states^* \rstates$.
By Carath\'eodory's theorem~\cite{billingsley-1995-probability}, 
this extends uniquely to a probability measure~$\probm_{\mdp,\state_0,\zstrat}$ on 
 measurable subsets of $s_0S^{\omega}$.
We will write $\expectval_{\mdp,\state_0,\zstrat}$
for the expectation w.r.t.~$\probm_{\mdp,\state_0,\zstrat}$.
We may drop the subscripts from  notations, if it is understood.

\medskip
\noindent {\bf Objectives.} 
The objective of the player is determined by a predicate on infinite plays.
We assume familiarity with the syntax and semantics of the temporal
logic LTL \cite{CGP:book}.
Formulas are interpreted on the structure $(\states,\transition)$.
We use 
$\denotationof{\formula}{\state} \subseteq \state \states^\omega$ to denote the set of runs starting from
$\state$ that satisfy the LTL formula $\formula$,
which is a measurable set \cite{Vardi:probabilistic}.
We also write $\denotationof{\formula}{}$ for $\bigcup_{s \in S} \denotationof{\formula}{s}$.
Where it does not cause confusion we will
identify $\varphi$ and $\denotationof{\formula}{}$
and just write
$\probm_{\mdp,\state,\zstrat}(\formula)$
instead of 
$\probm_{\mdp,\state,\zstrat}(\denotationof\formula\state)$.

Given a  set $\reachset \subseteq \states$   of  states, 
  the \emph{reachability} objective~$\reach{\reachset}$ is the set of  runs that visit  $\reachset$ at least once; and   the \emph{safety objective}~$\safety{\reachset}$ is the set of  runs that never visit~$T$.

Let $\cset \subseteq \nat$ be a finite set of colors.
A \emph{color function} $\coloring:\states\to \cset$ assigns to each state~$\state$ its color~$\colorof\state$. The parity objective, written as
$\Parity{\coloring}$, is the set of infinite runs such that the largest
color that occurs infinitely often along the run is even.
To define this formally, let
$\even(\cset)=\{i\in \cset\mid i\equiv 0\mod{2}\}$.
For 
$\mathord{\constraint}\in\{\mathord{<},\mathord{\le},\mathord{=},\mathord{\geq},\mathord{>}\}$,
$n\in\nat$, 
and $\stateset\subseteq\states$, let
$\colorset \stateset\constraint n \eqdef \setcomp{\state\in \stateset}{\colorof\state\constraint n}$
be the set of states in $\stateset$ with color $\constraint n$.
Then %
\[
 \Parity{\coloring} \eqdef
     \bigvee_{i\in \even(\cset)}\left(\always\eventually \colorset{\states}{= i}{} \wedge
 \eventually\always \colorset{\states}{\leq  i}{}\right).
\]

The Mostowski hierarchy \cite{Mostowski:84} classifies parity objectives
by restricting the range of~$\coloring$ to a set of colors $\cset \subseteq \nat$.
We write $\cParity{\cset}$ for such restricted parity objectives.
In particular, the classical B\"uchi and co\nobreakdash-B\"uchi  objectives  correspond to
$\cParity{\{1,2\}}$ and $\cParity{\{0,1\}}$, respectively.
These two classes are incomparable but both subsume the  reachability and safety objectives.
Assuming that $\reachset$ is a sink, $\reach{\reachset}=\Parity{\coloring}$
for the coloring with $\coloring(s) = 1 \iff s\notin\reachset$
and $\safety{\reachset}=\Parity{\coloring}$
for the coloring with $\coloring(s) = 1 \iff s\in\reachset$.
Similarly, $\cParity{\{0,1,2\}}$ and $\cParity{\{1,2,3\}}$ are incomparable,
but they both subsume (modulo renaming of colors) B\"uchi and co-B\"uchi objectives.

{\renewcommand{\denotationof}[1]{#1}
An objective $\formula$ is called a \emph{tail objective}
(resp.\ \emph{suffix-closed}) iff for every run
$\rho'\rho$ with some finite prefix $\rho'$ we have
$\rho'\rho \in \denotationof{\formula}{} \Leftrightarrow \rho \in \denotationof{\formula}{}$
(resp.\ $\rho'\rho \in \denotationof{\formula}{} \Rightarrow \rho \in \denotationof{\formula}{}$).
In particular, $\Parity{\coloring}$ is tail for every coloring $\coloring$.
Moreover, if $\formula$ is suffix-closed then $\eventually\formula$ is tail.
}

\medskip
\noindent{\bf Strategy Classes.}
Strategies  $\zstrat:\states^*\zstates \to \dist(S)$ are in general  \emph{randomized} (R) in the sense that they take values in $\dist(\states)$. 
A strategy~$\zstrat$ is \emph{deterministic} (D) if $\zstrat(\rho)$ is a Dirac distribution 
for all partial runs~$\rho\in \states^{*} \zstates$.

We formalize the amount of \emph{memory} needed to implement strategies.
Let $\memory$ be a countable set of memory modes. An \emph{update function}
is a function $\updatefun: \memory\times \states \to \dist(\memory\times \states)$
that meets the following two conditions, for all modes $\memconf \in \memory$:
\begin{itemize}
	\item for all controlled states~$\state\in \zstates$, 
            the distribution
            $\updatefun((\memconf,\state))$ is over 
	$\memory \times \{\state'\mid \state \transition{} \state'\}$.
	\item for all random states $\state \in \rstates$, we have that
            $\sum_{\memconf'\in \memory} \updatefun((\memconf,\state))(\memconf',\state')=P(\state)(\state')$.
\end{itemize}

\newcommand{\inducedStrat}[2]{#1[#2]}

An update function~$\updatefun$ together with an initial memory~$\memconf_0$
induce a strategy~$\inducedStrat{\updatefun}{\memconf_0}:\states^*\zstates \to \dist(S)$ as follows.
Consider the Markov chain with states set $\memory \times \states$,
transition relation $(\memory\x\states)^2$ and probability function~$\updatefun$.
Any partial run $\rho=s_0 \cdots s_i$ in $\mdp$ gives rise to a
set 
{$H(\rho)=\{(\memconf_0,s_0) \cdots (\memconf_i,s_i) \mid \memconf_0,\ldots, \memconf_i\in \memory\}$}
of partial runs in this Markov chain.
Each $\rho s \in \state_0 \states^{*} \zstates$
induces a 
probability distribution~$\mu_{\rho \state}\in \dist(\memory)$, 
the  probability of being in state $(\memconf,\state)$
conditioned on having taken some partial run
from~$H(\rho \state)$.
We define~$\inducedStrat{\updatefun}{\memconf_0}$ such that
$\inducedStrat{\updatefun}{\memconf_0}(\rho \state)(\state')\eqdef\sum_{\memconf,\memconf'\in \memory} \mu_{\rho \state}(\memconf) \updatefun((\memconf,\state))(\memconf',\state')$
for all $\rho \state\in \states^{*} \zstates$ and~$\state' \in \states$.

We say that a strategy $\zstrat$ can be \emph{implemented} with
memory~$\memory$ (and initial memory $\memconf_0$) if there exists %
an update function $\updatefun$ such that $\zstrat=\inducedStrat{\updatefun}{\memconf_0}$.
In this case we may also write $\inducedStrat{\zstrat}{\memconf_0}$ to explicitly specify the initial
memory mode $\memconf_0$. Based on this, we can define several classes of strategies:
\begin{itemize}
\item A strategy $\zstrat$ is \emph{memoryless}~(M) (also called \emph{positional}) 
if it can be implemented with a memory  of size~$1$. 
We may view
M-strategies as functions $\zstrat: \zstates \to \dist(\states)$.

\item %
A strategy~$\zstrat$ is \emph{finite memory}~(F) if 
there exists a finite memory~$\memory$ implementing~$\zstrat$.
More specifically, a strategy is \emph{$k$-bit} 
if it can be implemented with a memory of size~$2^k$.
Such a strategy is then determined by a function
$\updatefun:\{0,1\}^k\x \states \to \dist(\{0,1\}^k \x \states)$.

\item A strategy~$\zstrat$ is \emph{Markov} if 
 it can be implemented with the natural numbers $\memory=\mathbb{N}$
as the memory, initial memory mode $\memconf_0=0$
and a function $\updatefun$ such that the distribution
$\updatefun(\memconf,\state)$ is over $\{\memconf+1\}\times \states$
for all $\memconf\in \memory$ and $\state\in \states$.
Intuitively, such a  strategy depends only on 
the current state and the number of steps taken so far.

\item A strategy~$\zstrat$ is \emph{k-bit Markov} if
it can be implemented with memory $\memory = \N \times \{0,1\}^k$,
$\memconf_0 \in \{0\}\x\{0,1\}^k$ 
and a function $\updatefun$ such that the distribution
$\updatefun((n,b,\state))$ is over $\{n+1\}\times \{0,1\}^k\times \states$
for all $(n,b)\in \memory$ and $\state\in \states$.
\end{itemize}

\emph{Deterministic 1-bit} strategies are central in this paper; by this we mean strategies that are both deterministic and 1-bit.

\medskip
\noindent{\bf Optimal and $\eps$-optimal Strategies.}
Given an objective~$\formula$, the value of state~$s$ in an MDP~$\mdp$, denoted by 
$\valueof{\mdp}{s}$, is the supremum probability of achieving~$\formula$.
 Formally, we have  $\valueof{\mdp}{s} \eqdef\sup_{\sigma \in \Sigma} \probm_{\mdp,\state,\zstrat}(\formula)$ where $\Sigma$ is the set of all strategies.
For $\eps\ge 0$ and state~$s\in\states$, we say that a strategy is $\eps$-optimal from $s$
iff $\probm_{\mdp,\state,\zstrat}(\formula) \geq \valueof{\mdp}{s} -\eps$.
A $0$-optimal strategy is called optimal. 
An optimal strategy is almost-surely winning if $\valueof{\mdp}{s} = 1$. 
 
Considering an MD strategy as a function $\zstrat: \zstates \to \states$ and $\eps\ge 0$, $\zstrat$ is \emph{uniformly} $\eps$-optimal  (resp.~uniformly optimal) if it is $\eps$-optimal (resp.~optimal) from every $s\in S$.

\medskip
\noindent{\bf Fixing and Safe Sets.} Let  $\sigma$ 
be  an MD strategy. 
Given a set $S'\subseteq S$ of states, 
   write $\fixin{\mdp}{\sigma,S'}$ for the MDP
    obtained from $\mdp$ by fixing the strategy $\sigma$ for all states in $S'$, that is,  
        $\fixin{\mdp}{\sigma,S'}\eqdef
        (\states,\zstates\setminus S',\rstates\cup S',\transition{},\probp')$
        where $\probp'(s)\eqdef\sigma(s)$ for all $s\in S'$.
   
For an objective $\formula$ and 
 a threshold $\beta \in [0,1]$,
denote by $\safesub{\mdp,\zstrat,\formula}{\beta}$ the set of all states~$s$ starting from which $\zstrat$ attains at least probability~$\beta$; and denote by 
$\safesub{\mdp,\formula}{\beta}$ 
the set of states whose value for $\formula$ is at least $\beta$. Formally, 
\begin{equation}
    \label{def:safeset}
    \begin{aligned}[c]
        \safesub{\mdp,\zstrat,\formula}{\beta}&\eqdef\{\state \in \states \mid \probm_{\mdp,\state,\zstrat}(\formula) \ge \beta\}
    ,&& %
        \safesub{\mdp,\formula}{\beta}&\eqdef\{\state \in \states \mid \valueof{\mdp,\formula}{\state} \ge \beta\}.
    \end{aligned}
\end{equation}

\section{$\eps$-Optimal Strategies for Parity}\label{epsParity}
In this section we prove \cref{theo-eps-opt-full}, stating that $\eps$-optimal strategies for parity objectives can be chosen 1-bit Markov. 
Given an MDP  we convert it by three successive  reductions to a structurally simpler MDP where   strategies require  less sophistication to achieve  parity.

\subparagraph*{First reduction (Finitely Branching).} 
This reduction converts an infinitely branching MDP~$\mdp$ to a finitely branching one~$\mdp'$, with a clear
 bijection between the strategies in $\mdp$ and $\mdp'$.
 The construction,  first presented in our previous work~\cite{KMST:ICALP2019}, replaces each controlled state~$\state$, that has infinitely many  successors $(s_i)_{i\in \mathbb{N}}$, with a ``ladder'' of controlled states $(q_i)_{i\in \mathbb{N}}$, where each $q_i$ has only two successors: $q_{i+1}$ and $s_i$.
 Roughly speaking, the controller choice of
 successor $s_n$ at $s$ in $\mdp$, 
 is simulated by a series of choices  $q_{i+1}$ at $q_{i}$, $0\leq i <n$, followed by a choice of successor~$s_n$ in state~$q_n$ in~$\M'$, and vice versa.
  
  To prevent  scenarios when the controller in~$\M'$  stays on a ladder  and
  never commits to a decision, we assign color~$1$
  to all states~$(q_i)_{i\geq 1}$ on  the ladder
  ($q_0$ inherits the color of~$s$). Hence, a hesitant
  run on the ladder is losing for  parity.
  So w.l.o.g.\ we can assume that the given $\mdp$ is finitely branching.

\begin{restatable}{lemma}{infbranchingtofinitebranching}
\label{lem:inf-branching-to-finite-branching}
\
\begin{enumerate}
\item
Suppose that for every finitely branching acyclic MDP with
a finite set $\states_0$ of initial states, 
and a parity objective,
there exist $\eps$-optimal deterministic $1$-bit strategies
from $\states_0$. 

Then even for every infinitely branching acyclic MDP 
with a finite set $\states_0$ of initial states and 
a parity objective,
there exist $\eps$-optimal deterministic $1$-bit strategies
from $\states_0$.
\item
Suppose that for every finitely branching acyclic MDP with
a parity objective,
there exists a deterministic $1$-bit strategy that is optimal
from all states that have an optimal strategy.

Then even for every infinitely branching acyclic MDP
with a parity objective,
there exists a deterministic $1$-bit strategy that is optimal
from all states that have an optimal strategy.
\end{enumerate}
\end{restatable}

\subparagraph*{Second reduction (Acyclicity).}
A deterministic 1-bit Markov strategy can be seen as a function~$\zstrat: \mathbb{N}\times \{0,1\} \times \states \to \{0,1\} \times \states$, where $\sigma$
has access to an internal bit $b\in \{0,1\}$,
which can be updated freely, and a step counter $k\in \mathbb{N}$, which
increments by one in each step. Having $b$ and $k$, $\sigma$ produces a
decision based on the current state of the MDP.

Following \cite{KMST:ICALP2019}, we encode the step-counter from  strategies
into MDPs s.t.\ the current state of the system uniquely determines
the length of the path taken so far.
This translation
allows us to focus on acyclic MDPs.

\begin{restatable}{lemma}{acyclicMarkov}
\label{lem:acyclic-Markov}
Consider MDPs with a parity objective and $k \in \nat$.
\begin{enumerate}
\item
Suppose that for every acyclic MDP $\mdp'$ and every \emph{finite} set of initial
states $\states_0'$ and $\eps > 0$, there exists a deterministic $k$-bit strategy
that is $\eps$-optimal from all states $\state \in \states_0'$.

Then for every MDP $\mdp$ and every \emph{finite} set of initial states $\states_0$ and $\eps >0$,
there exists a deterministic $k$-bit Markov strategy that is
$\eps$-optimal from all states $\state \in \states_0$.
\item
  Suppose that for every acyclic MDP $\mdp'$ and $\eps > 0$,
  there exists a deterministic $k$-bit strategy
that is $\eps$-optimal from all states.
Then for every MDP $\mdp$ and $\eps >0$,
there exists a deterministic $k$-bit Markov strategy that is
$\eps$-optimal from all states.
\item
Suppose that for every acyclic MDP $\mdp'$, where
$\states_{\mathit{opt}}'$ is the subset of states that have an optimal
strategy, 
there exists a deterministic $k$-bit strategy
that is optimal from all states $\state \in \states_{\mathit{opt}}'$.
Then for every MDP $\mdp$, where
$\states_{\mathit{opt}}$ is the subset of states that have an optimal
strategy, 
there exists a deterministic $k$-bit Markov strategy
that is optimal from all states $\state \in \states_{\mathit{opt}}$.
\end{enumerate}
\end{restatable}

By \cref{lem:acyclic-Markov}, 
the sufficiency of deterministic 1-bit strategies in acyclic MDPs  implies the sufficiency of 
deterministic 1-bit Markov strategies in general MDPs. Thus to prove \cref{theo-eps-opt-full}, it suffices to prove the
following:

\begin{theorem}\label{theo-eps-opt}
Consider an acyclic MDP $\mdp$, a parity objective and
a finite set $\states_0$ of  states.
For every $\eps >0$ there exists a deterministic 1-bit  strategy
that is $\eps$-optimal from every~$\state \in \states_0$.
\end{theorem}
\subparagraph*{Third reduction (Layered MDP).} 
This reduction is in the same spirit of the previous one, in which the bit $b\in \{0,1\}$ is transferred from strategies to MDPs. Given an MDP~$\M$, the corresponding \emph{layered} MDP~$\?L(\mdp)$ has two copies of each state~$s\in S$ and each transition~$t\in\transition_1$ of~$\mdp$, one augmented with bit~$0$ and another with bit~$1$: $(s,i)$ and $(t,j)$ with $i,j\in \{0,1\}$.   
The states~$(s,i)$ are random  if $s\in \states_\rsymbol$ and controlled if $s\in \states_\zsymbol$ . All the~$(t,j)$  are controlled.
If there is a transition $t=(a,b)$ from  state $a$ to $b$ in~$\mdp$, there will be 
two transitions from $(a,i)$ to $(t,i)$, and four transitions from $(t,i)$ to $(b,j)$ in $\?L(\mdp)$; see Figure~\ref{fig:layeredsys}. 
 
 A 1-bit deterministic strategy in~$\mdp$ 
 at a state~$a$ picks a single   successor~$b$ and may flip the bit from~$i$ to~$j$; this is simulated in $\?L(\mdp)$ with an MD strategy~$\sigma$  within two consecutive steps: $\sigma$ first chooses the transition~$t=(a,b)$ by 
 $\sigma(a,i)=(t,i)$ 
 and then updates the bit by $\sigma(t,i)=(b,j)$ thereby moving from layer~$i$ to layer~$j$. 
  The controlled states~$(t,i)$ are essential  for a correct simulation, since otherwise
    the controller cannot freely flip the bit (switch between layers) after it observes the successor chosen randomly  at a random state.

\begin{definition}[Layered MDP] 
\label{def:layered}
  Given an MDP $\mdp=(\states,\states_\zsymbol,\states_\rsymbol, \transition_1, \probp_1)$ with coloring $\coloring_1:\states\to \cset$,
  we define the corresponding
  \emph{layered MDP} $\?L(\mdp)=(L,L_\zsymbol,L_\rsymbol, \transition_2, \probp_2)$
  with coloring $\coloring_2:L\to \cset$
  as follows.
 \begin{itemize}
 \item  $L\eqdef (\states\cup \transition_1 )\x\{0,1\}$ where the set of controlled states is $L_\zsymbol \eqdef (\states_\zsymbol\cup \transition_1 )\x\{0,1\}$.
 \item    For all 
    $t\in \transition_1$ such that $t=(\state,\state')$ and for all  $i,j\in \{0,1\}$, we have:
     \begin{enumerate}

    \item $(\state,i)\transition_2 (t,i)$ and $(t,i)\transition_2 (\state',j)$,
      
    \item  $\probp(\state,i)((t,i)) \eqdef
      \probp(\state)(\state')$ iff $\state \in \states_\rsymbol$, and
  \item $\coloring_2((\state,i)) \eqdef \coloring_1(\state)$
          and  $\coloring_2((t,i)) \eqdef \coloring_1(\state')$.
        
        \end{enumerate}
  
  \end{itemize}

\end{definition}

The layered MDP of an acyclic MDP is acyclic.
For $q\in \states\cup \transition_1$, we refer to the copies of~$q$ in layer $0$ and layer $1$ as \emph{siblings}: $(q,0)$ and $(q,1)$. A set $B\subseteq L$ is 
\emph{closed} if for each state~$(q,i)\in B$ its sibling is also in $B$.
 Denote by~$\closure{B}$  the minimal closed  superset of~$B$.

\begin{restatable}{lemma}{proplayered}
\label{prop:layered}
Consider an acyclic MDP~$\mdp=\mdptuple$ with
a parity objective~$\formula=\Parity{\coloring}$ and let $\?L(\mdp)$ be the
corresponding layered MDP.

For every deterministic $1$-bit strategy $\inducedStrat{\updatefun}{\memconf_0}$ in $\mdp$
there is a corresponding MD strategy $\tau$ in~$\?L(\mdp)$, and vice-versa,
such that for every~$\state_0 \in \states$,
$
\probm_{\?L(\mdp),(\state_0,\memconf_0),\tau}(\formula) =
\probm_{\mdp,\state_0,\inducedStrat{\updatefun}{\memconf_0}}(\formula) 
$.
\end{restatable}

\begin{remark}\label{rem:as-layered-memconf}
We note that in a layered system $\?L(\mdp)$, any two siblings have the same value w.r.t.\ a parity objective $\formula$.
Moreover, any state $\state$ in~$\mdp$ has an optimal strategy iff
$(\state,0)\in \?L(\mdp)$ has an optimal strategy
iff its sibling $(\state,1)$ has an optimal strategy.

Suppose $\tau$ is an MD strategy in $\?L(\mdp)$ that is optimal for all states that have an optimal strategy.
Let $\updatefun$ be the update function of a corresponding $1$-bit strategy in $\mdp$, derived as described in \cref{prop:layered}.
Then for every state $\state$ in~$\mdp$ that has an optimal strategy we have
$
\probm_{\mdp,\state,\inducedStrat{\updatefun}{0}}(\formula) =
\probm_{\?L(\mdp),(\state,0),\tau}(\formula) =
\probm_{\?L(\mdp),(\state,1),\tau}(\formula) =
\probm_{\mdp,\state,\inducedStrat{\updatefun}{1}}(\formula)
$.
That is, both $\inducedStrat{\updatefun}{0}$ and $\inducedStrat{\updatefun}{1}$ are optimal from~$\state$, so the initial memory mode is irrelevant.
\qed 
\end{remark}

To prove \cref{theo-eps-opt}, given an 
acyclic MDP, a set of initial states~$S_0$ and $\eps>0$, we consider the 
 layered MDP~$\?L(\mdp)$ and set 
 $L_0=S_0\times \{0\}$ of initial states.
 In the following lemma, we 
 prove that there exists a single  MD strategy
that is  $\eps$-optimal starting
from every state~$\ell_0\in L_0$ in~$\?L(\mdp)$. This and \cref{prop:layered} will directly lead to
 \cref{theo-eps-opt}.

\begin{lemma}\label{lem-eps-opt}
Consider an acyclic MDP~$\mdp$ and parity objective~$\formula=\Parity{\coloring}$. Let~$\?L(\mdp)$ be the layered MDP of~$\mdp$ and $\coloring$.
For all finite sets~$L_0$ of  states in~$\?L(\mdp)$ and all~$\eps>0$
there exists a single MD strategy that is $\eps$-optimal  for~$\formula$
 from every state~$\ell_0\in L_0$.
\end{lemma}

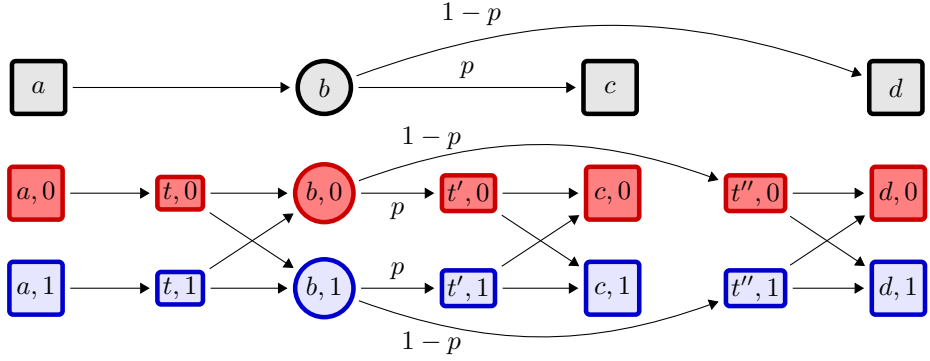
\begin{figure*}[t]
  \centering
  \begin{tikzpicture}[
scale=1.0, every node/.style={transform shape},
node distance=.8cm and 3cm]
\node[cstate] (s00) at (0,0) {$a$};
\node[rstate, right=of s00] (s10) {$b$};
\node[cstate, right=of s10] (s20) {$c$};
\node[cstate, right=of s20] (s30) {$d$};

\draw (s00) edge (s10);
\draw (s10) edge node[above,sloped]{$p$} (s20);
\draw (s10) edge[bend left=20] node[auto,pos=0.3]{$1-p$} (s30);

\end{tikzpicture}
  
  \begin{tikzpicture}[
scale=1, every node/.style={transform shape},
node distance=.5cm and 3cm]
\node[cstate,layer1] (s00) at (0,0) {$a,0$};
\node[rstate,layer1, right=of s00] (s10) {$b,0$};
\node[cstate,layer1, right=3cm of s10] (s20) {$c,0$};
\node[cstate,layer1, right=3cm of s20] (s30) {$d,0$};

\node[cstate,layer2, below=of s00] (s01) {$a,1$};
\node[rstate,layer2, right=of s01] (s11) {$b,1$};
\node[cstate,layer2, right=of s11] (s21) {$c,1$};
\node[cstate,layer2, right=of s21] (s31) {$d,1$};

\node[ucstate,layer1] (u10) at ($(s00)!.5!(s10)$) {$t,0$};
\node[ucstate,layer2] (u11) at ($(s01)!.5!(s11)$) {$t,1$};
\node[ucstate,layer1] (u20) at ($(s10)!.5!(s20)$) {$t',0$};
\node[ucstate,layer2] (u21) at ($(s11)!.5!(s21)$) {$t',1$};
\node[ucstate,layer1] (u30) at ($(s20)!.5!(s30)$) {$t'',0$};
\node[ucstate,layer2] (u31) at ($(s21)!.5!(s31)$) {$t'',1$};



\draw (s10) edge node[below] {$p$} (u20);
\draw (s11) edge node[above] {$p$} (u21);

\draw (s00) edge (u10);
\draw (s01) edge (u11);

\draw (u10) edge  (s10);
\draw (u11) edge  (s11);
\draw (u10) edge  (s11);
\draw (u11) edge  (s10);

\draw (u20) edge  (s20);
\draw (u20) edge  (s21);
\draw (u21) edge  (s20);
\draw (u21) edge  (s21);

\draw (s10) edge[bend left=20] node[above,pos=0.2] {$1-p$} (u30);
\draw (s11) edge[bend right=20] node[below,pos=0.2] {$1-p$} (u31);
\draw (u30) edge  (s30);
\draw (u30) edge  (s31);
\draw (u31) edge  (s30);
\draw (u31) edge  (s31);
\end{tikzpicture}
  \caption{An MDP $\mdp$ (in grey) and the corresponding layered MDP $\?L(\mdp)$ with states of layer~0 and 1 in red and blue, respectively. Here, $t=(a,b)$, $t'=(b,c)$ and $t''=(b,d)$ are   transitions of~$\mdp$.
  }
\label{fig:layeredsys}
\end{figure*}

In the rest of this section, we prove \cref{lem-eps-opt}. We fix 
a  layered MDP $\?L(\mdp)$ (or simply~$\?L$) obtained from a given acyclic and finitely branching MDP~$\mdp$ and a coloring~$\coloring: \states \to \cset$, where 
the set of states is~$L$ and the finite set of initial states is~$L_0\subseteq L$. 
Let $\formula$ be the resulting parity objective in $\?L$.

Recall that $\even(\cset)=2\nat \cap \cset$
denotes the set of even colors.
We denote by $e_{\max}$ the largest even color in $\even(\cset)$
and assume w.l.o.g., that $\even(\cset)$ contains all even numbers from
$2$ to $e_{\max}$ inclusive.
 We have:
\begin{align*}
 \varphi & \eqdef \bigvee_{e\in \even(\cset)}\left(\always\eventually \colorset{L}{= e}{} \wedge
 \eventually\always \colorset{L}{\leq  e}{}\right)&\\
& = \bigvee_{e\in \even(\cset)}\left(\eventually \always\eventually \colorset{L}{= e}{} \wedge
 \eventually\always \colorset{L}{\leq  e}{}\right)&
  \text{since $\always\eventually \colorset{L}{= e}{}$ is a tail objective}\\
&= \bigvee_{e\in \even(\cset)}\eventually \left(\always\eventually \colorset{L}{= e}{} \wedge
\always \colorset{L}{\leq  e}{}\right)& \text{since $\eventually \always
 A \wedge \eventually \always B = \eventually (\always A \wedge \always B)$} \\
&= \bigvee_{e\in \even(\cset)} \eventually \varphi_e\;,& 
\end{align*}
where $\varphi_e \eqdef  \left(\always\eventually \colorset{L}{= e}{} \wedge
\always \colorset{L}{\leq  e}{}\right)$.
Indeed, $\varphi_e$ is the set of runs that win through color~$e$ (i.e., by visiting color~$e$ infinitely often and never visiting larger colors).
Since the $\eventually \varphi_e$ are disjoint,  for all states~$\ell$ and strategies~$\sigma$, we have:
\begin{equation}\label{eq:varphii}
\Prob{\?L,\ell,\sigma}(\varphi)= \sum_{e\in \even(\cset)} \Prob{\?L,\ell,\sigma}(\eventually \varphi_e).
\end{equation}

Fix $\eps > 0$ and define $\gamma\eqdef \frac{\eps}{e_{\max}+2}$. 
To construct an MD strategy~$\hat{\sigma}$ that is  $\eps$-optimal starting from every state  in~$L_{0}$
we have an iterative procedure. 
In each iteration, we define $\hat{\sigma}$ at states in some carefully chosen region; and continuing in this fashion, we gradually fix all choices of~$\hat{\sigma}$. In an iteration,
in order  to fix ``good'' choices in the ``right'' region 
we need to carefully observe the behavior of finitely many 
$\frac{\gamma}{2}$-optimal 
strategies~$\sigma_{\ell_0}$, one for each $\ell_0\in L_0$, which must respect the 
choices already fixed in previous iterations. 
We thus view these strategies~$\sigma_{\ell_0}$ to be $\frac{\gamma}{2}$-optimal not in~$\?L$
but in another layered MDP that is derived from~$\?L$ after
fixing the choices of partially defined~$\hat{\sigma}$.

In more detail, 
the proof consists of  exactly $\frac{e_{\max}}{2}+1$ iterations:
one iteration for each even color~$e$ and a final ``reach'' iteration.
Starting from color~$2$ and $\?L_0 \eqdef \?L$,
 in the iteration~$e \in \{2,\cdots ,e_{\max}\}$,
we obtain a layered MDP~$\?L_{e}$ from~$\?L_{e-2}$
by fixing a single choice for each controlled state in a set~$\fix_{e}$.
Roughly speaking, a run that falls in the set~$\fix_{e}$
is likely going to win through~$\varphi_e$ (win through color~$e)$.
We identify a certain subspace of~$\fix_e$, referred to as $\core_e$, such that 
the following crucial fact holds:  
Once  
 $\core_e$ is visited the run remains in $\fix_e$  with probability at least~$1-\gamma$.
At the final iteration,
we fix the choices of all remaining states to maximize  the probability 
of falling into the union of $\core_e$ sets. 
As mentioned, the majority of such runs that visit  $\core_e$, for some color~$e$, will stay in   $\fix_e$ forever and thus win  parity  through color~$e$. 
After all the iterations, all choices of all controlled states are fixed, and this prescribes the MD strategy~$\hat{\sigma}$ from~$L_0$ in~$\?L$.

In order to define the sets~$\fix_e$ we heavily use L\'evy's zero-one law and follow an inductive transformation on objectives.
L\'evy's zero-one states that, for a given set of (infinite) runs of a Markov chain, if we gradually observe a random run of the chain,  we will become more and more certain whether the random run belongs to that set. This law has a strong implications for tail objectives. It asserts that
on almost all runs $s_0 s_1 s_2 \cdots $ the limit of the value of~$s_i$
w.r.t.\ a tail objective tends to either 0 or 1. %

In each iteration~$e\in \{2,\cdots,e_{\max}\}$, we transform   an objective~$\psi_{e-2}$ to a next objective~$\psi_{e}$ 
where $\psi_0\eqdef \varphi$ is the parity objective and the result of the last transformation is $\psi_{e_{\max}}= \bigvee_{e\in \even(\cset)}  \eventually \core_e$. We will also move from the MDP~$\?L_{e-2}$ to~$\?L_{e}$ after the fixings so as to maintain the following {\bf invariant}:  
For all $\ell_0\in L_0$, the value of~$\ell_0$ for~$\psi_e$ in $\?L_{e}$
is almost as high as its value for $\varphi$ in~$\?L$, that is
\begin{equation}\label{inv-eps}
	\valueof{\?L_e,\psi_e}{\ell_0}\geq \valueof{\?L,\varphi}{\ell_0}-e\cdot\gamma.
\end{equation}
Recall that $\formula=\bigvee_{e\in \even(\cset)} \eventually \varphi_e$. Let $\Fix{0}\eqdef \emptyset$ and write 
$\Fix{e}\eqdef \bigcup_{e'\leq e} \closure{\fix_{e'}}$
 for $e\in \{2, 4,\cdots, e_{\max}\}$.
We   define:
\begin{align}\label{def-psi}
	\psi_0\eqdef  \, \, \bigvee_{e'>0} \eventually \varphi_{e'}
\wedge  \always \, \neg \Fix{0} =\varphi & \quad \quad &
\psi_{e} \eqdef  \, \, \bigvee_{e'\leq e} \eventually \core_{e'} \vee 
\bigvee_{e'>e}  ( \eventually \varphi_{e'}
\wedge \always \, \neg \Fix{e}). 
\end{align}

At each transformation, we examine  the disjunct~$\chi_{e}\eqdef \eventually \varphi_e \wedge  \always  \neg \Fix{e-2}$ in $\psi_{e-2}$. The set of runs satisfying  this objective~$\chi_{e}$ not only win through color~$e$ but also avoid the previously fixed regions. 
Roughly speaking, the aim is to transform~$\chi_e$ to $\eventually \core_e$, to move from~$\psi_{e-2}$ to~$\psi_e$. We apply L\'evy's zero-one law to 
deduce that the runs satisfying the $\chi_{e}$  are likely to enter a region that has a high value for 
a slightly simpler objective, namely
\begin{equation}\label{def-theta}
	\theta_e\eqdef \varphi_{e}
\wedge  \always \, \neg \Fix{e-2}.
\end{equation}
To do so, we  observe in~$\?L_{e-2}$  the behavior of  several arbitrary $\frac{\gamma}{2}$-optimal 
strategies~$\sigma_{\ell_0}$ for~$\psi_{e-2}$, one for each~$\ell_0\in L_0$. Then, for each $\sigma_{\ell_0}$,
we   apply L\'evy's zero-one law separately; this provides that there exists a finite set~$\ini_e$ of states that have a high value for $\theta_e$, and is reached by one of the~$\sigma_{\ell_0}$ with probability as high as the probability of satisfying the  disjunct~$\chi_{e}$. 
Now we use our previous results~\cite{KMST:ICALP2019} on the strategy
complexity of B\"uchi objectives
and  prove the existence of  an MD strategy~$\tau_e$
that is almost optimal  for~$\theta_e$ (error less than $\gamma$), starting from every state in~$\ini_e$. We define sets $\fix_e$ and $\core_e$ to be the set of states from which~$\tau_e$ attains a high probability for~$\theta_e$ in~$\?L_{e-2}$; see Figure~\ref{fig:eps-opt-parity}. Define
$\beta\eqdef  1- \gamma$  and  $ \alpha \eqdef  1 - \gamma^2$, and 
\begin{align} \label{eq:defShell-fix-core}
\fix_e \eqdef   \safesub{\?L_{e-2},\tau_e, \theta_{e}}{\beta} &  \quad &
\core_e \eqdef  \safesub{\?L_{e-2},\tau_e, \theta_{e}}{\alpha}.
\end{align}
We fix the strategy $\tau_e$ in the  $\fix_e$-region to derive the MDP~$\?L_e$ from $\?L_{e-2}$.
Formally,
\begin{equation}
\label{eq:Le}
\?L_e \eqdef \, \, \fixin{\?L_{e-2}}{\tau_e,\fix_e}.
\end{equation}

\begin{figure}[t]
\centering
\includegraphics[width=\linewidth]{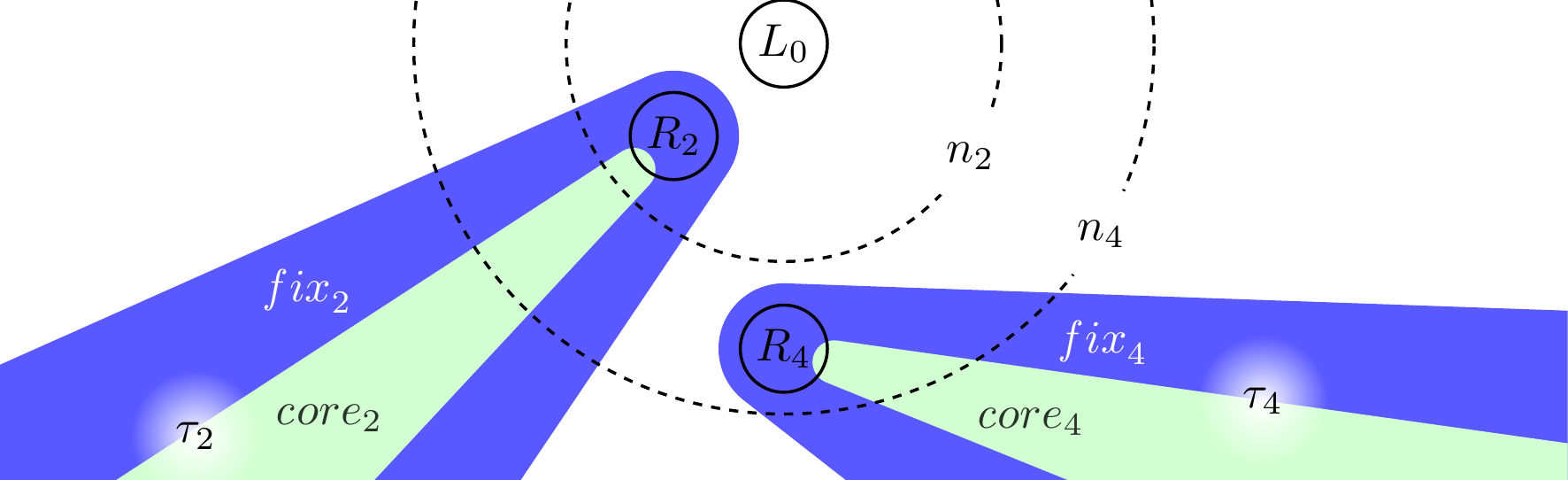}
\caption[bla]{
    The construction for \cref{lem-eps-opt}.
    In the first iteration, for color~$2$, we fix the MD strategy $\tau_2$ in the $\fix_2$-region.
In the second iteration, for color~$4$, we fix $\tau_4$ in  $\fix_4$, and so on for all even colors. Everywhere else we fix an~$\gamma$-optimal reachability strategy towards $\bigcup_{e=2}^{e_{\max}} \core_e$ (in green).}
  \label{fig:eps-opt-parity}
\end{figure}

\subparagraph*{Iteration~$e\in\{2,\cdots ,e_{\max}\}$:}  
For  all states~$\ell_0\in L_0$, 
let $\sigma_{\ell_0}$ be a general (not necessarily MD) $\frac{\gamma}{2}$-optimal
strategy w.r.t.~$\psi_{e-2}$ in the layered MDP~$\?L_{e-2}$.
Consider  the Markov chain~$\chain_{\ell_0}$ induced by~$\?L_{e-2}$,  
the fixed initial state~$\ell_0$ and strategy~$\sigma_{\ell_0}$. 

By definition (Equation \ref{def-theta}),  $\theta_{e}$ is suffix-closed and $\eventually \theta_e$ is tail.  
The strategy~$\sigma_{\ell_0}$ attains~$\eventually \theta_e$ with probability at least  as large as it achieves disjunct~$\chi_{e}$ in~$\psi_{e-2}$.
We apply L\'evy's zero-one law to deduce that the  winning runs of $\eventually \theta_e$ likely reach a finite set~$\ini_e$ of states that have a high value for~$\theta_e$.
In other words, most runs that  eventually win through
 color~$e$, while eventually avoiding  $\Fix{e-2}$, will  reach~$\ini_e$ within a bounded number of
 steps.

\begin{restatable}{lemma}{lemtworeachGzeroone}
\label{lem2:reachG01}\label{lem:LZO-3}
Let~$s_0\in S$ and $\E$ be a suffix-closed objective.
For all $\eps, \eps'>0$, there exist $n$ and a finite set
$F \subseteq \safesub{\E}{1-\eps}$
such that 
$
\Prob{s_0}(\eventually \E \land \eventually^{\leq n} \,F) \geq \Prob{s_0}(\eventually \E)-\eps'.
$
\end{restatable}

By \cref{lem2:reachG01},
there exist $n_{\ell_0}$ and a finite set~$\ini_{\ell_0} \subseteq \safesub{\?L_{e-2}, \theta_e}{\alpha}$
such that 
\begin{equation}\label{eq:fsafe2}
\begin{aligned}
\Prob{\?L_{e-2},\ell_0,\sigma_{\ell_0}}(\eventually\theta_e   \wedge \eventually^{\leq n_{\ell_0}} \, \ini_{\ell_0}) \geq \Prob{\?L_{e-2},\ell_0,\sigma_{\ell_0}}(\eventually \theta_e
)-\frac{\gamma}{2}.
\end{aligned}
\end{equation}
Define $n_e\eqdef \max_{\ell_0\in L_0}(n_{\ell_0})$ and
$\ini\eqdef \bigcup_{\ell_0\in L_0} \ini_{\ell_0}$. 
Write $\ini_e\eqdef \{(s,0)\mid \exists b \cdot (s,b)\in \ini\}$ for the
projection of~$\ini_e$ on the layer $0$.

\begin{remark}\label{rem:quasi-tail}
Suppose $\E' \subseteq \E$ and $\eps > 0$ are such that $\Prob{}(\E') \ge \Prob{}(\E) - \eps$.
Then, for any $\playset$, we have $\Prob{}(\E' \cap \playset) \ge \Prob{}(\E \cap \playset) - \eps$.
\end{remark}
\begin{proof}
We have:
\[
\Prob{}(\E' \cap \playset) 
\ =\ \Prob{}(\E') - \Prob{}(\E' \setminus \playset)
\ \ge\ \Prob{}(\E) - \eps - \Prob{}(\E' \setminus \playset)
\ \ge\ \Prob{}(\E) - \eps - \Prob{}(\E \setminus \playset)
\ =\ \Prob{}(\E \cap \playset) - \eps\,.
\]
\end{proof}

We  
apply Remark~\ref{rem:quasi-tail} to Equation~\eqref{eq:fsafe2}  to get
\begin{equation*}
\begin{aligned}
\Prob{\?L_{e-2},\ell_0,\sigma_{\ell_0}}(\eventually\theta_e \wedge \always \, \neg  \Fix{e-2} \wedge \eventually \closure{\ini_e) }  \geq \,\Prob{\?L_{e-2},\ell_0,\sigma_{\ell_0}}(\eventually \theta_{e}
\wedge \always \, \neg \Fix{e-2})-\frac{\gamma}{2}.
\end{aligned}
\end{equation*}
Since $\eventually \always \, \neg \Fix{e-2} 
\wedge \always \, \neg \Fix{e-2}=\always \, \neg \Fix{e-2}$ and $\chi_{e}=\eventually \varphi_e \wedge  \always  \neg \Fix{e-2}$, 
\begin{equation}
\begin{aligned}\label{eq:fsafe22}
\Prob{\?L_{e-2},\ell_0,\sigma_{\ell_0}}(\chi_{e} \wedge \eventually \closure{\ini_e) }  \geq \, \Prob{\?L_{e-2},\ell_0,\sigma_{\ell_0}}(\chi_{e})-\frac{\gamma}{2}.
\end{aligned}
\end{equation}

We think of ~$\always\eventually \colorset{\states}{= e}{}$ as a B\"uchi condition on a slightly modified MDP.
This allows us to apply the following theorem from \cite{KMST:ICALP2019} about the strategy
complexity of B\"uchi objectives.
\begin{restatable}[Theorem 5 in~\cite{KMST:ICALP2019}]{theorem}{theoBuchiIcalp}
\label{theo-buchi-icalp}
	For every acyclic countable MDP~$\mdp$, a B\"uchi objective~$\varphi$, finite set~$I$ of initial states and $\eps>0$, there exists a deterministic 1-bit strategy   that is $\eps$-optimal from every $s\in I$.
\end{restatable}

Using \cref{theo-buchi-icalp}, we prove the following. %

\begin{restatable}{claim}{claimepsoptimaltaue}
\label{claim:eps-opt-parity-1}
In MDP $\?L_{e-2}$, there is an MD strategy $\tau_e$,  
that is $(\alpha-\beta)$-optimal   for~$\theta_e$ 
from~$\ini_e$.
\end{restatable}

Notice that $\tau_e$ is used to define regions~$\core_e\subseteq \fix_e$; see Equation~\eqref{eq:defShell-fix-core} and  \cref{fig:eps-opt-parity}.
Since $\valueof{\?L_{e-2},\theta_e}{\ell}=\valueof{\?L_{e-2},\theta_e}{\ell'}$ holds for all siblings~$\ell$ and $\ell'$, all states in~$\ini_{e}$ have value~$\geq \alpha$
w.r.t.~$\theta_e$. We have chosen $\tau_e$ to be $(\alpha-\beta)$-optimal, which implies 
$\Prob{\?L_{e-2},\ell,\tau_e}(\theta_e) \geq \beta$ for all~$\ell \in \ini_e$. This shows that  $\ini_e\subseteq \fix_e$.
Strategy~$\tau_e$ is also used to obtain $\?L_e$ from~$\?L_{e-2}$:
 for all controlled states~$\ell\in \fix_e$,  the successor is fixed to be~$\tau_e(\ell)$ in~$\?L_e$, 
 see Equation~\eqref{eq:Le}.

\subparagraph*{Invariant~\eqref{inv-eps}:}
  Given a state~$\ell_0\in L_0$, this invariant states that, for all colors~$e$, $\valueof{\?L_e,\psi_e}{\ell_0}\geq \valueof{\?L,\varphi}{\ell_0}-e\cdot\gamma$ holds.  Recall that $\psi_0=\formula$ and $\?L_0=\?L$. To prove the invariant, by an induction on even colors~$e$, it suffices to  
 prove the following:
\[\valueof{\?L_e,\psi_e}{\ell_0}\geq \valueof{\?L_{e-2},\psi_{e-2}}{\ell_0}-2\gamma.\] 
We construct a
 strategy~$\pi$ for~$\psi_{e}$ in~$\?L_e$ such that 
 $\Prob{\?L_{e},\ell_0,\pi}(\psi_{e})\geq \valueof{\?L_{e-2},\psi_{e-2}}{\ell_0}-2\gamma$.
 Intuitively speaking, $\pi$ enforces that
most runs that win through colors~$e'$, with $e' \leq e$, eventually
reach the $\core_{e'}$-region
 and most remaining winning runs always avoid the $\Fix{e}$-region.

The strategy $\pi$ is defined by combining $\sigma_{\ell_0}$ and $\tau_e$; recall that the strategy $\sigma_{\ell_0}$ is  $\frac{\gamma}{2}$-optimal
w.r.t.~$\psi_{e-2}$ starting from~$\ell_0$ in~$\?L_{e-2}$. 
We define~$\pi$  such that it starts by following  $\sigma_{\ell_0}$.   If it ever enters $\closure{\fix_e}$ then we ensure that it 
enters~$\fix_e$ as well (in at most one more step). Then $\pi$ continues by  playing as~$\tau_e$ does forever.	

The following claim concludes the proof of {\bf Invariant~\ref{inv-eps}}.

\begin{restatable}{claim}{claimepsoptimalpi}
\label{claim-eps-optimal-pi}
$\Prob{\?L_{e},\ell_0,\pi}(\psi_{e})\geq \valueof{\?L_{e-2},\psi_{e-2}}{\ell_0}-2\gamma$.
\end{restatable} 

We summarize the main steps in the proof of \cref{claim-eps-optimal-pi} here.
We first prove the claim that  if $\pi$ ever enters $\closure{\fix_e}$ then it is possible to define it in such a way that it actually enters~$\fix_e$.

Comparing~$\psi_e$ with~$\psi_{e-2}$, one notices that 
two significant terms in the symmetric difference of these two objectives are $\chi_e$ 
and $\eventually \core_e$. 
Roughly speaking, we  use Equation~\eqref{eq:fsafe22} to move from $\chi_e$ to $\eventually \closure{\fix_e}$.
Then  we move from $\eventually \closure{\fix_e}$ to $\eventually \core_e$  by proving that
 $\Prob{\?L_{e},\ell_0,\pi}(\eventually \core_e)$
is almost as high as  
$\Prob{\?L_{e-2},\ell_0,\pi}(\eventually \closure{\fix_e})$, modulo small errors. 
To derive the latter, we rely on    two   facts: 
 another application of L\'evy's zero-one law that guarantees $\Prob{\?L_{e},\ell_0,\pi}(\theta_e \wedge \eventually \core_e)$
 is equal to $\Prob{\?L_{e},\ell_0,\pi}(\theta_e)$;
and the fact that, as soon as $\pi$ visits the first state $\ell\in \fix_e$, it switches to $\tau_e$ forever, and thus attains~$\theta_e$ with probability at least~$\beta$.

\subparagraph*{Reach iteration:}
After all $\frac{e_{\max}}{2}$-iterations for even colors and the fixing, by {\bf Invariant~\eqref{inv-eps}}, for all $\ell_0\in L_0$, we have:
\begin{equation}\label{eqpsimax}
	\valueof{\?L_{e_{\max}},\psi_{e_{\max}}}{\ell_0}\geq \valueof{\?L,\varphi}{\ell_0}-e_{\max}\gamma.
\end{equation}

Recall that~$\psi_{e_{\max}}=\bigvee_{e\in \even(\cset)}  \eventually \core_e$.
At this last iteration, we  fix the choice of all remaining states in $\?L_{e_{\max}}$ such that  the probability of~$\psi_{e_{\max}}$ is maximized.
Recall that there are uniformly $\eps$-optimal MD strategies for reachability objectives~\cite{Ornstein:AMS1969}.
Hence, there is a single MD strategy~$\tau_{\mathrm{reach}}$ in $\?L_{e_{\max}}$ that is uniformly $\gamma$-optimal w.r.t.~$\psi_{e_{\max}}$;
in particular, $\tau_{\mathrm{reach}}$ is  $\gamma$-optimal from every state~$\ell_0\in L_0$.

Let $\?L'\eqdef \fixin{\?L_{e_{\max}}}{\tau_{\mathrm{reach}},L}.$
Let $\hat{\sigma}$ be the MD strategy in~$\?L$ that plays from $L_0$
as prescribed by all the fixings in $\?L'$.
Since all choices in all the $\fix_e$-region are resolved according
 to~$\tau_e$, $e\in \{2,\cdots, e_{\max}\}$, we can apply L\'evy's zero-one
 law another time.

\begin{restatable}{lemma}{lemLZOone}\label{lem:LZO-1}
Let $0< \beta_1<  \beta_2 \leq 1$ and $\E$ a tail objective. For
$s \in \safesub{\E}{\beta_2}$, the following holds:
$
\Prob{s}(\always \, \safesub{\E}{\beta_1})\geq \frac{\beta_2-\beta_1}{1-\beta_1}.
$
\end{restatable}

By \cref{lem:LZO-1}, for all states~$\ell\in \core_e$,  
\begin{equation}\label{eq-gfix-eps1}
\Prob{\?L_{e_{\max}}, \ell,\tau_e} (\always \fix_e)\geq \frac{\alpha-\beta}{1-\beta}\geq 1-\gamma.
\end{equation}
States in $\fix_e$ have a high value for $\theta_e$ and thus also for
$\eventually \varphi_e$.

\begin{restatable}{lemma}{lemLZOtwo}
\label{lem:01lawG}\label{lem:LZO-2}
Let $0 < \beta <1$ and~$\E$ a tail objective. For all states $s\in \safesub{\E}{\beta}$:
\begin{enumerate}
	\item $\Prob{s} (\eventually \always \safesub{\E}{\beta} \setminus \E )=0$; and 
	\item $\Prob{s} (\E \setminus \eventually \always \safesub{\E}{\beta})=0$. 
\end{enumerate}
\end{restatable}

By \cref{lem:LZO-2}.2, we  satisfy  $\eventually \varphi_e$ almost surely:
\begin{equation}\label{eq-gfix-eps2}
\Prob{\?L_{e_{\max}}, \ell,\tau_e} (\eventually \varphi_e \mid \always \fix_e)=1.
\end{equation}
Using Equations~\eqref{eqpsimax} and~\eqref{eq-gfix-eps1}, we prove that
following. %

\begin{restatable}{claim}{claimepsoptimalhatsigma}
\label{claim-eps-optimal-hatsigma}
The MD strategy~$\hat{\sigma}$ is $\eps$-optimal for  parity objective~$\varphi$, from every state $\ell_0\in L_0$.
\end{restatable}

This concludes the proof of \cref{lem-eps-opt}.

\section{Optimal Strategies for Parity}\label{as-par}
In this section we show \cref{theo:opt-par-main}, i.e.,
that optimal strategies for parity, where they exist,
can be chosen deterministic 1-bit Markov.

\begin{figure}
\centering
    \includegraphics[width=\linewidth]{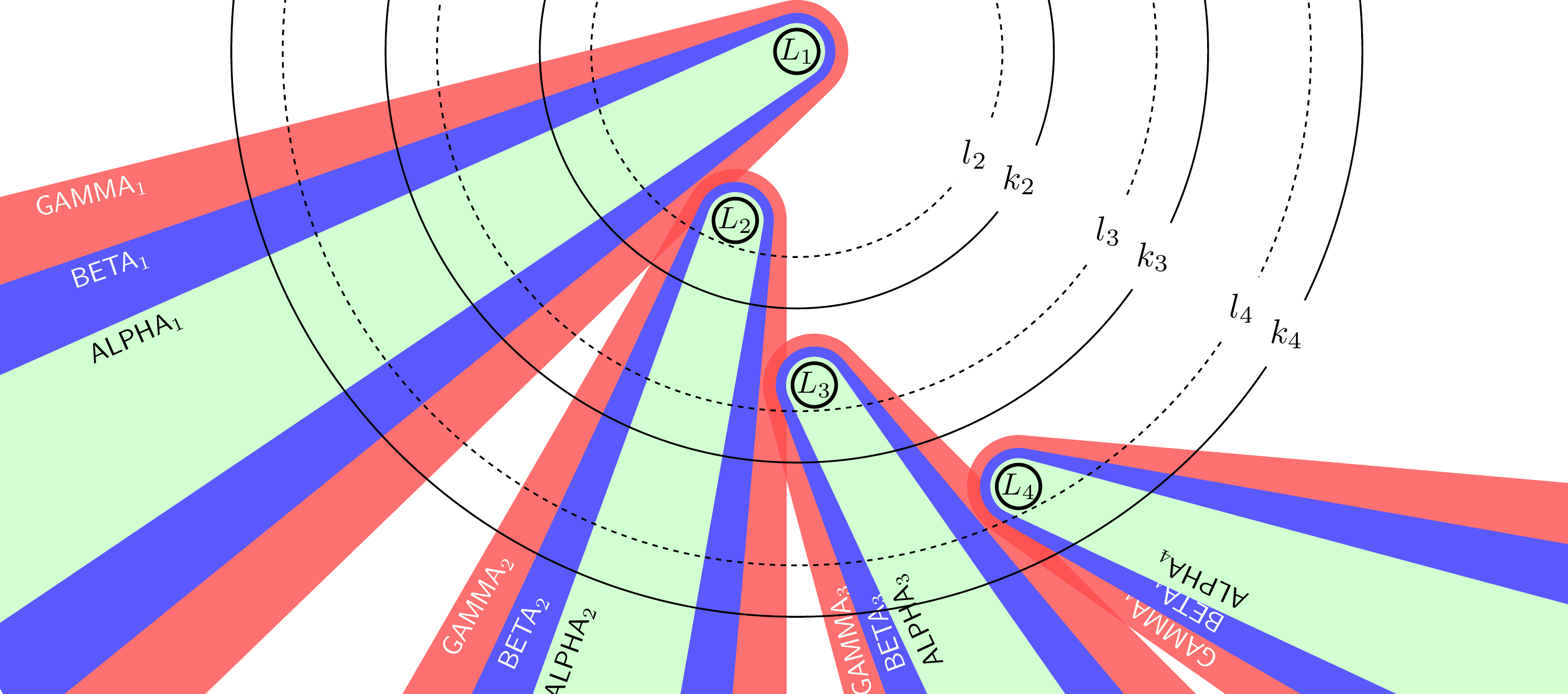}
\caption{Initial segment of the sea urchin construction.
$\?L_i$ is the result of fixing $\obs{i}$ inside $\Betaset{i}$ and then $\ors{i}$
inside the $k_i$-bubble (the set of states reachable from the initial state(s)
in $\le k_i$ steps).
Drawn here for $i=1,2,3,4$. 
}\label{fig:as-par:sea-urchin}
\end{figure}

First we show the main technical result of this section.

\begin{restatable}{lemma}{optparacyclic}
\label{thm:opt-par-acyclic}
Let $\?L(\mdp)$ be the layered MDP obtained from an acyclic and
finitely branching MDP~$\mdp$ and a coloring~$\coloring$
such that
all states are almost surely winning for $\formula=\Parity{\coloring}$
(i.e., every state $\state$ has a strategy $\zstrat_\state$ such that
$\probm_{\?L(\mdp),\state,\zstrat_\state}(\formula) = 1$).

For every initial state $\state_0$
there exists an MD
strategy $\zstrat$
that almost surely wins, i.e.,
$\probm_{\?L(\mdp),\state_0,\zstrat}(\formula) = 1$.
\end{restatable}
\begin{proof}[Proof sketch]
The full version of this rather complex proof can be found in \cref{as-parapp}.

For some intuition consider \cref{fig:as-par:sea-urchin}.
The sea urchin construction is a plastering construction with infinitely many
iterations where MD strategies are fixed in larger and larger subspaces.
Its name comes from the shape of the subspace
in which player choices are fixed up-to iteration $i$: A growing finite body of states
that are reachable from the initial state $\state_0$
within $\le k_i$ steps, plus $i$ different spikes of infinite size.
Each spike is composed of nested subsets
$\Alphaset{i} \subseteq \Betaset{i}$
(and $\subseteq \Gammaset{i}$, which is used only in the correctness argument)
that correspond to different levels of
attainment of certain $\eps$-optimal MD strategies
$\obs{i}$, obtained from \cref{lem-eps-opt}.
Strategy $\obs{i}$ is then fixed in $\Betaset{i}$
(and thus in $\Alphaset{i}$).
Other MD strategies $\ors{i}$ are fixed elsewhere in the finite body,
up-to horizon $k_i$.
Using L\'evy's zero-one law, we prove that, once inside $\Alphaset{i}$,
there is a high chance of never leaving the $i$-th spike $\Betaset{i}$.
Moreover, almost all runs that stay in the $i$-th spike satisfy parity.
Finally, the strategies $\ors{i}$ ensure that
at least $1/2$ (by probability mass) of the runs from $\state_0$
that don't stay in one of the first $i$ spikes will
eventually stay in the $(i+1)$-th spike and
satisfy parity there.
Thus, at the stage with $i$ spikes, the fixed MD strategy attains parity with
some probability $\ge 1-2^{-i}$ already \emph{inside} this fixed subspace.
In the limit of $i \rightarrow \infty$, the resulting MD strategy attains parity almost surely.
\end{proof}

\begin{definition}
\label{def:conditionedmdp}
For a tail objective~$\formula$ and an MDP $\mdp=\mdptuple$,
we define the \emph{conditioned version} of~$\mdp$ w.r.t.~$\formula$
to be the MDP $\pmdp = \tuple{\pstates,\pzstates,\prstates,\ptransition,\pprobp}$ with
$
\pstates = \{s \in \states \mid \exists\,\zstrat.\; \probm_{\mdp,s,\zstrat}(\formula) = \valueof{\mdp}{s} > 0\}
$
and $\pzstates = \pstates \cap \zstates$ and $\prstates = \pstates \cap \rstates$ and
\[
\mathord{\ptransition} = 
\{(s,t) \in \pstates \times \pstates \mid {} s \transition t \text{ and if $s \in \pzstates$ then $\valueof{\mdp}{s} = \valueof{\mdp}{t}$}\}
\]
and $\pprobp : \prstates \to \dist(\pstates)$ so that
$
\pprobp(s)(t) = \probp(s)(t) \cdot \frac{\valueof{\mdp}{t}}{\valueof{\mdp}{s}}
$
for all $s \in \prstates$ and $t \in \pstates$ with $s \, \ptransition \, t$.

\end{definition}

See \cref{app:lics17} for a proof that 
$\pprobp(s)$ is a probability distribution for all $s \in \prstates$
and therefore that the conditioned MDP~$\pmdp$ is well-defined.
The name stems from a useful property
(cf.~\cref{lem:conditioned-construction}.2)
that %
for all strategies that are optimal for~$\formula$ in~$\mdp$, the probability in~$\pmdp$ of any event
is the same as that of its probability in~$\mdp$ conditioned under $\formula$.

The following theorem is a very slight generalization of
\cite[Theorem 5]{KMSW2017} (cf.\ \cref{as-parapp}).
It gives a sufficient condition
under which we can conclude the existence of MD optimal strategies
from the existence of MD almost-sure winning strategies. %

\begin{restatable}{theorem}{reductiontoas}\label{thm:reduction-to-as}
Let $\formula$ be a tail objective.
Let $\mdp=\mdptuple$ be an MDP and $\pmdp
= \tuple{\pstates,\pzstates,\prstates,\ptransition,\pprobp}$ its conditioned
version wrt.\ $\formula$.
Then:
\begin{enumerate}
\item For all $s \in \pstates$ there exists a strategy~$\zstrat$ with $\probm_{\pmdp,s,\zstrat}(\formula) = 1$.
\item Suppose that for every $s \in \pstates$ there exists an MD strategy~$\zstrat''$ with $\probm_{\pmdp,\state,\zstrat''}(\formula) = 1$.
    Then there is an MD strategy~$\zstrat'$ such that for all $\state \in \states$:
\[
\big(
\exists \zstrat \in \zstratset.\,
\probm_{\mdp,\state,\zstrat}(\formula) = \valueof{\mdp}{s}
\big) \Longrightarrow 
\ \probm_{\mdp,\state,\zstrat'}(\formula) = \valueof{\mdp}{s}
\]
\end{enumerate}
\end{restatable}

\begin{theorem}\label{thm:opt-par}
Consider an acyclic MDP $\mdp$ and 
a parity objective.

There exists a deterministic $1$-bit strategy that is
optimal from all states that have an optimal strategy.
\end{theorem}
\begin{proof}
Consider the corresponding layered system $\?L(\mdp)$
(cf.\ \cref{def:layered}), which is also acyclic.
Let $\states_{\mathit{opt}}$ be the subset of states that have an optimal
strategy in $\mdp$.
Thus all states in $\states_{\mathit{opt}}\x\{0,1\}$ have an optimal strategy in $\?L(\mdp)$
by \cref{prop:layered}.

We now use \cref{thm:reduction-to-as} to obtain an MD strategy $\zstrat'$ in $\?L(\mdp)$
that is optimal for all states in $\?L(\mdp)$ that have
an optimal strategy.
First, the parity objective is tail.
Second, in $\?L(\mdp)$, any two siblings have the same value w.r.t.~parity
by \cref{rem:as-layered-memconf}.
Therefore the changes from $\?L(\mdp)$ to its conditioned version
$\?L(\mdp)_*$ (wrt.\ the parity objective)
are symmetric in the two layers.
Thus $\?L(\mdp)_*$ is also a layered acyclic MDP
(i.e., there exists some acyclic MDP $\mdp'$ s.t.\ $\?L(\mdp)_* = \?L(\mdp')$), and
by \cref{thm:reduction-to-as}.1 all states in $\?L(\mdp)_*$ are almost surely winning.
Now we can apply
\cref{thm:opt-par-acyclic}
(generalized to infinitely branching acyclic layered MDPs
by \cref{lem:inf-branching-to-finite-branching}) to $\?L(\mdp)_*$
and obtain that for every state in $\?L(\mdp)_*$ there is an MD strategy that
almost surely wins.
By \cref{thm:reduction-to-as}.2 there is an MD strategy $\zstrat'$ in $\?L(\mdp)$ that is
optimal for all states that have an optimal strategy.
In particular, $\zstrat'$ is optimal for the states in
$\states_{\mathit{opt}}\x\{0,1\}$ in $\?L(\mdp)$.
By \cref{prop:layered}, this yields a deterministic 1-bit strategy in $\mdp$
that is optimal for all states in $\states_{\mathit{opt}}$.
\end{proof}

In \cref{thm:opt-par} the initial memory mode of the 1-bit strategy is
irrelevant (recall \cref{rem:as-layered-memconf}).

\cref{theo:opt-par-main} now follows directly from \cref{thm:opt-par} and \cref{lem:acyclic-Markov}(3).

\section{Optimal Strategies for $\cParity{\{0,1,2\}}$}\label{as012}
\begin{restatable}{theorem}{thmZeroOneTwoPar}\label{thm:012quant}
Let $\mdp=\mdptuple$ be an MDP, $\formula$ a $\cParity{\{0,1,2\}}$ objective
and $\pmdp=\tuple{\pstates,\pzstates,\prstates,\ptransition,\pprobp}$
its conditioned version wrt.\ $\formula$.
Assume that in $\pmdp$ 
for every safety objective (given by some target $T\subseteq\pstates$)
and $\eps>0$ there exists a uniformly $\eps$-optimal MD strategy.
Let $\states_{\mathit{opt}}$ be the subset of states that have an optimal
strategy for $\formula$ in $\mdp$.

Then there exists an MD strategy in $\mdp$ that is optimal for $\formula$
from every state in $\states_{\mathit{opt}}$.
\end{restatable}

The above result generalizes \cite[Theorem~16]{KMSW2017}, which considers only
finitely-branching MDPs and uses the fact that for every safety
objective, an MD strategy exists that is uniformly \emph{optimal}.
This is not generally true for infinitely-branching acyclic MDPs \cite{KMSW2017}.
To prove \cref{thm:012quant}, we adjust the construction so that it only
requires uniformly \emph{$\eps$-optimal} MD strategies for safety objectives
(in the conditioned MDP $\pmdp$).

In order to apply \cref{thm:012quant} to infinitely-branching acyclic MDPs,
we now show that acyclicity guarantees the existence of
uniformly $\eps$-optimal MD strategies for safety objectives.

\begin{restatable}{lemma}{thmEpsOptSafety}
\label{thm:eps-optimal-safety}
For every acyclic MDP with a safety objective and every $\eps>0$
there exists an MD strategy that is uniformly $\eps$-optimal.
\end{restatable}

\noindent
While we defined $\eps$-optimality wrt.~additive errors
(cf.~\cref{sec:prelim}), our proof of \cref{thm:eps-optimal-safety} shows that the claim holds even wrt.~multiplicative
errors (in the style of \cite{Ornstein:AMS1969}).

\begin{restatable}{theorem}{corZeroOneTwoPar}
\label{cor:012Parity}
Consider an MDP $\mdp$ with a $\cParity{\{0,1,2\}}$ objective
and let $\states_{\mathit{opt}}$ be the subset of states that have an optimal strategy.
\begin{enumerate}
\item
If $\mdp$ is acyclic then there exists an MD strategy
that is optimal from every state in $\states_{\mathit{opt}}$.
\item
There exists a deterministic Markov strategy
that is optimal from every state in $\states_{\mathit{opt}}$.
\end{enumerate}
\end{restatable}
\begin{proof}
  Towards item 1, if $\mdp$ is acyclic then also its conditioned version $\pmdp$
  (wrt.\ $\cParity{\{0,1,2\}}$) is acyclic. Thus, by
  \cref{thm:eps-optimal-safety}, in $\pmdp$ for every $\eps >0$ and every
  safety objective there is a uniformly $\eps$-optimal MD strategy.
  The result now follows from \cref{thm:012quant}.

  Item 2 follows from Item 1 and \cref{lem:acyclic-Markov} (item 3 with $k=0$).
\end{proof}

\section{$\eps$-Optimal Strategies for $\cParity{\{0,1\}}$ (co-B\"uchi)}\label{eps01}

\begin{restatable}{theorem}{thmcoBuchi}
\label{thm:coBuchi}
Suppose that $\mdp=\mdptuple$ is an MDP such that
for every safety objective (given by some target $T\subseteq\states$)
and $\eps>0$
there exists a uniformly $\eps$-optimal MD strategy.

Then for every co-B\"uchi objective
(given by some coloring $\coloring:\states\to \{0,1\}$)
and $\eps>0$
there exists a uniformly $\eps$-optimal MD strategy.
\end{restatable}

The precondition of \cref{thm:coBuchi}
is satisfied by many classes of MDPs.
Indeed, we obtain the following.

\begin{corollary}\label{cor:coBuchi}
Consider an MDP $\mdp$ and a co-B\"uchi objective. %
\begin{enumerate}
\item
If $\mdp$ is acyclic then, for every $\eps>0$,
there exists a uniformly $\eps$-optimal MD strategy.
\item
If $\mdp$ is finitely branching 
then, for every $\eps>0$,
there exists a uniformly $\eps$-optimal MD strategy.
\item
For every $\eps>0$
there exists a deterministic Markov strategy that,
from every initial state $\state$,
attains at least $\valueof{\mdp}{\state} -\eps$.
\end{enumerate}
\end{corollary}
\begin{proof}
Towards (1), for acyclic MDPs,
uniformly $\eps$-optimal strategies for safety can be chosen MD by \Cref{thm:eps-optimal-safety}.
Towards (2), for finitely branching MDPs there always exists
even a uniformly optimal MD strategy for every safety objective.
In both cases the claim then follows from \cref{thm:coBuchi}.
Claim (3) follows directly from (1) and \cref{lem:acyclic-Markov} (item 2 with $k=0$). 
\end{proof}

\newpage

\bibliography{journals,conferences,refs}

\newpage
\appendix

\section{Reductions in Section~\ref{epsParity} and related Lemmas}\label{app:sec:auxiliary}

By the following lemma, the strategy complexity of general parity
objectives does not depend on the branching degree of the MDPs.
However, this does not hold for particular parity objectives with a restricted
set of colors, since the construction introduces an extra color.

\infbranchingtofinitebranching*
\begin{proof}
Towards item (1), we encode an infinitely branching acyclic MDP $\mdp$ into a finitely
branching acyclic MDP $\mdp'$.
Every controlled state $x$ with infinite branching
$x \to y_i$ for all $i \in \N$ is replaced by a gadget
$x \to z_1, z_i \to z_{i+1}, z_i \to y_i$ for all $i \in \N$
with fresh controlled states $z_i$.
Infinitely branching random states with $x \step{p_i}{} y_i$ for all $i \in \N$
are replaced by a gadget $x \step{1}{} z_1, z_i \step{1-p_i'} z_{i+1}, z_i \step{p_i'} y_i$ for all
$i \in \N$, with fresh random states $z_i$ and suitably adjusted probabilities
$p_i'$ to ensure that the gadget
is left at state $y_i$ with probability $p_i$, i.e.,
$p_i' = p_i/(\prod_{j=1}^{i-1}(1-p_j'))$.
The fresh states are labeled with an unfavorable color that is smaller than
all other colors, e.g., $-1$.

We take an $\eps$-optimal deterministic 1-bit strategy
$\zstrat'$ for parity from all states $\state \in \states_0$ in $\mdp'$.
We construct a 1-bit deterministic $\eps$-optimal strategy $\zstrat$
for $\mdp$ as follows.
Consider some state $x$ that is infinitely branching in $\mdp$
and its associated gadget in $\mdp'$.
Whenever a run in $\mdp'$ according to $\zstrat'$ reaches $x$ with some
memory value $\alpha \in \{0,1\}$ there exist
values $p_i$ for the probability that the gadget is left at state
$y_i$. Let $p \eqdef 1-\sum_{i\in \N} p_i$ be the probability that the gadget is
never left. (If $x$ is controlled then only one $p_i$ (or $p$) is nonzero,
since $\zstrat'$ is deterministic. If $x$ is random then $p=0$.)
Since $\zstrat'$ is deterministic, the memory updates are deterministic,
and thus there are values $\alpha_i' \in \{0,1\}$ such that whenever
the gadget is left at state $y_i$ the memory will be $\alpha_i'$.
We now define the behavior of the 1-bit deterministic strategy $\zstrat$ 
at state $x$ with memory $\alpha$ in $\mdp$.

If $x$ is controlled and $p\neq 1$ then $\zstrat''$ picks the successor state $y_i$ where
$p_i=1$ and sets the memory to $\alpha_i'$.
If $p=1$ then any run according to $\zstrat'$ that enters the gadget
does not satisfy the objective. Thus $\zstrat$ performs at least as
well in $\mdp$ regardless of its choice, e.g., pick successor $y_1$ and
$\alpha' = \alpha$.

If $x$ is random then $p=0$ and the successor is chosen according to the defined
distribution (which is the same in $\mdp$ and $\mdp'$)
and $\zstrat$ can only update its memory.
Whenever the successor $y_i$ is chosen, $\zstrat$ updates the memory to
$\alpha_i'$.

In states that are not infinitely branching in $\mdp$, $\zstrat$ does
exactly the same in $\mdp$ as $\zstrat'$ in $\mdp'$.

Since all states in the gadgets are labeled with color $-1$, $\zstrat$ performs
at least as well in $\mdp$ as $\zstrat'$ in $\mdp'$ and is thus
$\eps$-optimal from every $\state \in \states_0$.

Towards item (2), the proof is almost identical, expect that we consider optimal
strategies from initial states $\state$ that have an optimal strategy.
\end{proof}

In order to show the existence of Markov (resp.\ 1-bit Markov) strategies,
it suffices to show the existence of memoryless (resp.\ 1-bit)
strategies in an MDP that is made acyclic by encoding a step counter in the state space.
(Note that deterministic $0$-bit strategies are MD strategies
and $0$-bit Markov strategies are Markov strategies.)
This idea appears already in~\cite{KMST:ICALP2019} and can be formally stated as follows.

\acyclicMarkov*
\begin{proof}
The construction is similar for all three items.

Consider an MDP $\mdp=\mdptuple$ with sets of initial states $\states_0$ (finite),
$\states$ and $\states_{\mathit{opt}}$, respectively.

We transform it into an acyclic MDP $\mdp'$ by encoding a
step-counter into the states, i.e.,
$\mdp' = \tuple{\states',\zstates',\rstates',\transition',\probp'}$
where $\states' \eqdef \states \times \N$,
$\zstates' \eqdef \zstates \times \N$,
$\rstates' \eqdef \rstates \times \N$,
$\coloring((\state, n)) \eqdef \coloring(\state)$,
$(\state,n) \transition' (\state',n+1)$ iff $\state \transition \state'$
and $\probp'((\state,n))((\state',n+1)) \eqdef \probp(\state)(\state')$.

For every deterministic $k$-bit strategy $\zstrat'$ in $\mdp'$
there is a corresponding deterministic $k$-bit Markov strategy 
$\zstrat$ in $\mdp$, and vice-versa.
At any state $\state$, $\zstrat$ in memory mode $\memconf$ and
step-counter $n$ plays exactly like $\zstrat'$ in memory mode $\memconf$
at state $(\state,n)$.

It follows from the definition of the colorings that
$\zstrat$ (with memory mode $\memconf$) attains the same
from any initial state $\state$ as $\zstrat'$ (with
memory mode $\memconf$) attains from $(\state,0)$.
Moreover, every state $\state$ has the same value
as its corresponding state $(\state,0)$.

\begin{enumerate}
\item
In $\mdp'$ we consider the set of initial states
$\states_0' \eqdef \states_0 \times \{0\}$, which is finite since $\states_0$
is finite.
By our assumption, for every $\eps > 0$, there exists a deterministic $k$-bit
strategy $\zstrat'$ in $\mdp'$
that is $\eps$-optimal from all states $\state \in \states_0'$.
Thus $\zstrat$ is $\eps$-optimal from all states $\state \in \states_0$.
\item
Like above, except that the set of initial states $\states_0 \eqdef \states$
is not finite. Since $\zstrat'$ is assumed to be $\eps$-optimal from
all states in $\mdp'$, in particular it is $\eps$-optimal from
all states in $\states_0' \eqdef \states \times \{0\}$.
Thus $\zstrat$ is $\eps$-optimal from all states $\state \in \states$.
\item
Here the set of initial states is $\states_{\mathit{opt}}$.
Every state $\state \in \states_{\mathit{opt}}$ has the same value
as its corresponding state $(\state,0) \in \states_{\mathit{opt}} \times
\{0\}$ and the corresponding strategies $\zstrat$ and $\zstrat'$ attain the
same from $\state$ and $(\state,0)$, respectively.
Therefore $\states_{\mathit{opt}} \times \{0\} \subseteq \states_{\mathit{opt}}'$. 
Since the strategy $\zstrat'$ is assumed to be optimal from
all states $\state \in \states_{\mathit{opt}}'$, it is optimal from all states
in $\states_{\mathit{opt}} \times \{0\}$, and thus $\zstrat$ is
optimal from all states in $\states_{\mathit{opt}}$.
\end{enumerate}
\end{proof}

For ease of presentation, we will, instead of showing the existence of $1$-bit
strategies in an acyclic MDP $\mdp$, show the existence of MD strategies in the corresponding
\emph{layered} MDP $\?L(\mdp)$, which encodes the two memory modes into the
states by having two copies of $\mdp$ (called \emph{layers} $0$ and $1$). 
The transitions and probability functions, as well as whether a state is
randomized, and its (parity) color, are lifted naturally.

The next lemma shows the correspondence between deterministic 1-bit strategies
in $\mdp$ and MD strategies in $\?L(\mdp)$.

\proplayered*
\begin{proof}

    For the ``$\mdp \implies \?L(\mdp)$'' direction, given $u[m_0]$,
we define the MD strategy $\tau$ to play in $\?L(\mdp)$ as follows. For  $b,b'\in \{0,1\}$,
\begin{itemize}
  \item for a controlled state~$s\in \zstates$,
if $u[m_0](b,s)=(b',s')$ meaning that $u[m_0]$ chooses~$s'$ at $s$,  by taking a transition $t=(s,s')$,
and updates the bit to~$b'$, we define $\tau((s,b)) \eqdef (t,b)$ and $\tau((t,b))\eqdef(s',b')$;
\item for a random state~$s\in \rstates$,
if $u[m_0]$ updates the memory bit to~$b'$ in case the random successor resolves to~$s'$, by taking a transition $t=(s,s')$, 
we define $\tau((t,b)) \eqdef (s',b')$.
\end{itemize}

Similarly, for the ``$\mdp \impliedby \?L(\mdp)$'' direction,
given $\tau$  in $\?L(\mdp)$, we define an update function~$u$,
such that for all initial bit~$m_0\in \{0,1\}$
 the deterministic $1$-bit strategy $u[m_0]$ in $\mdp$  plays from any state $\state \in \states$ 
as $\tau$ plays in $\?L(\mdp)$ from $(\state, m_0)$. The construction is as follows. For all $b,b'\in \{0,1\}$ and all transitions $t=(s,s')$,
\begin{itemize}
\item if $s\in \zstates$, and if $\tau((s,b))=(t,b)$ and $\tau((t,b))=(s',b')$,  we define $u(b,s)\eqdef(b',s')$; 

\item if $s\in \rstates$, and if $\tau((t,b))=(s',b')$,  we define $u(b,s)(b',s')\eqdef P(s)(s')$. 
\end{itemize}

Denote by $\chain^{\tau}$  the Markov chain obtained from $\?L(\mdp)$ after fixing~$\tau$, and by $\chain^{u[m_0]}$  the Markov chain obtained from $\mdp$ after fixing~$u[m_0]$.
Observe there is  a clear bijection between the runs in the Markov chains $\chain^{\tau}$ and $\chain^{u[m_0]}$.
Since the parity colors are lifted accordingly, we conclude that $
\probm_{\?L(\mdp),(\state_0,m_0),\tau}(\formula) =
\probm_{\mdp,\state_0,\inducedStrat{\updatefun}{m_0}}(\formula) 
$, as required.
\end{proof}

\section{L\'evy's zero-one law}
\label{app:levy01}

We fix   a finitely branching Markov chain $\chain$ with state space~$S$. 
We use the probability measure~$\Prob{s}$ when  starting in a state s.

For an event $\mathcal{E}\in {\mathcal F}$,  the indicator function~$\indicatorfun_{\mathcal{E}}:S^{\omega} \to \{0,1\}$ is defined by
\[
\indicatorfun_{\mathcal{E}}(\rho)=\begin{cases}
1 & \text{if } \rho \in \mathcal{E},\\
0 & \text{otherwise.}
\end{cases}\]

Below we recall    L\'evy's zero-one law;
we state this result for a specific family of  sub $\sigma$-algebras that is used throughout our proofs.
Consider the simplest sequence of sub $\sigma$-algebras $({\mathcal F}_i)_{i\in \mathbb{N}}$ of $\mathcal{F}$
where each~${\mathcal F}_i$ is the $\sigma$-algebra generated by all events that depend only on the length-$i$ prefixes.
Formally, for all $i\in \nat$, define the sub $\sigma$-algebra 
\[\F_i=\{ A\cdot S^{\omega} \subseteq  S^{\omega}\mid   A \subseteq  S^{i}\}.\]
Observe that $\mathcal{F}_1 \subset \mathcal{F}_2 \subset \cdots \subset \mathcal{F}_{\infty}$
where $\mathcal{F}_{\infty}=\F$ is the smallest $\sigma$-algebra containing all the~$\mathcal{F}_i$.
The sub $\sigma$-algebra $\F_i$, $i\in \nat$, introduces an 
equivalence class~$\sim_i$ on~$S^{\omega}$
where $\rho \sim_i \rho'$ if and only if for all
$\E \in \F_i$, the condition
$\rho \in \E \Leftrightarrow \rho' \in \E$ is met.
Given a run~$\rho$, denote by $[\rho]_{\sim_ i}$ the equivalence class of~$\rho$.
By  definition of the~$\F_i$, if~$\rho\in s_1\cdots s_{i} S^{\omega}$
then $[\rho]_{\sim_i}=s_1\cdots s_{i} S^{\omega}$.

Given a state~$s$, 
define the random variable $\Prob{s}(\mathcal{E} \mid \mathcal{F}_i):S^{\omega} \to [0,1] \cup\{\bot\}$
such that, for all runs~$\rho\in s_1\cdots s_{i} S^{\omega}$,  
\begin{equation}\label{eq:defrand}
\Prob{s}(\E \mid \F_i)(\rho)=
\begin{cases}
\Prob{s}(\E\mid [\rho]_{\sim_i})& \text{ if } \Prob{s}(s_1\cdots s_{i}) \neq 0;\\
\bot \text{ (read as undefined) } & \text{ otherwise}.
\end{cases}
\end{equation}
By  L\'evy's zero-one law for all events $\mathcal{E} \subseteq {\mathcal F}_{\infty}$ we have that
\[\lim_{i\to \infty} \Prob{s}(\mathcal{E} \mid \mathcal{F}_i)=\indicatorfun_{\mathcal{E}}\]
holds $\Prob{}$-almost-surely.

\begin{remark}\label{remo1law}
Given a suffix-closed objective~$\E$ and a run~$\rho\in s_1\cdots s_i S^{\omega}$,
if $\Prob{s}(\E \mid \F_i)(\rho)$ is defined, then
\begin{equation*}
\begin{aligned}
\Prob{s_1}(\E \mid \F_i)(\rho)&=
\Prob{s_1}(\E \mid [\rho]_{\sim_i})  \\
&=\Prob{s_1}(\E \mid s_1\cdots s_i S^{\omega})\\
&\le\Prob{s_i}(\E \mid s_i S^{\omega})\\
&=\Prob{s_i}(\E).
\end{aligned}
\end{equation*}
If $\E$ is tail then $\Prob{s_1}(\E \mid \F_i)(\rho) = \Prob{s_i}(\E)$.
\end{remark}

For the fixed Markov chain~$\chain$ and $\eps>0$, 
we define $\safesub{\E}{1-\eps} \eqdef \{s\mid \Prob{s}(\E) \geq 1-\eps\}$.

\begin{lemma} \label{lem:reachG01-helper}
Let~$s_0\in S$ and $\E$ be a suffix-closed objective and $\eps > 0$.
Then $\Prob{s_0}(\eventually \E \land \neg \eventually \safesub{\E}{1-\eps}) = 0$.
\end{lemma}
\begin{proof}
Let $s_0 \in S$.
We have:
\begin{align*}
\denotationof{\always \neg \safesub{\E}{1-\eps}}{s_0}
&\ =\ \{s_0 s_1 \cdots \mid \forall\,i\,.\, \Prob{s_i}(\E) < 1-\eps \} \\
&\ \subseteq\ \{\rho \in s_0 S^\omega \mid \forall\,i\,.\, \Prob{s_0}(\E \mid \F_i)(\rho) < 1-\eps \} & \text{by \cref{remo1law}} \\
&\ \subseteq\ \{\rho \in s_0 S^\omega \mid \lim_{i \to \infty} \Prob{s_0}(\E \mid \F_i)(\rho) \ne 1 \}
\end{align*}
It follows
\begin{equation} \label{eq:reachG01-helper-1}
\Prob{s_0}(\E \land \always \neg \safesub{\E}{1-\eps}) \ \le\ \Prob{s_0}(\E \cap \{\rho \in
s_0 S^\omega \mid \lim_{i \to \infty} \Prob{s_0}(\E \mid \F_i)(\rho) \ne 1 \}) \ = \ 0
\end{equation}
by L\'evy's zero-one law.

Let $s_0 \in S$.
We have:
\begin{align*}
\Prob{s_0}(\eventually \E &\land \neg \eventually \safesub{\E}{1-\eps})\\
&\ =\ \Prob{s_0}(\eventually \E \land \always \neg \safesub{\E}{1-\eps}) \\
&\ =\ \Prob{s_0}\left(\bigcup_{s_1 \cdots s_i \in S^*} s_0 s_1 \cdots s_{i-1} (\E \cap s_i S^\omega) \land \always \neg \safesub{\E}{1-\eps}\right) & \text{} \\
&\ \le\ \sum_{s_1 \cdots s_i \in S^*} \Prob{s_0}(s_0 s_1 \cdots s_{i-1} (\E \cap s_i S^\omega) \land \always \neg \safesub{\E}{1-\eps}) & \text{union bound} \\
&\ \le\ \sum_{s_1 \cdots s_i \in S^*} \Prob{s_i}((\E \cap s_i S^\omega) \land \always \neg \safesub{\E}{1-\eps}) \\
&\ =\ \sum_{s_1 \cdots s_i \in S^*} \Prob{s_i}(\E \land \always \neg \safesub{\E}{1-\eps}) \\
&\ =\ 0 & \text{by \cref{{eq:reachG01-helper-1}}}\tag*{\qedhere}
\end{align*}
\end{proof}

\lemtworeachGzeroone*
\begin{proof}
By \cref{lem:reachG01-helper} we have
\[
\Prob{s_0}(\eventually \E \land \eventually \safesub{\E}{1-\eps}) = \Prob{s_0}(\eventually \E)\,.
\]
By continuity of measures it follows that there is $n$ such that
\[
\Prob{s_0}(\eventually \E \land \eventually^{\le n} \safesub{\E}{1-\eps}) \ge \Prob{s_0}(\eventually \E) - \eps'\,.
\]
Let $\bubble{n}{s_0}$ be the set of states that can be reached from $s_0$ within
at most $n$ steps. Since the Markov chain~$\chain$ is finitely branching,
$F \eqdef \safesub{\E}{1-\eps} \cap \bubble{n}{s_0}$ is a finite set.
Then we have $\denotationof{\eventually^{\le n} F}{} = \denotationof{\eventually^{\le n} \safesub{\E}{1-\eps}}{}$ and the statement of the lemma follows.
\end{proof}

\lemLZOtwo*
\begin{proof}
By  L\'evy's zero-one law, 
\begin{equation}\label{eqq01}
\begin{aligned}
& \, \Prob{s}(\{\rho\mid \lim_{i \to \infty} \Prob{s}(\E \mid \F_i)(\rho)=\indicatorfun_{\E}(\rho)\})=1, \text{ and } \\
 & \,
\Prob{s}(\{\rho\mid \lim_{i \to \infty} \Prob{s}(\E \mid \F_i)(\rho)\neq \indicatorfun_{\E}(\rho)\})=0.
\end{aligned}
\end{equation}
On one hand Equation~\eqref{eqq01} implies that 
\begin{align*}
 & \, \Prob{s}(\{\rho\mid \lim_{i \to \infty} \Prob{s}(\E \mid \F_i)(\rho)=0 \wedge \indicatorfun_{\E} (\rho)=0
\} \\&
\cup \{\rho\mid \lim_{i \to \infty} \Prob{s}(\E \mid \F_i)(\rho)= 1 \wedge \indicatorfun_{\E} (\rho)=1
\})=1 &\\
\Rightarrow & \, \Prob{s}(\{\rho\mid \lim_{i \to \infty} \Prob{s}(\E \mid \F_i)(\rho)=0 
\} \cup \{\rho \mid  \indicatorfun_{\E} (\rho)=1
\})=1 &\\
\Leftrightarrow & \, \Prob{s}(\{\rho\mid \forall \eps>0 \, \exists n\,  \forall i\geq n \, \Prob{s}(\E \mid \F_i)(\rho)< \eps   \}
 \cup \E
)=1 &\\
\Rightarrow & \, \Prob{s}(\{\rho\mid  \exists n\,  \forall i\geq n \, \Prob{s}(\E \mid \F_i)(\rho)< \beta  \}
 \cup \E
)=1 &\\
\Rightarrow & \, \Prob{s}( \eventually \always \neg \safesub{\E}{\beta}
 \cup \E
)=1 & \text{by Remark~\ref{remo1law}}
\intertext{
since $\bigdenotationof{\eventually \always \neg \safesub{\E}{\beta}}{}\subseteq \bigdenotationof{\neg \eventually \always \safesub{\E}{\beta} }{}$}
\Rightarrow & 
\, \Prob{s}(\neg \eventually \always \safesub{\E}{\beta}  \cup  \E)=1  & \\
\Leftrightarrow & 
\, \Prob{s}(\eventually \always \safesub{\E}{\beta}  \cap \neg \E)=0  &\\
\Leftrightarrow  & 
\, \Prob{s}(\eventually \always \safesub{\E}{\beta} \setminus \E)=0 & 
\end{align*}
On the other hand Equation~\eqref{eqq01} implies that 
\begin{equation*}
\begin{aligned}
 & \, \Prob{s}(\{\rho\mid \lim_{i \to \infty} \Prob{s}(\E \mid \F_i)(\rho)\neq 1 \wedge \indicatorfun_{\E} (\rho)=1
\})=0 &\\
\Rightarrow & 
\, \Prob{s}(\{\rho\mid \lim_{i \to \infty} \Prob{s}(\E \mid \F_i)(\rho)\neq 1\} \cap \{\rho \mid \indicatorfun_{\E} (\rho)=1
\})=0 &\\
\Rightarrow & 
\, \Prob{s}(\{\rho\mid \forall n \,  \exists i\geq n \, \Prob{s}(\E \mid \F_i)(\rho)< \beta\} \cap \{\rho \mid \indicatorfun_{\E} (\rho)=1
\})=0 &\\
\Leftrightarrow & 
\, \Prob{s}(\neg \eventually \always \safesub{\E}{\beta}  \cap  \E)=0  & & \text{by Remark~\ref{remo1law}}\\
\Leftrightarrow  & 
\, \Prob{s}(\E \setminus \eventually \always \safesub{\E}{\beta}  )=0  &
\end{aligned}
\end{equation*}

\end{proof}

\begin{corollary}\label{col:01lawG}
Let $0 < \beta <1$ and~$\E$ a tail objective. For all states $s\in \safesub{\E}{\beta}$, we have
 $\Prob{s} (\E \mid \always \safesub{\E}{\beta})=1$.
\end{corollary}
\begin{proof}
Since $\always \safesub{\E}{\beta}$ is contained in $\eventually \always \safesub{\E}{\beta}$,  Lemma~\ref{lem:01lawG}  leads to $\Prob{s} (\always \safesub{\E}{\beta}\setminus \E )=0$. Then, 
\begin{align*}
\Prob{s} ( \E \cap\always \safesub{\E}{\beta})
& =\Prob{s} (\always \safesub{\E}{\beta} )  - \Prob{s} (\always \safesub{\E}{\beta}\setminus \E )\\
& =\Prob{s} (\always \safesub{\E}{\beta}).
\end{align*}
By the above equality, we get that $\Prob{s} (\E \mid \always \safesub{\E}{\beta}) =\frac{\Prob{s} (\E \cap \always \safesub{\E}{\beta})}{\Prob{s} (\always \safesub{\E}{\beta})}=1$.
\end{proof}

\begin{corollary}\label{col:01law-gamma-alpha}
Let $0 < \beta_1, \beta_2 <1$ and~$\E$ a tail objective.
For all states $s$ we have
\[
\Prob{s}(\eventually\always \, \safesub{\E}{\beta_1} \setminus
\eventually\always \, \safesub{\E}{\beta_2})
= 0
\]
\end{corollary}
\begin{proof}
We have that
\begin{equation}
\begin{aligned}
\eventually\always \, \safesub{\E}{\beta_1} \setminus
\eventually\always \, \safesub{\E}{\beta_2}
\\
&=
[\E \cap (\eventually\always \, \safesub{\E}{\beta_1} \setminus
\eventually\always \, \safesub{\E}{\beta_2})]\\
&
\quad  \cup
[(\eventually\always \, \safesub{\E}{\beta_1} \setminus
\eventually\always \, \safesub{\E}{\beta_2}) \setminus \E]\\
&\subseteq
(\E \setminus \eventually\always \, \safesub{\E}{\beta_2})
\cup
(\eventually\always \, \safesub{\E}{\beta_1} \setminus \E).
\end{aligned}
\end{equation}
Thus  
$\Prob{s}(\eventually\always \, \safesub{\E}{\beta_1} \setminus
\eventually\always \, \safesub{\E}{\beta_2})
\le
\Prob{s}(\E \setminus \eventually\always \, \safesub{\E}{\beta_2})
+
\Prob{s}(\eventually\always \, \safesub{\E}{\beta_1} \setminus \E)
= 0$, by \cref{lem:01lawG}.
\end{proof}

\lemLZOone*
\begin{proof}
Write $x$ for~$\Prob{s}(\always  \, \safesub{\E}{\beta_1})$.
We condition the probability of~$\E$ under~$\always \, \safesub{\E}{\beta_1}$. 
By the law of total probability, we have
\[
\beta_2 \leq \, \Prob{s}(\E) =\Prob{s}(\E \mid \always \, \safesub{\E}{\beta_1}) \cdot x+ \Prob{s}(\E\mid \neg \always  \,\safesub{\E}{\beta_1}) \cdot (1-x).\]
By  Corollary~\ref{col:01lawG},  we have $\Prob{s} (\E \mid \always \safesub{\E}{\beta_1})=1$.
Hence  we have
$\beta_2 \leq \,  x+ \beta_1 \cdot (1-x)$; and $x \geq \frac{\beta_2-\beta_1}{1-\beta_1}$ follows. 
\end{proof}

\section{The Conditioned MDP}
\label{app:lics17}
\newcommand{\cyl}{\mathfrak C}
\newcommand{\classcyl}{\mathcal{C}}
\newcommand{\classmon}{\mathcal{Q}}

In this section we adapt some results from~\cite{KMSW2017}.

We will need the following lemma, which is a variant of~\cite[Lemma~20]{KMShW17}:

\begin{lemma}\label{lem:prefix-ind-optimality}
Let $\formula$ be a tail objective.
Let $\mdp=\mdptuple$ be an MDP, and $s_0 \in \states$, and $\zstrat$ be a strategy with $\probm_{\mdp,s_0,\zstrat}(\formula) = \valueof{\mdp}{s_0}$.
Suppose that $s_0 s_1 \cdots s_n$ for some $n \ge 0$ is a partial run starting in $s_0$ and induced by~$\zstrat$.
Then:
\begin{enumerate}
\item $\valueof{\mdp}{s_n} = \probm_{\mdp,s_0,\zstrat}(\denotationof{\formula}{s_0} \mid s_0 s_1 \cdots s_n \states^\omega)$.
\item If $s_n \in \rstates$ then $\valueof{\mdp}{s_n} = \sum_{s_{n+1} \in \states} \probp(s_n)(s_{n+1}) \cdot \valueof{\mdp}{s_{n+1}}$.
\item If $s_n \in \zstates$ then $\valueof{\mdp}{s_n} = \valueof{\mdp}{s_{n+1}}$ for all $s_{n+1} \in \supp(\sigma(s_0 s_1 \cdots s_n))$.
\end{enumerate}
\end{lemma}
\begin{proof}
First we show $\probm_{\mdp,s_0,\zstrat}(\denotationof{\formula}{s_0} \mid s_0 s_1 \cdots s_n \states^\omega) \le \valueof{\mdp}{s_n}$.
Define a strategy $\zstrat' : \states^*\zstates \to \dist(S)$ by $\zstrat'(w) = \zstrat(s_0 s_1 \cdots s_{n-1} w)$ for all $w \in \states^*\zstates$.
Then we have $\probm_{\mdp,s_0,\zstrat}(\denotationof{\formula}{s_0} \mid s_0 s_1 \cdots s_n \states^\omega) = \probm_{\mdp,s_n,\zstrat'}(\denotationof{\formula}{s_n}) \le \valueof{\mdp}{s_n}$.

Next we show $\valueof{\mdp}{s_n} \le \probm_{\mdp,s_0,\zstrat}(\denotationof{\formula}{s_0} \mid s_0 s_1 \cdots s_n \states^\omega)$.
Towards a contradiction, suppose that $\valueof{\mdp}{s_n} > \probm_{\mdp,s_0,\zstrat}(\denotationof{\formula}{s_0} \mid s_0 s_1 \cdots s_n \states^\omega)$.
Then, by the definition of $\valueof{\mdp}{s_n}$, there is a strategy~$\zstrat'$ with $\probm_{\mdp,s_n,\zstrat'}(\denotationof{\formula}{s_n}) > \probm_{\mdp,s_0,\zstrat}(\denotationof{\formula}{s_0} \mid s_0 s_1 \cdots s_n \states^\omega)$.
Define a strategy~$\zstrat''$ that plays according to~$\zstrat$;
if and when partial run $s_0 s_1 \cdots s_n$ is played, then $\zstrat''$ acts like~$\zstrat'$ henceforth;
otherwise $\zstrat''$ continues with~$\zstrat$ forever.
Using the tail property we get:
\begin{align*}
& \probm_{\mdp,s_0,\zstrat''}(\denotationof{\formula}{s_0}) \\
& = \probm_{\mdp,s_0,\zstrat''}(\denotationof{\formula}{s_0} \mid s_0 s_1 \cdots s_n \states^\omega) \cdot \probm_{\mdp,s_0,\zstrat''}(s_0 s_1 \cdots s_n \states^\omega) \\
& \ + \probm_{\mdp,s_0,\zstrat''}(\denotationof{\formula}{s_0} \setminus s_0 s_1 \cdots s_n \states^\omega) \\
& = \probm_{\mdp,s_n,\zstrat'}(\denotationof{\formula}{s_n}) \cdot \probm_{\mdp,s_0,\zstrat}(s_0 s_1 \cdots s_n \states^\omega) \\
& \ + \probm_{\mdp,s_0,\zstrat}(\denotationof{\formula}{s_0} \setminus s_0 s_1 \cdots s_n \states^\omega) && \text{def.~of~$\zstrat''$}\\
& > \probm_{\mdp,s_0,\zstrat}(\denotationof{\formula}{s_0} \mid s_0 s_1 \cdots s_n \states^\omega) \cdot \probm_{\mdp,s_0,\zstrat}(s_0 s_1 \cdots s_n \states^\omega) \\
& \ + \probm_{\mdp,s_0,\zstrat}(\denotationof{\formula}{s_0} \setminus s_0 s_1 \cdots s_n \states^\omega) && \text{def.~of~$\zstrat'$}\\
& = \probm_{\mdp,s_0,\zstrat}(\denotationof{\formula}{s_0}) \\
& = \valueof{\mdp}{s_0} && \text{def.~of~$\zstrat$}
\end{align*}
This contradicts the definition of $\valueof{\mdp}{s_0}$.
Hence we have shown item~1.

Towards items 2~and~3, we extend $\zstrat:\states^*\zstates \to \dist(\states)$ to $\zstrat : \states^* \states \to \dist(\states)$ by defining $\zstrat(w s) = \probp(s)$ for $w \in \states^*$ and $s \in \rstates$.
Then we have for all $s_{n+1} \in \states$:
\begin{equation} \label{eq:lem:prefix-ind-optimality}
\probm_{\mdp,s_0,\zstrat}(s_0 s_1 \cdots s_n s_{n+1} \states^\omega)
= \probm_{\mdp,s_0,\zstrat}(s_0 s_1 \cdots s_n \states^\omega)
  \cdot \zstrat(s_0 s_1 \cdots s_n)(s_{n+1})
\end{equation}
Further we have:
\begin{align*}
& \valueof{\mdp}{s_n} \\
& = \probm_{\mdp,s_0,\zstrat}(\denotationof{\formula}{s_0} \mid s_0 s_1 \cdots s_n \states^\omega)
  && \text{by item 1} \\
& = \frac{\probm_{\mdp,s_0,\zstrat}(\denotationof{\formula}{s_0} \cap s_0 s_1 \cdots s_n \states^\omega)}{\probm_{\mdp,s_0,\zstrat}(s_0 s_1 \cdots s_n \states^\omega)} \\
& = \frac{\sum_{s_{n+1} \in \states} \probm_{\mdp,s_0,\zstrat}(\denotationof{\formula}{s_0} \cap s_0 s_1 \cdots s_n s_{n+1} \states^\omega)}{\probm_{\mdp,s_0,\zstrat}(s_0 s_1 \cdots s_n \states^\omega)} \\
& = \frac{1}{\probm_{\mdp,s_0,\zstrat}(s_0 s_1 \cdots s_n \states^\omega)}
\cdot \sum_{s_{n+1} \in \states} \probm_{\mdp,s_0,\zstrat}(s_0 s_1 \cdots s_n s_{n+1} \states^\omega)
\cdot \mbox{} \\
& \hspace{47mm} \mbox{} \cdot \probm_{\mdp,s_0,\zstrat}(\denotationof{\formula}{s_0} \mid s_0 s_1 \cdots s_n s_{n+1} \states^\omega) \\
& = \sum_{s_{n+1} \in \states}  \zstrat(s_0 s_1 \cdots s_n)(s_{n+1})
\cdot \probm_{\mdp,s_0,\zstrat}(\denotationof{\formula}{s_0} \mid s_0 s_1 \cdots s_n s_{n+1} \states^\omega)
 && \text{by~\eqref{eq:lem:prefix-ind-optimality}} \\
& = \sum_{s_{n+1} \in \states}  \zstrat(s_0 s_1 \cdots s_n)(s_{n+1})
\cdot \valueof{\mdp}{s_{n+1}} 
  && \text{by item 1}
\end{align*}
Thus we have shown item~2.
Towards item~3, suppose $s_n \in \zstates$.
Then, by the tail property, $\valueof{\mdp}{s_n} \ge \valueof{\mdp}{s_{n+1}}$ for all $s_{n+1}$ with $s_n \transition s_{n+1}$.
Since $\zstrat(s_0 s_1 \cdots s_n)$ is a probability distribution, the equality chain above shows that $\valueof{\mdp}{s_n} = \valueof{\mdp}{s_{n+1}}$ for all $s_{n+1} \in \supp(\sigma(s_0 s_1 \cdots s_n))$.
Thus we have shown item 3.
\end{proof}

\begin{lemma}
    The conditioned version $\pmdp$ of $\mdp$ w.r.t.\ tail objective $\formula$ (cf. \cref{def:conditionedmdp}
is well defined.
\end{lemma}
\begin{proof}
By Lemma~\ref{lem:prefix-ind-optimality}.2 we have that $\pprobp(s)$ is a probability distribution for all $s \in \prstates$; hence the conditioned MDP~$\pmdp$ is well-defined.
\end{proof}

The following lemma is a reformulation of \cite[Lemma~6]{KMSW2017}:
\begin{lemma} \label{lem:conditioned-construction}
Let $\formula$ be a tail objective.
Let $\mdp=\mdptuple$ be an MDP, and let $\pmdp = \tuple{\pstates,\pzstates,\prstates,\ptransition,\pprobp}$ be its conditioned version.
Then:
\begin{enumerate}
\item
For all $\zstrat \in \zstratset_{\pmdp}$ and all $n \ge 0$ and all $s_0, \ldots, s_n \in \pstates$ with $s_0 \ptransition\, s_1 \ptransition\, \cdots \ptransition\, s_n$:
\[
\probm_{\pmdp,s_0,\zstrat}(s_0 s_1 \cdots s_n \states^\omega) \ = \ 
\probm_{\mdp,s_0,\zstrat}(s_0 s_1 \cdots s_n \states^\omega) \cdot \frac{\valueof{\mdp}{s_n}}{\valueof{\mdp}{s_0}}
\]
\item 
For all $s_0 \in \pstates$ and all $\zstrat \in \zstratset_{\mdp}$ with $\probm_{\mdp,s_0,\zstrat}(\formula) = \valueof{\mdp}{s_0} > 0$ and all measurable $\playset \subseteq s_0 \states^\omega$ we have
$
\probm_{\pmdp,s_0,\zstrat}(\playset) = \probm_{\mdp,s_0,\zstrat}(\playset \mid \denotationof{\formula}{s_0})
$.
\end{enumerate}
\end{lemma}
\begin{proof}
We prove item~1 by induction on~$n$.
For $n=0$ it is trivial.
For the step, suppose that the equality in item~1 holds for some~$n$.
If $s_n \in \prstates$ then we have:
\begin{align*}
& \probm_{\pmdp,s_0,\zstrat}(s_0 s_1 \cdots s_n s_{n+1} \states^\omega) \\
& = \probm_{\pmdp,s_0,\zstrat}(s_0 s_1 \cdots s_n \states^\omega) \cdot \pprobp(s_n)(s_{n+1}) \\
& = \probm_{\mdp,s_0,\zstrat}(s_0 s_1 \cdots s_n \states^\omega) \cdot \frac{\valueof{\mdp}{s_n}}{\valueof{\mdp}{s_0}} \cdot \pprobp(s_n)(s_{n+1}) && \text{ind.\ hyp.} \\
& = \probm_{\mdp,s_0,\zstrat}(s_0 s_1 \cdots s_n \states^\omega) \cdot \frac{\valueof{\mdp}{s_n}}{\valueof{\mdp}{s_0}} \cdot \probp(s_n)(s_{n+1}) \cdot \frac{\valueof{\mdp}{s_{n+1}}}{\valueof{\mdp}{s_n}} && \text{def.~of~$\pprobp$} \\
& = \probm_{\mdp,s_0,\zstrat}(s_0 s_1 \cdots s_n s_{n+1} \states^\omega) \cdot \frac{\valueof{\mdp}{s_{n+1}}}{\valueof{\mdp}{s_0}}
\end{align*}
Let now $s_n \in \pzstates$.
If $\zstrat(s_0 s_1 \ldots s_n)(s_{n+1}) = 0$ then the inductive step is trivial.
Otherwise we have:
\begin{align*}
& \probm_{\pmdp,s_0,\zstrat}(s_0 s_1 \cdots s_n s_{n+1} \states^\omega) \\
& = \probm_{\pmdp,s_0,\zstrat}(s_0 s_1 \cdots s_n \states^\omega) \cdot \zstrat(s_0 s_1 \ldots s_n)(s_{n+1}) \\
& = \probm_{\mdp,s_0,\zstrat}(s_0 s_1 \cdots s_n \states^\omega) \cdot \frac{\valueof{\mdp}{s_n}}{\valueof{\mdp}{s_0}} \cdot \zstrat(s_0 s_1 \ldots s_n)(s_{n+1}) && \text{ind.\ hyp.} \\
& = \probm_{\mdp,s_0,\zstrat}(s_0 s_1 \cdots s_n \states^\omega) \cdot \frac{\valueof{\mdp}{s_{n+1}}}{\valueof{\mdp}{s_0}} \cdot \zstrat(s_0 s_1 \ldots s_n)(s_{n+1}) && \text{def.~of~$\mathord{\ptransition}$} \\
& = \probm_{\mdp,s_0,\zstrat}(s_0 s_1 \cdots s_n s_{n+1} \states^\omega) \cdot \frac{\valueof{\mdp}{s_{n+1}}}{\valueof{\mdp}{s_0}}
\end{align*}
This completes the inductive step, and we have proved item~1.

Towards item~2, let $s_0 \in \pstates$ and $\zstrat \in \zstratset_{\mdp}$ such that $\probm_{\mdp,s_0,\zstrat}(\formula) = \valueof{\mdp}{s_0} > 0$.
Observe that $\zstrat$ can be applied also in the MDP~$\pmdp$.
Indeed, for any $s \in \pzstates$, if $t$ is a possible successor state of~$s$ under~$\zstrat$, then $\valueof{\mdp}{s} = \valueof{\mdp}{t}$ by Lemma~\ref{lem:prefix-ind-optimality}.3 and thus $t \in \pstates$.

Let again $n \ge 0$ and $s_0, s_1, \ldots, s_n \in \states$.
\begin{itemize}
\item
Suppose $s_0 s_1 \cdots s_n$ is a partial run in~$\pmdp$ induced by~$\zstrat$.
Then we have:
\begin{align*}
& \probm_{\pmdp,s_0,\zstrat}(s_0 s_1 \cdots s_n \states^\omega) \cdot \probm_{\mdp,s_0,\zstrat}(\formula) \\
& = \probm_{\mdp,s_0,\zstrat}(s_0 s_1 \cdots s_n \states^\omega) \cdot \frac{\valueof{\mdp}{s_n}}{\valueof{\mdp}{s_0}}
\cdot \probm_{\mdp,s_0,\zstrat}(\formula)
 && \text{item~1} \\
& = \probm_{\mdp,s_0,\zstrat}(s_0 s_1 \cdots s_n \states^\omega) \cdot \valueof{\mdp}{s_n}
 && \text{assumption~on~$\zstrat$} \\
& = \probm_{\mdp,s_0,\zstrat}(s_0 s_1 \cdots s_n \states^\omega) \cdot \probm_{\mdp,s_0,\zstrat}(\denotationof{\formula}{s_0} \mid s_0 s_1 \cdots s_n \states^\omega)
 && \text{Lemma~\ref{lem:prefix-ind-optimality}.1} \\
& = \probm_{\mdp,s_0,\zstrat}(\denotationof{\formula}{s_0} \cap s_0 s_1 \cdots s_n \states^\omega)
\end{align*}
\item
Suppose $s_0 s_1 \cdots s_n$ is not a partial run in~$\pmdp$ induced by~$\zstrat$.
Hence $\probm_{\pmdp,s_0,\zstrat}(s_0 s_1 \cdots s_n \states^\omega) = 0$.
If $s_0 s_1 \cdots s_n$ is not a partial run in~$\mdp$ induced by~$\zstrat$ then $\probm_{\mdp,s_0,\zstrat}(s_0 s_1 \cdots s_n \states^\omega) = 0$.
Otherwise, since $\zstrat$ is optimal, there is $i \le n$ with $\valueof{\mdp}{s_i} = 0$, hence $\probm_{\mdp,s_0,\zstrat}(\denotationof{\formula}{s_0} \cap s_0 s_1 \cdots s_n \states^\omega)$.
In either case we have
$
\probm_{\pmdp,s_0,\zstrat}(s_0 s_1 \cdots s_n \states^\omega) \cdot \probm_{\mdp,s_0,\zstrat}(\formula)
= 0 
= \probm_{\mdp,s_0,\zstrat}(\denotationof{\formula}{s_0} \cap s_0 s_1 \cdots s_n \states^\omega)
$.
\end{itemize}
In either case we have the equality $\probm_{\pmdp,s_0,\zstrat}(\playset) = \probm_{\mdp,s_0,\zstrat}(\playset \mid \denotationof{\formula}{s_0})$ for cylinders $\playset = s_0 s_1 \cdots s_n \states^\omega$.
Since probability measures extend uniquely from cylinders~\cite{billingsley-1995-probability}, the equality holds for all measurable $\playset \subseteq s_0 \states^\omega$.
Thus we have shown item~2.
\end{proof}

The following lemma is~\cite[Lemma~7]{KMSW2017}.
\begin{lemma}\label{lem:MDP-as-uniform}
Let $\mdp=\mdptuple$ be an MDP.
Let $\formula$ be an objective that is prefix-independent in~$\{\mdp\}$.
Suppose that for any $s \in \states$ and any strategy~$\zstrat$ with $\probm_{\mdp,\state,\zstrat}(\formula) = 1$ there exists an MD-strategy~$\zstrat'$ with $\probm_{\mdp,\state,\zstrat'}(\formula) = 1$.
Then there is an MD-strategy~$\zstrat'$ such that for all $\state \in \states$:
\[
\big(
\exists \zstrat \in \zstratset.\,
\probm_{\mdp,\state,\zstrat}(\formula) = 1
\big)
\quad\Longrightarrow\quad
\probm_{\mdp,\state,\zstrat'}(\formula) = 1
\]
\end{lemma}
\begin{proof}
We can assume that all states are almost-surely winning, since in order to achieve an almost-sure winning objective, the player must forever remain in almost-surely winning states.
So we need to define an MD-strategy~$\zstrat'$ so that for all $s \in \states$ we have $\probm_{\mdp,\state,\zstrat'}(\formula) = 1$.

Fix an arbitrary state $s_1 \in \states$.
By assumption there is an MD-strategy~$\zstrat_1$ with $\probm_{\mdp,s_1,\zstrat_1}(\formula) = 1$.
Let $U_1 \subseteq \states$ be the set of states that occur in plays that both start from~$s_1$ and are induced by~$\zstrat_1$.
We have $\probm_{\mdp,s_1,\zstrat_1}(\denotationof{\formula}{s_1} \cap U_1^\omega) = 1$.
In fact, for any $s \in U_1$ and any strategy~$\zstrat$ that agrees with~$\zstrat_1$ on~$U_1$ we have $\probm_{\mdp,s,\zstrat}(\denotationof{\formula}{s} \cap U_1^\omega) = 1$.

If $U_1=\states$ we are done.
Otherwise, consider the MDP~$\mdp_1$ obtained from~$\mdp$ by fixing~$\zstrat_1$ on~$U_1$ (i.e., in~$\mdp_1$ we can view the states in~$U_1$ as random states).
We argue that, in~$\mdp_1$, for any state~$s$ there is an MD-strategy~$\zstrat_1'$ with $\probm_{\mdp_1,\state,\zstrat_1'}(\formula) = 1$.
Indeed, let $s \in \states$ be any state.
Recall that there is an MD-strategy~$\zstrat$ with $\probm_{\mdp,\state,\zstrat}(\formula) = 1$.
Let $\zstrat_1'$ be the MD-strategy obtained by restricting~$\zstrat$ to the non-$U_1$ states (recall that the $U_1$ states are random states in~$\mdp_1$).
This strategy~$\zstrat_1'$ almost surely generates a run that \emph{either} satisfies~$\formula$ without ever entering~$U_1$ \emph{or} at some point enters~$U_1$.
In the latter case, $\formula$ is satisfied almost surely: this follows from prefix-independence and the fact that $\zstrat_1'$ agrees with~$\zstrat_1$ on~$U_1$.
We conclude that $\probm_{\mdp_1,\state,\zstrat_1'}(\formula) = 1$.

Let $s_2 \in \states \setminus U_1$.
We repeat the argument from above, with $s_2$ instead of~$s_1$, and with $\mdp_1$ instead of~$\mdp$.
This yields an MD-strategy~$\zstrat_2$ and a set $U_2 \ni s_2$ with $\probm_{\mdp_1,s_2,\zstrat_2}(\denotationof{\formula}{s_2} \cap U_2^\omega) = 1$.
In fact, for any $s \in U_2$ and any strategy~$\zstrat$ that agrees with~$\zstrat_2$ on~$U_2$ and with~$\zstrat_1$ on~$U_1$ we have $\probm_{\mdp,s,\zstrat}(\denotationof{\formula}{s} \cap U_2^\omega) = 1$.

If $U_1 \cup U_2=\states$ we are done.
Otherwise we continue in the same manner, and so forth.
Since $\states$ is countable, we can pick $s_1, s_2, \ldots$ to have $\bigcup_{i \ge 1} U_i= \states$.
Define an MD-strategy~$\zstrat'$ such that for any $s \in \zstates$ we have $\zstrat'(s) = \zstrat_i(s)$ for the smallest~$i$ with $s \in U_i$.
Thus, if $s \in U_i$, we have $\probm_{\mdp,s,\zstrat'}(\formula) \ge \probm_{\mdp,s,\zstrat'}(\denotationof{\formula}{s} \cap U_i^\omega) = 1$.
\end{proof}

The following lemma is~\cite[Lemma~8]{KMShW17}.
\begin{lemma} \label{lem:measure-theory}
Let $\states$ be countable and $s \in \states$.
Call a set of the form $s w \states^\omega$ for $w \in \states^*$ a \emph{cylinder}.
Let $\probm, \probm'$ be probability measures on $s S^\omega$ defined in the standard way, i.e., first on cylinders and then extended to all measurable sets $\playset \subseteq s \states^\omega$.
Suppose there is $x \ge 0$ such that $x \cdot \probm(\cyl) \le \probm'(\cyl)$ for all cylinders~$\cyl$.
Then $x \cdot \probm(\playset) \le \probm'(\playset)$ holds for all measurable $\playset \subseteq s S^\omega$.
\end{lemma}
\begin{proof}
Let $\classcyl = \{\cyl \subseteq s \states^\omega \mid \cyl \text{ cylinder}\}$ denote the class of cylinders.
This class generates an algebra $\classcyl_* \supseteq \classcyl$, which is the closure of~$\classcyl$ under finite union and complement.
The classes $\classcyl$ and $\classcyl_*$ generate the same $\sigma$-algebra $\sigma(\classcyl)$.
The class~$\classcyl_*$ is the set of finite disjoint unions of cylinders~\cite[Section~2]{billingsley-1995-probability}.
Hence $x \cdot \probm(\playset) \le \probm'(\playset)$ for all $\playset \in \classcyl_*$.

Define
\[
 \classmon = \{\playset \in \sigma(\classcyl) \mid x \cdot \probm(\playset) \le \probm'(\playset) \}\,.
\]
We have $\classcyl \subseteq \classcyl_* \subseteq \classmon \subseteq \sigma(\classcyl)$.
We show that $\classmon$ is a \emph{monotone} class, i.e., if $\playset_1, \playset_2, \ldots \in \classmon$, then $\playset_1 \subseteq \playset_2 \subseteq \cdots$ implies $\bigcup_i \playset_i \in \classmon$, and $\playset_1 \supseteq \playset_2 \supseteq \cdots$ implies $\bigcap_i \playset_i \in \classmon$.
Suppose $\playset_1, \playset_2, \ldots \in \classmon$ and $\playset_1 \subseteq \playset_2 \subseteq \cdots$.
Then:
\begin{align*}
x \cdot \probm\Big(\bigcup_i \playset_i\Big) 
& = \sup_i x \cdot \probm(\playset_i) && \text{measures are continuous from below} \\
& \le \sup_i \probm'(\playset_i) && \text{definition of~$\classmon$} \\
& = \probm'\Big(\bigcup_i \playset_i\Big) && \text{measures are continuous from below}
\end{align*}
So $\bigcup_i \playset_i \in \classmon$.
Using the fact that measures are continuous from above, one can similarly show that if $\playset_1, \playset_2, \ldots \in \classmon$ and $\playset_1 \supseteq \playset_2 \supseteq \cdots$ then $\bigcap_i \playset_i \in \classmon$.
Hence $\classmon$ is a monotone class.

Now the \emph{monotone class theorem} (see, e.g., \cite[Theorem~3.4]{billingsley-1995-probability}) implies that $\sigma(\classcyl) \subseteq \classmon$, thus $\classmon = \sigma(\classcyl)$.
Hence $x \cdot \probm(\playset) \le \probm'(\playset)$ for all $\playset \in \sigma(\classcyl)$.
\end{proof}

The following theorem is a variant of~\cite[Theorem~5]{KMSW2017}.

\reductiontoas*
\begin{proof}
Towards item~1, let $s \in \pstates$.
By the definition of~$\pstates$, there is a strategy~$\zstrat$ with $\probm_{\mdp,s,\zstrat}(\formula) = \valueof{\mdp}{s} > 0$.
By Lemma~\ref{lem:conditioned-construction}.2, we have $\probm_{\pmdp,s,\zstrat}(\formula) = 1$, as desired.

It remains to prove item~2.
Suppose that for any $s \in \pstates$ there exists an MD-strategy~$\zstrat''$ with $\probm_{\pmdp,\state,\zstrat''}(\formula) = 1$.
By Lemma~\ref{lem:MDP-as-uniform}, it follows that there is an MD-strategy~$\zstrat'$ with $\probm_{\pmdp,s,\zstrat'}(\formula) = 1$ for all $s \in \pstates$.
We show that this strategy~$\zstrat'$ satisfies the property claimed in the statement of the theorem.

To this end, let $n \ge 0$ and $s_0, s_1, \ldots, s_n \in \states$.
If $s_0 s_1 \cdots s_n$ is a partial run in~$\pmdp$ then, by Lemma~\ref{lem:conditioned-construction}.1, 
\[
\probm_{\pmdp,s_0,\zstrat'}(s_0 s_1 \cdots s_n \states^\omega) 
\ = \
\probm_{\mdp,s_0,\zstrat'}(s_0 s_1 \cdots s_n \states^\omega) \cdot \frac{\valueof{\mdp}{s_n}}{\valueof{\mdp}{s_0}} \,,
\]
and thus, as $\valueof{\mdp}{s_n} \le 1$,
\[
\valueof{\mdp}{s_0} \cdot \probm_{\pmdp,s_0,\zstrat'}(s_0 s_1 \cdots s_n \states^\omega) 
\ \le \ \probm_{\mdp,s_0,\zstrat'}(s_0 s_1 \cdots s_n \states^\omega)\,.
\]
If $s_0 s_1 \cdots s_n$ is not a partial run in~$\pmdp$ then $\probm_{\pmdp,s_0,\zstrat'}(s_0 s_1 \cdots s_n \states^\omega) = 0$ and the previous inequality holds as well.
Therefore, by Lemma~\ref{lem:measure-theory}, we get for all measurable sets $\playset \subseteq s_0 \states^\omega$:
\[
\valueof{\mdp}{s_0} \cdot \probm_{\pmdp,s_0,\zstrat'}(\playset) \le \probm_{\mdp,s_0,\zstrat'}(\playset)
\]
In particular, since $\probm_{\pmdp,s_0,\zstrat'}(\formula) = 1$, we obtain $\valueof{\mdp}{s_0} \le \probm_{\mdp,s_0,\zstrat'}(\formula)$.
The converse inequality $\probm_{\mdp,s_0,\zstrat'}(\formula) \le \valueof{\mdp}{s_0}$ holds by the definition of~$\valueof{\mdp}{s_0}$, hence we conclude $\probm_{\mdp,s_0,\zstrat'}(\formula) = \valueof{\mdp}{s_0}$.
\end{proof}

\section{Missing proofs in Section~\ref{epsParity}}\label{epsParityapp}
We first recall our results~\cite{KMST:ICALP2019} on the strategy complexity of B\"uchi objectives:

\theoBuchiIcalp*

We will prove that

\claimepsoptimaltaue*
\begin{claimproof}
Consider the original MDP~$\mdp$. Given a set $B\subseteq L$ in~$\?L$, we
use $\mathrm{project}(B) \eqdef \{s\mid (s,b)\in B, b\in\{0,1\}\}$ to project the set into~$\mdp$.

We first slightly modify~$\mdp$ to obtain~$\mdp'$. The  modification guarantees
that, for all states~$s$ and  runs~$\rho$ of~$\M'$,
\[s\rho \in \bigdenotationof{\always\eventually \colorset{\states}{= e}{} }{} \quad \text{ if and only if } \quad s\rho \in \bigdenotationof{\always\eventually \colorset{S}{= e}{} \wedge
\always \colorset{S}{\leq  e }{} \wedge \always \neg \mathrm{project}(\Fix{e})}{}.\]
We redirect all  out-going transitions of  states $s'\in \mathrm{project}(\Fix{e})$ or $s'$
with $\coloring(s')>e$  to an infinite chain~ $q_0 q_1 q_2 \cdots$ of controlled states
where $\coloring(q_i)=1$ and~$q_i \transition{} q_{i+1}$. 
We also update the color of all states~$s$ with $\coloring(s)<e$ to~$1$.

The objective $\always\eventually \colorset{\states}{= e}{}$ is 
a B\"uchi Objective in~$\M'$.
By  Theorem~\ref{theo-buchi-icalp}, given 
the \emph{finite} set~$\mathrm{project}(\ini_e)$ of initial states,
there exists  a deterministic 1-bit strategy~$\sigma$ in $\mdp'$ that is
$(\alpha-\beta)$-optimal w.r.t.~$\always\eventually \colorset{\states}{= e}{}$ for
every state $s\in \mathrm{project}(\ini_e)$ (with the memory bit initially set to~$0$).  

Since the fixed choices  in~$\?L_{e-2}$ are only in the $\Fix{e-2}$-region,
strategy~$\sigma$ can be translated in a natural way
to  a deterministic memoryless strategy~$\sigma'$ in~$\?L_{e-2}$:
For a state~$s\in \zstates$ and $b \in \{0,1\}$,
if $\sigma$  chooses the successor state~$s'$,  by taking a transition $t=(s,s')$,
and updates the bit to~$b'$, we define $\sigma'((s,b)) \eqdef (t,b)$ and $\sigma'((t,b))=(s',b')$.
For a random state~$s\in \rstates$ and $b \in \{0,1\}$,
if the strategy
$\sigma$ updates the memory bit to~$b'$ in case the random successor resolves to~$s'$, by taking a transition $t=(s,s')$ , 
we define $\sigma'((t,b)) \eqdef (s',b')$.
Recall that the bit is initially set to~$0$ in $\sigma$. Consequently,
 the strategy $\sigma'$ is 
$(\alpha-\beta)$-optimal  for~$\theta_e$  from  every state~$\ell\in \ini_e$ in the layered MDP~$\?L_{e-2}$.
\end{claimproof}

\bigskip

We next prove the main technical claim in Section~\ref{epsParity}: 

\claimepsoptimalpi*
\begin{claimproof}

Recall the definition of~$\pi$: it starts by following  $\sigma_{\ell_0}$.   If it ever enters $\closure{\fix_e}$ then we ensure that it 
enters~$\fix_e$ as well (in at most one more step). Then $\pi$ continues by  playing as~$\tau_e$ does forever.	

Below we argue that   
 if  $\pi$  ever enters $\closure{\fix_e}$ then 
  it is in fact possible to choose the layer in such a way that $\pi$  enters~$\fix_e$ instead. Assume 
  $\pi$ enters $\closure{\fix_e}$ at $q$ after taking a transition from $p$ to $q$. Let $\bar{q}\in \fix_e$ be the sibling of~$q$.
  By construction, 
  \begin{enumerate}
  \item either $p \in \transition_1\times \{0,1\}$ is controlled: the controller switches the layer in~$p$, by choosing $\bar{q}$ rather than~$q$ and enters~$\fix_e$; 
  \item or 	$q \in \transition_1\times \{0,1\}$ is  controlled.
    By definition~\eqref{eq:defShell-fix-core}, the MD strategy $\tau_e$ attains a high value from state~$\bar{q}$ for~$\theta_e$. Hence, $\tau_e(\bar{q})\in \fix_e$. Hence, the controller can switch the layer in~$q$ by playing~$\tau_e(\bar{q})$ and enters~$\fix_e$.
  \end{enumerate}

For all $e'\in \{2,4,\cdots,e_{\max}\}$ define 
\begin{align*}
	\chi_{e'}\eqdef \, \eventually \varphi_{e'} \wedge \always \, \neg \Fix{e-2} &\quad &
	\tilde{\chi}_{e'}\eqdef \, \eventually \varphi_{e'} \wedge \always \, \neg \Fix{e}.
\end{align*}
We define 
\begin{equation}
\begin{aligned}
	\psi\eqdef \bigvee_{e'<e} \eventually \core_{e'} \vee \bigvee_{e'>e} \chi_{e'}\\
	\psi'\eqdef \bigvee_{e'<e} \eventually \core_{e'} \vee \bigvee_{e'>e} \tilde{\chi}_{e'}
\end{aligned}
\end{equation}
By definition of $\psi_{e-2}$ and $\psi_e$, see definition~\eqref{def-psi}, we have $\psi_{e-2}=\psi \vee \chi_e$ and $\psi_e=\psi' \vee \eventually \core_{e}$. 
For brevity, further define $\rho\eqdef\eventually \closure{\fix_e}$.
 Observe that $\denotationof{ \psi \wedge \neg \rho}{}\subseteq \denotationof{\psi'}{}$.

We first have that 
\begin{equation} \label{eq-claim4.5.2}
\begin{aligned}
\Prob{\?L_{e},\ell_0,\pi}&(\eventually \core_e)\, &\\
\geq &\, \Prob{\?L_{e},\ell_0,\pi}(\eventually \core_{e} \land \rho)	&\\
\geq & \, \Prob{\?L_{e},\ell_0,\pi}(\neg \closure{\fix_e} \text{ until }(\closure{\fix_e} \wedge \eventually \varphi_e \wedge \eventually \core_e \wedge \always \,\neg \Fix{e-2} ))&\\
= & \, \sum_{\ell \in \closure{\fix_e}}\Prob{\?L_{e},\ell_0,\pi}(\neg \closure{\fix_e} \text{ until } \ell) \cdot  
\Prob{\?L_{e},\ell,\tau_e}(\theta_e \wedge \eventually \core_e)&\\
= & \, \sum_{\ell \in \closure{\fix_e}}\Prob{\?L_{e},\ell_0,\pi}(\neg \closure{\fix_e} \text{ until } \ell) \cdot  
\Prob{\?L_{e},\ell,\tau_e}(\theta_e )& \text{by~\cref{lem:LZO-2}.2}\\
\geq &\, \sum_{\ell \in \closure{\fix_e}}\Prob{\?L_{e},\ell_0,\pi}(\neg \closure{\fix_e} \text{ until } \ell) \cdot  
\beta &\\
= &\, \Prob{\?L_{e},\ell_0,\pi}(\rho) \cdot  
\beta &\\
= &\, \Prob{\?L_{e-2},\ell_0,\sigma_{\ell_0}}(\rho) \cdot  
\beta &\\
 = &\, \Prob{\?L_{e-2},\ell_0,\sigma_{\ell_0}}(\rho) \cdot  
(1-\gamma) &\\
\geq &\, \Prob{\?L_{e-2},\ell_0,\sigma_{\ell_0}}(\rho) - \gamma &
\end{aligned}
\end{equation}

We use the  law of total probability:
\begin{equation}\label{eq:claim4.5.h}
\begin{aligned}
	\Prob{\?L_{e-2},\ell_0,\sigma_{\ell_0}} (\psi_{e-2})= \Prob{\?L_{e-2},\ell_0,\sigma_{\ell_0}}(\psi \wedge \neg \rho) + \Prob{\?L_{e-2},\ell_0,\sigma_{\ell_0}}(\chi_e)+ \Prob{\?L_{e-2},\ell_0,\sigma_{\ell_0}}(\psi \wedge \rho)
\end{aligned}	
\end{equation}

In one hand, since $\?L_{e}$ and $\?L_{e-2}$ only differ in the $\fix_e$-region, and since $\pi$ plays as $\sigma_{\ell_0}$ on all runs contained in $\neg \rho$:
\[\Prob{\?L_{e-2},\ell_0,\sigma_{\ell_0}} (\psi \wedge \neg \rho) = \Prob{\?L_{e},\ell_0,\pi} (\psi \wedge \neg \rho) \leq \Prob{\?L_{e},\ell_0,\pi} (\psi')\]

In the other hand, by Equation~\eqref{eq:fsafe22}:  
\begin{align*}
	\Prob{\?L_{e-2},\ell_0,\sigma_{\ell_0}}(\chi_e) \leq &  \, \Prob{\?L_{e-2},\ell_0,\sigma_{\ell_0}}(\chi_e \wedge \rho)+\frac{\gamma}{2}%
\end{align*}

Applying the above to Equation~\eqref{eq:claim4.5.h} yields:
\begin{equation*}
\begin{aligned}
	\Prob{\?L_{e-2},\ell_0,\sigma_{\ell_0}}& (\psi_{e-2}) 
	  \\
	\leq & \, \Prob{\?L_{e},\ell_0,\pi} (\psi')+ \Prob{\?L_{e-2},\ell_0,\sigma_{\ell_0}}(\chi_e \wedge \rho) + \Prob{\?L_{e-2},\ell_0,\sigma_{\ell_0}}(\psi \wedge \rho) + \frac{\gamma}{2}&\\
	= & \, \Prob{\?L_{e},\ell_0,\pi} (\psi')+ \Prob{\?L_{e-2},\ell_0,\sigma_{\ell_0}}((\chi_e \vee \psi) \wedge \rho)+ \frac{\gamma}{2}& \text{since $\chi_e$ and $\psi$ are disjoint}\\
	\leq &\, \Prob{\?L_{e},\ell_0,\pi} (\psi')+ \Prob{\?L_{e-2},\ell_0,\sigma_{\ell_0}}( \rho) + \frac{\gamma}{2}&\\
	\leq & \, \Prob{\?L_{e},\ell_0,\pi} (\psi')+ \Prob{\?L_{e-2},\ell_0,\pi}( \eventually \core_e)  + \gamma+ \frac{\gamma}{2} & \text{by Equation~\eqref{eq-claim4.5.2}}\\
	= & \, \Prob{\?L_{e},\ell_0,\pi} (\psi_{e})  + \frac{3\gamma}{2} &
	\end{aligned}	
\end{equation*}

To conclude the proof we recall that $\sigma_{\ell_0}$ is $\frac{\gamma}{2}$-optimal w.r.t~$\psi_{e-2}$.
\end{claimproof}

\claimepsoptimalhatsigma*
\begin{claimproof}
	
 For the MD strategy~$\hat{\sigma}$, by the law of total probability, we have
\begin{align*}
\Prob{\?L, \ell_0,\hat{\sigma}}(\varphi)\geq & \sum_{e\in \even(\cset)}\Prob{\?L,\ell_0,\hat{\sigma}}(\eventually \varphi_e \land \eventually \core_e)\,. &
\intertext{Let $\cset'$ be the set of even colors~$e$ where $\Prob{\?L, \ell_0,\hat{\sigma}} (\eventually \core_e)>0$. Then:}
=&  \, \sum_{e\in \cset'}\Prob{\?L,\ell_0,\hat{\sigma}}(\eventually \varphi_e \mid \eventually \core_e) \cdot \Prob{\?L, \ell_0,\hat{\sigma}} (\eventually \core_e) 
&\\
\geq &  \,  \sum_{e\in \cset'}\Prob{\?L,\ell_0,\hat{\sigma}}(\always \fix_e \mid \eventually \core_e) \cdot \Prob{\?L, \ell_0,\hat{\sigma}} (\eventually \core_e)& \text{
by Equation~\eqref{eq-gfix-eps2}
}\\
\geq &\, \sum_{e\in \cset'}(1-\gamma) \cdot \Prob{\?L, \ell_0,\hat{\sigma}} (\eventually \core_e)& \text{
by Equation~\eqref{eq-gfix-eps1}
}\\
= &\, (1-\gamma) \cdot \sum_{e\in \even(\cset)}\Prob{\?L, \ell_0,\hat{\sigma}} (\eventually \core_e)&\\
\geq &\, (1-\gamma) \cdot \Prob{\?L, \ell_0,\hat{\sigma}} (\psi_{e_{\max}})&
\intertext{since $\tau_{\mathrm{reach}}$ is $\gamma$-optimal and by Equation~\eqref{eqpsimax}, }
\geq &\, (1-\gamma) \cdot (\valueof{\?L,\varphi}{\ell_0}-e_{\max}\gamma-\gamma)&\\
\geq &\, \valueof{\?L,\varphi}{\ell_0}-(e_{\max}+2)\gamma &
\end{align*} 
Recall that $\eps=(e_{\max}+2)\gamma$.
Thus we have shown that the MD strategy $\hat{\sigma}$ is $\eps$-optimal w.r.t.~$\varphi$ from every state~$\ell_0 \in L_0$.
\end{claimproof}

\section{Missing proofs in Section~\ref{as-par}}\label{as-parapp}

\begin{definition}[Bubbles]
    \label{def-fix-bubb}
    Let $\mdp$ be an MDP with states $\states$, $R\subseteq \states$, $l\in\N$.
    The \emph{$l$-bubble} around $R$ is the set
    \[\bubblearound{\mdp}{R}{l}
    \eqdef\{s \mid \exists s_0\in R. \exists \tau. \probm_{\mdp,\state_0,\tau}(\eventually^{\leq l} s) > 0 \}\]
    of states that can be reached from $R$ in at most $l$ steps.
    Any bubble around a closed set $R\subseteq L$ is closed.

\end{definition}

Recall that for an MD strategy $\tau$,
    we write $\fixin{\mdp}{\tau,R}$ for the MDP
    obtained from $\mdp$ by fixing the strategy $\tau$ for all states in $R$.
    We will simply write $\fixin{\mdp}{\tau}$ for $\fixin{\mdp}{\tau, \states}$,
    where $\tau$ is fixed everywhere,
    and 
    $\fixinbubble{\mdp}{\tau}{R}{l}\eqdef
    \fixin{\mdp}{\tau,\bubblearound{\mdp}{R}{l}}$
    that fixes $\tau$ in the $l$-bubble around~$R$.

\optparacyclic*
\begin{proof}
Directly from \cref{thm:opt-par-acyclic-general} (let $L_0 \eqdef \closure{\{\state_0\}}$).
\end{proof}

\begin{lemma}\label{thm:opt-par-acyclic-general}
Let $\?L(\mdp)$ be the layered MDP obtained from an acyclic and
finitely branching MDP~$\mdp$ and a coloring~$\coloring$
such that
all states are almost surely winning for $\formula=\Parity{\coloring}$
(i.e., every state $\state$ has a strategy $\zstrat_\state$ such that
$\probm_{\?L(\mdp),\state,\zstrat_\state}(\formula) = 1$).

For every finite closed set $L_0$ of initial states
there exists an MD
strategy $\hat{\zstrat}$
that almost surely wins
from every state $\state_0\in L_0$.
That is,
$\forall \state_0 \in L_0.~
\probm_{\?L(\mdp),\state_0,\hat{\zstrat}}(\formula) = 1$.
\end{lemma}
\begin{proof}
We iteratively produce an infinite sequence $\?L_0,\?L_1,\?L_2,\ldots$ of
layered MDPs. They have the same structure as $\?L(\mdp)$, but in each step
from $\?L_i$ to $\?L_{i+1}$ the choices in some subset of states
(reachable from $L_0$) are fixed.
In the limit all choices from all controlled states reachable from $L_0$ are
fixed. Hence this prescribes an MD strategy $\hat{\sigma}$ from $L_0$ in $\?L(\mdp)$.
It is not sufficient that these fixings of MD strategies in subspaces are
compatible with some strategy almost sure winning for $\formula$, since progress
(e.g., towards visiting a particular color) might only be made outside of the fixed
subspace, and thus be delayed forever.
Instead we prove the stronger property that $\hat{\sigma}$ ensures $\formula$ with some probability
$p_i(s_0)$ from $s_0 \in L_0$ already in the fixed
subspace of $\?L_i$ alone, and that $\lim_{i\rightarrow\infty} p_i(s_0) =1$.
This then implies that $\hat{\sigma}$ is almost surely winning for $\formula$ in $\?L(\mdp)$.

\paragraph*{The sea urchin construction.}
Its name comes from the shape of the subspace where strategies are fixed: a
finite body $H_i$ out of which come finitely many spikes ($\Betaset{i}$,
where each spike is infinite). As the body grows, more spikes are added.
Eventually the sea urchin covers the entire space; see \cref{fig:as-par:sea-urchin}.

The construction uses some global thresholds 
$1>\alpha>\beta>\gamma>0$, to be determined later.
Moreover, in each step from $\?L_i$ to $\?L_{i+1}$ we will define the
following notions.
\begin{itemize}
\item
  Small error thresholds $\err{i}{j}>0$
  for $j\in \{0,1,2,3\}$).
\item
  Thresholds $l_i, k_i \in \N$ of a number of steps from $L_0$.
\item
  Finite closed subsets of states $H_i$ where $H_0 \eqdef \emptyset$ and
  $H_i \eqdef \bubblearound{\?L}{L_0}{k_i}$ for $i > 0$.
  ($H_i$ is finite, because $\?L$ is finitely branching.)
\item
  Finite subsets $L_i \subseteq L$ as starting sets for certain modified
  objectives $\formula_i$ (see below).
\item
  MD strategies $\obs{i}$ (for $i >0$) and
  subsets of states $\Alphaset{i} \subseteq \Betaset{i} \subseteq \Gammaset{i} \subseteq L$,
  where $\Alphaset{i}$ (resp.\ $\Betaset{i}$, $\Gammaset{i}$) are the sets of
  states from which $\obs{i}$ attains $\ge \alpha$ (resp.\ $\ge \beta$,
  $\ge \gamma$) for objective $\formula_i$ (see below) in $\?L_{i-1}$ (and $\?L_i$).
  Let $\Alphaset{0} = \Betaset{0} = \Gammaset{0} \eqdef \emptyset$,
  $\Alphaset{i} \eqdef \safesub{\?L_{i-1},\obs{i},\formula_{i}}{\alpha}$
  (and similar for $\Betaset{i}, \Gammaset{i}$).

We write $\Alphaset{\le i}\eqdef \bigcup_{j\le i}\Alphaset{j}$ and similar for
  $\Betaset{\le i},\Gammaset{ \le i}$.
\item
  Let $\Fixx{i} \eqdef \Betaset{\le i} \cup H_i$.
  This is the subspace where choices are fixed in rounds up-to $i$.
\item
 Modified objectives $\formula_i$ with $\formula_0 \eqdef \formula$ and
 $\formula_\nexti \eqdef \formula \land \always (L\setminus\closure{\Fixx{i}})$.
 For $i>0$ the $\formula_i$ are not strictly tail objectives,
 but they still enjoy the same properties as tail objectives
 wrt.\ the Levy zero-one law; cf.~\cref{rem:quasi-tail}. 
\item
In $\?L_i$ the choices inside $\Fixx{i}$ are already fixed.
Inside $\Betaset{i}$ the strategy $\obs{i}$ is fixed,
and inside $H_i \setminus \Betaset{\le i}$ the choices are fixed according
to another MD strategy $\ors{i}$.
\item
It follows from the properties above that we have the invariant
\begin{equation}\label{eq:invariant}
  \Gammaset{\nexti} \cap \closure{\Fixx{i}} = \emptyset
\end{equation}
In particular, the sets $\Betaset{j}$ are disjoint for different $j$.
However, for $j' >j$, it is possible that $\Gammaset{j}$ overlaps with
$\Betaset{j'}$ (and $\Gammaset{j'}$).
\end{itemize}

\paragraph*{Base case.}
We start with the MDP $\?L_0 \eqdef \?L \eqdef \?L(\mdp)$.
By assumption, in $\?L_0$ all states are almost surely winning for
$\formula_0 \eqdef \formula$ (the unrestricted parity objective).
The invariant \eqref{eq:invariant} is trivially satisfied for $i=0$,
since $\closure{\Fixx{0}} = \emptyset$.

\paragraph*{Step.}
Now we define the step from $\?L_i$ to $\?L_{i+1}$ for $i \ge 0$.
We assume that for all $j \le i$ the MD strategies $\obs{j}$ and the 
sets $\Alphaset{j} \subseteq \Betaset{j} \subseteq \Gammaset{j}$ and $H_j$ are already 
defined.
Moreover, $\obs{j}$ is fixed inside $\Betaset{j}$, and in $\?L_j$ the strategy
$\obs{j}$ attains at least $\beta$ for objective $\formula_j$
from each state $\state \in \Betaset{j}$.
Moreover, some other MD strategy is fixed in $H_j \setminus \Betaset{\le j}$.
(All this trivially holds for the base case $i=0$. For $i>0$ our construction
will ensure these properties.)

We now consider $\?L_i$.
By $\state_0$ we denote initial states in $L_0$. (General states are denoted by $\state$.)
We show that in $\?L_i$, all initial states $\state_0 \in L_0$ are still almost surely winning
for $\formula$, as witnessed by a \emph{resetting strategy} $\zstrat$ defined
below (where $\zstrat$ is generally not MD, except inside the subspace $\Fixx{i}$).
First we need a basic property of $\Alphaset{j}, \Betaset{j}$.

\begin{claim}
    \label{claim:as-par:progress-in-B}
    Let $0\le j\le \previ$ and $\sigma$ be an arbitrary strategy in $\?L_\previ$.
    If $\state\in\Alphaset{j}$ then
    $\probm_{\?L_\previ,\state,\sigma}(\always~\Betaset{j}) \ge \frac{\alpha-\beta}{1-\beta}$.
\end{claim}
\begin{proof}
    By \cref{lem:LZO-1}, since $\sigma$ behaves just like $\obs{j}$ in the relevant subspaces already fixed to $\obs{j}$ in $\?L_\previ$.
\end{proof}

Recall that for every state $\state\in L$ there exists an almost surely winning strategy $\sigma(\state)$ for $\formula$ in $\?L$.
The \emph{resetting strategy} $\zstrat$ in $\?L_\previ$ starts in $L_0$ and behaves as
specified in the three different modes $\memconf_1,\memconf_2,\memconf_3$ as follows. For all $j \le i$:
\begin{enumerate}
\item%
    In $H_j$ it plays as prescribed by the fixing there, (starting in memory mode $\memconf_1$).
\item%
  Whenever $\zstrat$ enters a set $\closure{\Betaset{j}} \setminus H_i$ then 
  it switches to mode $\memconf_2$ and chooses the layer in such a way that it
  enters even $\Betaset{j}$ (in at most one more step)
  and continues playing $\obs{j}$, as required by the fixing inside $\Betaset{j}$.
  \footnote{By \cref{def:layered}, either the current state or the next state
    allows to switch between layers; cf.\ the proof of \cref{claim-eps-optimal-pi}.}
\footnote{Remember that \eqref{eq:invariant} implies that the sets $\Betaset{j}$ are disjoint.}
Inside $\Betaset{j}$, it plays $\obs{j}$ that is fixed in $\Betaset{j}$.
It continues to play $\obs{j}$ even in $\Gammaset{j}\setminus (\Betaset{\le i} \cup H_i)$.
\item%
While playing in mode $\memconf_2$ (or $\memconf_1$), upon
reaching an unfixed state $\state$ outside of $\Gammaset{j}$
(and outside of $H_i$), it goes to mode $\memconf_3$ and resets to an almost
surely winning strategy $\sigma(\state)$ for $\formula$ in $\?L$.
It keeps playing $\sigma(\state)$ until (and if) it reaches the fixed part
$\Betaset{\le i} \cup H_i$, whereupon it continues as before with mode $\memconf_2$.
\end{enumerate}

We will see that, not only is the resetting strategy $\zstrat$ almost surely winning for
$\formula$,
but every time it re-enters $\Betaset{\le i}$ it has a lower-bounded chance of
\emph{eventually} staying in $\Alphaset{\le i}$ forever.

We now classify the runs induced by the resetting strategy $\zstrat$ (from
some initial state $\state_0 \in L_0$) according to how often
which modes $\memconf_1,\memconf_2,\memconf_3$ are used.

First we note that, since $\?L_i$ is acyclic, under any strategy
(and in particular $\zstrat$), any run can
visit any finite set (in particular $H_\previ$) only finitely often and therefore has an
infinite suffix that is always outside $H_i$. Thus $\zstrat$ is eventually
always not in mode $\memconf_1$.

By our invariant \eqref{eq:invariant},
playing in $\Betaset{i}$ and $\Gammaset{i}$ is not restricted by our previous
fixings in $\Betaset{\le i-1} \cup H_{i-1}$.
Thus, when playing from $\state \in \Betaset{i}$ in mode $\memconf_2$, we keep playing $\obs{i}$
even in $\Gammaset{i}$. Analogously to \cref{claim:as-par:progress-in-B},
the chance of staying in the set $\Gammaset{i}$ can be lower bounded.
\begin{equation}
    \label{claim:as-par:exit-1}
\forall \state \in \Betaset{i}\, \probm_{\?L_\previ,\state,\zstrat[\memconf_2]}(\always~\Gammaset{i}) \ge \frac{\beta-\gamma}{1-\gamma} > 0
\end{equation}
This again follows from \cref{lem:LZO-1}, observing that $\zstrat$ behaves
just like $\obs{i}$ even inside $\Gammaset{i}$ while staying in mode $\memconf_2$.

When playing from $\state \in \Betaset{j}$ for some $j < i$, a similar
property holds.
If a run visits some state $\state' \in \closure{\Betaset{j'}}$, for some $j' >j$,
then we can assume that we have even $\state' \in \Betaset{j'}$ by our
assumption on $\sigma$ above, because outside of the fixed region the layer can be chosen freely.
Then the strategy switches from $\obs{j}$ to $\obs{j'}$ from
$\state' \in \Betaset{j'}$.
Otherwise we keep playing $\obs{j}$ while in $\Gammaset{j}$, i.e., by
\cref{lem:LZO-1},
we get, for every $\state \in\Betaset{j}$, that
\begin{equation}
    \label{claim:as-par:exit-1-new}
    \probm_{\?L_\previ,\state,\zstrat[\memconf_2]}(\always~\Gammaset{j} \vee \eventually~\Betaset{>j}) \ge \frac{\beta-\gamma}{1-\gamma} > 0.
\end{equation}
From \eqref{claim:as-par:exit-1} and \eqref{claim:as-par:exit-1-new},
we obtain that the set of runs that infinitely often switch from mode $\memconf_2$ to
$\memconf_3$ are a null-set.
Moreover, as shown above, every run has an infinite suffix where the mode is
not $\memconf_1$.
It follows that, except for a null-set, all runs either have an infinite
suffix in mode $\memconf_2$ or an infinite suffix in mode $\memconf_3$.
Let $\playset_2$ and $\playset_3$ denote these subsets of runs, respectively.
I.e., we have $\forall \state_0 \in L_0$
\begin{equation}\label{eq:mode2or3}
\probm_{\?L_\previ,\state_0,\zstrat}(\playset_2 \cup \playset_3)=1.
\end{equation}
In mode $\memconf_3$ the resetting strategy $\zstrat$ plays an almost surely winning strategy
for $\formula$ outside of $\closure{\Fixx{i}}$ that is
not impeded by the fixings in $\?L_i$, and $\formula$ is a tail
objective. Thus, for all $\state_0 \in L_0$,
\begin{equation}\label{eq:mode3impliesParity}
\probm_{\?L_\previ,\state_0,\zstrat}(\playset_3)
= \probm_{\?L_\previ,\state_0,\zstrat}(\playset_3 \wedge \formula \wedge \eventually\always(L\setminus\closure{\Fixx{i}})).
\end{equation}
From the property that $\formula$ is tail and the definition of
$\formula_\nexti$ as $\formula \land \always (L\setminus\closure{\Fixx{i}})$
we obtain that, for all $\state_0 \in L_0$,
\begin{equation}\label{eq:mode3impliesNextObj}
\probm_{\?L_\previ,\state_0,\zstrat}(\playset_3)
= \probm_{\?L_\previ,\state_0,\zstrat}(\playset_3 \wedge \eventually \formula_\nexti)
\end{equation}

In mode $\memconf_2$ the resetting strategy $\zstrat$ plays some MD strategy
$\obs{j}$ in $\Gammaset{j}$ (for some $j \le i$). 
Thus, for all $\state_0 \in L_0$,
\begin{equation}\label{eq:mode2impliesStayGamma}
\probm_{\?L_\previ,\state_0,\zstrat}(\playset_2)
= \probm_{\?L_\previ,\state_0,\zstrat}(\playset_2 \wedge \eventually\always\,\Gammaset{\le i}).
\end{equation}
Since in mode $\memconf_2$ the resetting strategy $\zstrat$ plays some MD strategy
$\obs{j}$ with attainment $\ge \gamma$ (resp.\ $\ge \alpha$) in
$\Gammaset{j}$ (resp.\ $\Alphaset{j}$), we can apply Levy's
zero-one law (\cref{col:01law-gamma-alpha})
and obtain even
\begin{equation}\label{eq:mode2impliesStayAlpha}
\probm_{\?L_\previ,\state_0,\zstrat}(\playset_2)
= \probm_{\?L_\previ,\state_0,\zstrat}(\playset_2 \wedge \eventually\always\,\Alphaset{\le i})
\end{equation}
By \eqref{eq:mode2or3}, \eqref{eq:mode3impliesNextObj}
and \eqref{eq:mode2impliesStayAlpha}
we obtain
\begin{equation}\label{eq:alphaOrNextObj}
\probm_{\?L_\previ,\state_0,\zstrat}(\eventually\always\,\Alphaset{\le i}
\,\vee\, \eventually \formula_\nexti)=1
\end{equation}
$\zstrat$ plays like $\obs{j}$ inside $\Alphaset{j}$ which attains
$\ge \alpha$ for $\formula$.
By using Levy's zero-one law (\cref{lem:01lawG}(1))
for safety sets at level $\alpha$, we obtain that
$\probm_{\?L_\previ,s_0,\zstrat}\left(\formula \wedge
\eventually\always\,\Alphaset{\le i}\right) =
\probm_{\?L_\previ,s_0,\zstrat}\left(\eventually\always\,\Alphaset{\le i}\right)
$.
Since $\eventually\formula_{i+1} \subseteq \formula$ it follows from
\cref{eq:alphaOrNextObj}
that
$\probm_{\?L_\previ,s_0,\zstrat}\left(\formula\right) =1$, i.e., the resetting
strategy $\zstrat$ wins $\formula$ almost surely.

For $s_0 \in L_0$ let
\[
p_i(s_0) \eqdef \probm_{\?L_\previ,s_0,\sigma}\left(\formula\, \wedge\,
\always\,\Fixx{i}\right)
\]
be the attainment for $\formula$ inside the fixed region $\Fixx{i}$ of $\?L_i$.

Since $H_i$ is finite and $\?L_i$ is acyclic, almost surely $H_i$ is
eventually left forever. Moreover, the sets $\Alphaset{j}$ are safety sets (at
level $\alpha$) for $\formula$. It follows from Levy's
zero-one law (cf.~\cref{col:01lawG}) that 
\begin{equation}\label{eq:fixi-implies-parity}
p_i(s_0)
=
\probm_{\?L_\previ,s_0,\sigma}\left(\formula \wedge \always \,\Fixx{i}\right)
=
\probm_{\?L_\previ,s_0,\sigma}\left(\always \,\Fixx{i}\right)
\end{equation}

Let's now consider only those runs from states $s_0 \in L_0$ that do \emph{not}
satisfy $\always\, \Fixx{i}$ (the rest satisfy $\formula$ already inside the
fixed part of $\?L_i$ by \eqref{eq:fixi-implies-parity}).
From \eqref{eq:alphaOrNextObj} we obtain
\begin{equation}\label{eq:as-par:eventually-decide-nonfixed2}
\begin{aligned}
\probm_{\?L_\previ,s_0,\zstrat}\left(\left(\eventually\always\,\Alphaset{\le i}
\,\vee\, \eventually \formula_\nexti\right)
\wedge \neg\always\,\Fixx{i}\right)\\
    =
    \probm_{\?L_\previ,s_0,\zstrat}\left(\neg\always\,\Fixx{i}\right)
\end{aligned}
\end{equation}

Using \cref{lem:LZO-3}, we show the following claim.

\begin{restatable}{claim}{claimreachLnext}
\label{claim:reach-L-next}
For every $\err{\nexti}{1}, \err{\nexti}{2} >0$,
there must exist a
threshold $l_{\nexti}$ and a finite set
\[
L_{\nexti}' \subseteq \safesub{\?L_\previ,(\always\,\Alphaset{\le
i}\lor \formula_\nexti)}{1-\err{\nexti}{1}}
\]
such that,
following $\sigma$ from any state $\state_0\in L_0$,
the chance of satisfying $\neg\always\,\Fixx{i}$
and within at most $l_{\nexti}$ steps reaching a state $s$
in $L_{\nexti}'$ is at least
$\probm_{\?L_\previ,s_0,\zstrat}\left(\neg\always\,\Fixx{i}\right)(1-\err{\nexti}{2})$.
\begin{equation}\label{eq:claim:reach-L-next}
\probm_{\?L_\previ,s_0,\zstrat}\left(\neg\always\,\Fixx{i} \wedge \eventually^{\le
l_{\nexti}}L_{\nexti}'\right) \ge
\probm_{\?L_\previ,s_0,\zstrat}\left(\neg\always\,\Fixx{i}\right)(1-\err{\nexti}{2})
\end{equation}
\end{restatable}
\begin{proof}
For those $\state_0 \in L_0$ where $\probm_{\?L_\previ,s_0,\zstrat}\left(\neg\always\,\Fixx{i}\right)=0$ 
the claim holds trivially.

We now consider the remaining cases of those states $\state_0 \in L_0$ where
$\probm_{\?L_\previ,s_0,\zstrat}\left(\neg\always\,\Fixx{i}\right) > 0$.
Let
\begin{equation}\label{eq:cl53-0}
\delta \eqdef \err{\nexti}{2} \cdot
\min_{\state_0 \in
  L_0}\{\probm_{\?L_\previ,s_0,\zstrat}\left(\neg\always\,\Fixx{i}\right) >0\}
\end{equation}
where $\delta >0$ since $L_0$ is finite.
Let $\E \eqdef \always\,\Alphaset{\le i} \,\vee\, \formula_\nexti$.
By \eqref{eq:alphaOrNextObj} we have for every $\state_0 \in L_0$
\[
\probm_{\?L_\previ,\state_0,\zstrat}(\eventually\E)=1
\]
We now consider the finitely many Markov chains $\chain_{\state_0}$
induced by playing $\zstrat$ in $\?L_\previ$
from the finitely many initial states $\state_0 \in L_0$.
Thus we obtain for every $\state_0 \in L_0$
\begin{equation}\label{eq:cl53-1}
\probm_{\chain_{s_0}}(\eventually\E)=1
\end{equation}
Since $\E$ is suffix-closed,
we can apply \cref{lem:LZO-3} to each Markov chain $\chain_{\state_0}$.
Thus there exist thresholds $l^{\state_0}$ and finite sets 
\[
L^{\state_0} \subseteq \safesub{\chain_{{\state_0}},(\E)}{1-\err{\nexti}{1}}
\]
such that
\begin{equation}\label{eq:cl53-2}
\probm_{\chain_{s_0}}\left(\eventually^{\le l^{\state_0}}L^{\state_0}\right) \ge
\probm_{\chain_{s_0}}\left(\eventually\E\right) - \delta
=
1 - \delta
\end{equation}
where the last equality is due to \eqref{eq:cl53-1}.
Let now
$L_{\nexti}' \eqdef \bigcup_{s_0 \in L_0} L^{\state_0}$
(which is finite, since it is a finite union of finite sets)
and
$l_{\nexti} \eqdef \max_{s_0 \in L_0} l^{\state_0}$
(which is finite as the maximum of a finite set of numbers).

For every $\state_0 \in L_0$ we have
\[
\safesub{\chain_{{\state_0}},(\E)}{1-\err{\nexti}{1}} \subseteq 
\safesub{\?L_\previ,(\always\,\Alphaset{\le i}\lor \formula_\nexti)}{1-\err{\nexti}{1}}
\]
since the required value
for $\E = \always\,\Alphaset{\le i} \,\vee\, \formula_\nexti$
is witnessed by the strategy $\zstrat$, and thus
\[
L_{\nexti}' \subseteq \safesub{\?L_\previ,(\always\,\Alphaset{\le
i}\lor \formula_\nexti)}{1-\err{\nexti}{1}}
\]
as required.
From \eqref{eq:cl53-2} we obtain that for all $\state_0 \in L_0$
\begin{equation}\label{eq:cl53-3}
\probm_{\?L_\previ,s_0,\zstrat}\left(\eventually^{\le l_{\nexti}}L_{\nexti}'\right) \ge
1 - \delta
\end{equation}
Now we are ready to show \eqref{eq:claim:reach-L-next}. We have for all $\state_0 \in L_0$
\begin{align*}
\probm_{\?L_\previ,s_0,\zstrat}&\left(\neg\always\,\Fixx{i} \wedge \eventually^{\le
  l_{\nexti}}L_{\nexti}'\right) \\& = 
\probm_{\?L_\previ,s_0,\zstrat}\left(\eventually^{\le l_{\nexti}}L_{\nexti}'\right)
-
\probm_{\?L_\previ,s_0,\zstrat}\left(\always\,\Fixx{i} \wedge \eventually^{\le
  l_{\nexti}}L_{\nexti}'\right) & \text{law of total prob.}\\
& \ge 
\probm_{\?L_\previ,s_0,\zstrat}\left(\eventually^{\le l_{\nexti}}L_{\nexti}'\right)
-
\probm_{\?L_\previ,s_0,\zstrat}\left(\always\,\Fixx{i}\right)\\
& \ge 
1-\delta
-
(1-\probm_{\?L_\previ,s_0,\zstrat}\left(\neg\always\,\Fixx{i}\right)) &
\text{by \eqref{eq:cl53-3}}\\
& \ge
1 - \err{\nexti}{2} \cdot
\probm_{\?L_\previ,s_0,\zstrat}\left(\neg\always\,\Fixx{i}\right)
-
(1-\probm_{\?L_\previ,s_0,\zstrat}\left(\neg\always\,\Fixx{i}\right)) &
\text{by \eqref{eq:cl53-0}}\\
& =
\probm_{\?L_\previ,s_0,\zstrat}\left(\neg\always\,\Fixx{i}\right)(1-\err{\nexti}{2})
\end{align*}
\end{proof}

\noindent
Notice that $L_{\nexti} \cap \Fixx{i} = \emptyset$, because
every state in $L_{\nexti}$ must have a value
$\ge 1-\err{\nexti}{1}$ for $\formula_\nexti$.
(In the special case of $i=0$ we have $\Fixx{0}=\emptyset$
and $\formula_1 = \formula$
and thus $l_1=0$ and $L_1 = L_1' = L_0$.)
Also recall that
\begin{align*}
L_{\nexti}'
&\subseteq
\safesub{\?L_\previ,(\always\,\Alphaset{\le i}\lor \formula_\nexti)}{1-\err{\nexti}{1}}\\
&\subseteq
\Alphaset{\le i} \cup \safesub{\?L_\previ,\formula_{\nexti}}{1-\err{\nexti}{1}}.
\end{align*}
We define $L_{\nexti}$ as 
$L_{\nexti} \eqdef L_{\nexti}' \setminus \Alphaset{\le i}$.

Since $\formula_{\nexti} \eqdef \formula \land \always (L\setminus\closure{\Fixx{i}})$ and $\formula$ is a parity objective, we can, by \cref{lem-eps-opt} and \cref{rem:quasi-tail}, pick an MD strategy $\obs{\nexti}$ that is
$\err{\nexti}{0}$-optimal for $\formula_{\nexti}$
from all states in $L_{\nexti}$.

Based on this strategy
$\obs{\nexti}$
and parameters $\alpha>\beta>\gamma>0$, we define
$\Alphaset{\nexti} \subseteq \Betaset{\nexti} \subseteq \Gammaset{\nexti}
\subseteq L$
to be the sets of states from which $\obs{\nexti}$ attains at least values
$\alpha,\beta$ and $\gamma$, for $\formula_{\nexti}$, respectively.
E.g.,
\[
\Betaset{\nexti} \eqdef \safesub{\?L_i,\obs{\nexti},\formula_{\nexti}}{\beta}
\]
In particular, this definition satisfies our invariant \eqref{eq:invariant}, i.e.,
$\Gammaset{\nexti} \cap \closure{\Fixx{i}} = \emptyset$,
because a high attainment $\gamma$ for
$\formula_{\nexti} = \formula \land \always (L\setminus \closure{\Fixx{i}}$
requires that $\closure{\Fixx{i}}$ is not visited.

W.l.o.g., by choosing $\err{\nexti}{1}, \err{\nexti}{0}$ sufficiently small, we
can assume that 
$\alpha<(1-{\err{\nexti}{1}}-{\err{\nexti}{0}})$, and therefore that
$L_{\nexti}\subseteq \Alphaset{\nexti} \subseteq \Betaset{\nexti}$
(we only need $\subseteq \Betaset{\nexti}$).

\smallskip
Let $\?L_i'\eqdef\fixin{\?L_i}{\obs{\nexti},\Betaset{\nexti}}$.
Note that in $\?L_i'$ the strategy $\sigma$ might not be able
to reach $L_{\nexti}$ with the same probability as in $\?L_i$,
because the choices in $\Betaset{\nexti}$ are now fixed.
However, a similar strategy $\sigma'$ can reach $\Betaset{\nexti}$ in $\?L_i'$
with at least the probability by which $\sigma$ reaches $L_{\nexti}$ in $\?L_i$.
We now define a new resetting strategy $\sigma'$ in $\?L_i'$.
It behaves like the previous strategy $\sigma$ until (and if) it reaches
$\closure{\Betaset{\nexti}}$. Without restriction we can assume that it
reaches even $\Betaset{\nexti}$ in this case (similar to the argument for
$\sigma$ above).
Then it plays like $\obs{\nexti}$ while in $\Gammaset{\nexti}$.
This is possible, since $\Gammaset{\nexti} \cap \closure{\Fixx{i}} = \emptyset$~ by our invariant \eqref{eq:invariant}.
If and when it exits $\Gammaset{\nexti}$ at some state $s$ then it resets to some
almost surely winning strategy $\sigma(s)$ for $\formula$ in $\?L$ until it reaches
$\Betaset{\nexti}$ (or another previously fixed part) again, etc. 

From \cref{claim:reach-L-next} (\cref{eq:claim:reach-L-next}) and the fact that $\sigma'$ behaves like
$\sigma$ until it reaches $\Betaset{i+1}$ we obtain that
\begin{equation}\label{eq:reach-alpha-or-next-beta}
\begin{aligned}
&\probm_{\?L_i',s_0,\sigma'}\left(\neg\always\,\Fixx{i} \wedge \eventually^{\le l_{i+1}}\left(\Alphaset{\le i}\lor \Betaset{i+1}\right)\right)
\\
&\ge
\probm_{\?L_\previ,s_0,\zstrat}\left(\neg\always\,\Fixx{i}\right)(1-\err{\nexti}{2})\\
&=
\probm_{\?L_\previ',s_0,\zstrat'}\left(\neg\always\,\Fixx{i}\right)(1-\err{\nexti}{2}),
\end{aligned}
\end{equation}
where the last equality holds because $\?L_\previ$ and $\?L_\previ'$
(resp.\ $\zstrat$ and $\zstrat'$) coincide inside $\Fixx{i}$.
Analogously to \cref{claim:as-par:progress-in-B}, from any state in $\Betaset{i+1}$,
the chance of staying in the set $\Gammaset{i+1}$ can be lower-bounded.
\begin{equation}\label{claim:stay-gamma-next}
\forall \state \in \Betaset{i+1}\, \probm_{\?L_i',\state,\sigma'}(\always~\Gammaset{i+1}) \ge \frac{\beta-\gamma}{1-\gamma} > 0
\end{equation}

$\Gammaset{\nexti} \cap \closure{\Fixx{i}} = \emptyset$
by \eqref{eq:invariant}
and $\sigma'$ continues to play $\obs{\nexti}$ in $\Gammaset{i+1}$.
Since $\always~\Gammaset{i+1} \subseteq \eventually\always~\Gammaset{i+1}$,
we can apply Levy's zero-one law
(\cref{col:01law-gamma-alpha})
to \eqref{claim:stay-gamma-next}
and obtain
\begin{equation}\label{eq:get-to-alpha-next}
\forall \state \in \Betaset{i+1}\, \probm_{\?L_i',\state,\sigma'}(\eventually\always~\Alphaset{i+1}) \ge \frac{\beta-\gamma}{1-\gamma} > 0.
\end{equation}
By combining \eqref{eq:reach-alpha-or-next-beta} with
\eqref{eq:get-to-alpha-next}, we get
\begin{equation*}
    \label{eq:reach-alpha-or-next-alpha}
    \begin{aligned}
    &\probm_{\?L_i',s_0,\sigma'}\left(\neg\always\,\Fixx{i} \wedge \eventually^{\le
      l_{i+1}}\left(\Alphaset{\le i}\lor \eventually\always~\Alphaset{i+1}\right) \right) \\
    &\ge
    \left(\probm_{\?L_\previ',s_0,\zstrat'}\left(\neg\always\,\Fixx{i}\right)(1-\err{\nexti}{2})\right)
\frac{\beta-\gamma}{1-\gamma}\\
    \end{aligned}
\end{equation*}
By continuity of measures (recall that $\eventually X = \bigcup_{k\in\N}\eventually^{k} X$),
for every $\err{\nexti}{3}>0$ there must exist a threshold $k_\nexti \ge
l_{i+1}$ of steps such that, for all $\state_0 \in L_0$,
\begin{equation}\label{eq:as-par:reach-phase}
    \begin{aligned}
    &\probm_{\?L_i',s_0,\sigma'}\left(\neg\always\,\Fixx{i} \wedge \eventually^{\le
      k_{i+1}} \Alphaset{\le \nexti}\right) \\
    &\ge
    \left(\probm_{\?L_\previ',s_0,\zstrat'}\left(\neg\always\,\Fixx{i}\right)(1-\err{\nexti}{2})\right)
\frac{\beta-\gamma}{1-\gamma}(1-\err{\nexti}{3}).
    \end{aligned}
\end{equation}
(Since $L_0$ is finite, we can have the same multiplicative error
$(1-\err{\nexti}{3})$ for all $\state_0 \in L_0$.)
(In the special case of $i=0$, we have $k_1=0$, since
$L_0=L_1\subseteq \Alphaset{1}$.)
Once inside $\Alphaset{\le\nexti}$, there is a bounded chance
$\ge \frac{\alpha-\beta}{1-\beta}$
of staying inside $\Betaset{\le\nexti}$ forever, by
\cref{claim:as-par:progress-in-B}.
Thus from \eqref{eq:as-par:reach-phase} we get
\begin{align}
&\probm_{\?L_i',\state_0,\sigma'}
    \label{eq:as-par:reach-phase-stay}
  (\neg\always\,\Fixx{i} \wedge \eventually^{\le k_\nexti}\ \always\ \Betaset{\le\nexti}) \\
&\ge
  \left(
\left(\probm_{\?L_\previ',s_0,\zstrat'}\left(\neg\always\,\Fixx{i}\right)(1-\err{\nexti}{2})\right)
\frac{\beta-\gamma}{1-\gamma}
(1- \err{\nexti}{3})
\right) \frac{\alpha-\beta}{1-\beta}
\nonumber
\end{align}
Consider the finite $k_\nexti$-bubble
$H_{i+1} \eqdef \bubblearound{\?L}{L_0}{k_{i+1}}$
around $L_0$.
Remember that in finite MDPs, there are uniformly optimal MD strategies for reachability
objectives \cite{Ornstein:AMS1969}.
Consequently, since $H_{i+1}$ is finite, there exists an MD strategy
$\ors{\nexti}$ in $\?L_i'$ that is optimal from $H_{i+1}$ for the objective of reaching $\Alphaset{\le \nexti}$
(from $L_0$) \emph{inside} $H_{i+1}$
without leaving $H_{i+1}$.
We fix $\ors{\nexti}$ inside $H_{i+1}$,
and obtain our new MDP
\[
\?L_{\nexti}\eqdef 
\fixinbubble{\?L_i'}{\ors{\nexti}}{L_0}{k_{\nexti}}
\]
See \cref{fig:as-par:sea-urchin} for an illustration
after round $i=3$.
We define $\Fixx{i+1} \eqdef \Betaset{\le i+1} \cup H_{i+1}$ as the region where the strategy
is already fixed in $\?L_{i+1}$.
We now define an almost surely winning resetting strategy $\sigma''$ in
$\?L_{\nexti}$, analogously as $\sigma$ previously in $\?L_i$.
Similarly as in \cref{eq:fixi-implies-parity} for $\?L_{i}$,
we can derive the corresponding property for $\?L_\nexti$.
\begin{equation}\label{eq:fixi-implies-parity-next}
\probm_{\?L_\nexti,s_0,\sigma''}\left(\formula \wedge \always \,\Fixx{\nexti}\right)
=
\probm_{\?L_\nexti,s_0,\sigma''}\left(\always \,\Fixx{\nexti}\right)
\end{equation}

$\?L_\nexti$, $\?L_\previ'$ and $\?L_\previ$ 
(resp. the strategies $\zstrat''$, $\zstrat'$ and
$\zstrat$) coincide inside $\Fixx{i}$.
Thus by \eqref{eq:fixi-implies-parity} we have
\begin{equation}\label{eq:coincide-fixi}
\begin{aligned}
\probm_{\?L_\nexti,s_0,\zstrat''}\left(\always\,\Fixx{i}\right) & = &
\probm_{\?L_\previ',s_0,\zstrat'}\left(\always\,\Fixx{i}\right) \\
& = & \probm_{\?L_\previ,s_0,\zstrat}\left(\always\,\Fixx{i}\right)\\
& = & 
p_i(\state_0)
\end{aligned}
\end{equation}
By the optimality of the reachability strategy $\ors{\nexti}$ that is fixed in
$H_{\nexti}$ and 
$\Fixx{i+1} = \Betaset{\le i+1} \cup H_{i+1}$, we obtain from
this and \cref{eq:as-par:reach-phase-stay} that
\begin{equation}\label{eq:as-par:reach-phase-stay2}
\begin{aligned}
&\probm_{\?L_\nexti,\state_0,\sigma''}
  (\neg\always\,\Fixx{i} \wedge \always\,\Fixx{\nexti}) \\
  &\ge
  \left(
(1 - p_i(s_0)) (1-\err{\nexti}{2})
\frac{\beta-\gamma}{1-\gamma}
(1-\err{\nexti}{3})
\right) \frac{\alpha-\beta}{1-\beta}
\end{aligned}
\end{equation}

The crucial question is how much $\sigma''$ attains for $\formula$ in the fixed part alone,
i.e., how large is 
$
\probm_{\?L_\nexti,s_0,\sigma''}\left(\formula\,
\wedge\, \always\,\Fixx{i+1}\right) = p_{i+1}(s_0)
$
?
For all $\state_0 \in L_0$ we have
\begin{align*}
&p_{i+1}(s_0)\\
&=\probm_{\?L_\nexti,s_0,\sigma''}\left(\always \,\Fixx{\nexti}\right)\\
&=\probm_{\?L_\nexti,s_0,\sigma''}\left(\always\,\Fixx{i}\right) +\probm_{\?L_\nexti,\state_0,\sigma''}(\neg\always\,\Fixx{i} \wedge \always\,\Fixx{\nexti}) \\
&\ge
p_i(s_0)
+ (1-p_i(s_0))\left((1-\err{\nexti}{2})\frac{\beta-\gamma}{1-\gamma}
(1-\err{\nexti}{3})\frac{\alpha-\beta}{1-\beta}\right),
\end{align*}
where the first equality is due to \eqref{eq:fixi-implies-parity-next}
and the last inequation is due to
\cref{eq:as-par:reach-phase-stay2,eq:coincide-fixi}.

We can suitably choose the parameters
$\alpha,\beta,\gamma,\err{\nexti}{2},\err{\nexti}{3}$ such that
$\left((1-\err{\nexti}{2})\frac{\beta-\gamma}{1-\gamma}
(1-\err{\nexti}{3})\frac{\alpha-\beta}{1-\beta}\right)$
is arbitrarily close to $1$, and thus in particular $\ge 1/2$,
and obtain that
$p_{i+1}(s_0) \ge p_i(s_0) + (1-p_i(s_0))/2$.
Since $p_0(s_0)=0$, we get
$1-p_i(s_0) \le 2^{-i}$ and thus
$\lim_{i\rightarrow\infty} p_i(s_0) = 1$, as required.

Finally, let $\hat{\sigma}$ be the MD strategy in $\?L$ that plays from $L_0$
as prescribed by all the fixings in $\bigcup_i \Fixx{i}$ in the systems $\?L_i$.
Then, for all $\state_0\in \states_0$ and every $i \in \N$, it holds that
\[
\probm_{\?L,s_0,\hat{\sigma}}\left(\formula\right)
\ge
\probm_{\?L_i,s_0,\hat{\sigma}}\left(\formula\, \wedge \, \always\,\Fixx{i}\right)
=
p_i(s_0) \ge 1-2^{-i}
\]
Since this holds for every $i \in \N$ we get that
$\probm_{\?L,s_0,\hat{\sigma}}\left(\formula\right) = 1$, i.e., the MD
strategy $\hat{\sigma}$
wins $\formula$ almost surely from every $s_0 \in L_0$.
\end{proof}

\section{Optimal Strategies for $\cParity{\{0,1,2\}}$}
\label{app:as012}
\thmZeroOneTwoPar*

In the rest of this section we prove \cref{thm:012quant}.
It generalizes \cite[Theorem 16]{KMSW2017}, which considers only
finitely-branching MDPs and uses the fact that for every safety
objective, an MD strategy exists that is uniformly \emph{optimal}.
This is not generally true for infinitely-branching acyclic MDPs \cite{KMSW2017}.
To prove \cref{thm:012quant}, we adjust the construction so that it only
requires uniformly \emph{$\eps$-optimal} MD strategies for safety objectives
(in the conditioned MDP $\pmdp$).

\begin{theorem}[from Theorem B in \cite{Ornstein:AMS1969}]\label{thm:reach-eps}
For every MDP $\mdp$ there exist uniform $\eps$-optimal MD-strategies for
reachability objectives.
\end{theorem}

The following simple lemma provides a scheme for proving almost-sure properties.
\newcommand{\lemaspartitionscheme}{
Let $\probm$ be a probability measure over the sample space~$\Omega$.
Let $(\playset_i)_{i \in I}$ be a countable partition of~$\Omega$ in measurable events.
Let ${\?E} \subseteq \Omega$ be a measurable event.
Suppose $\probm(\playset_i \cap {\?E}) = \probm(\playset_i)$ holds for all $i \in I$.
Then $\probm({\?E}) = 1$.
}
\begin{lemma}[Lem.~18 in \cite{KMSW2017}]
    \label{lem-as-partition-scheme}
\lemaspartitionscheme
\end{lemma}

We need a few lemmas about safety objectives first.
Recall the definition of safe sets (\cref{def:safeset}).

\newcommand{\lemaszotreturntosafe}{
  Let $\mdp=\mdptuple$ be an MDP, $\reachset\subseteq \states$, $\zstrat$ a
  strategy from state $\state \in \states$ and $\tau < 1$.
It holds that $\probm_{\mdp,\state,\zstrat}(\eventually \always \neg\safesub{\mdp,\safety{\reachset}}{\tau} \land \eventually \always (\states\setminus\reachset)) = 0$.
}
\begin{lemma}
\label{lem:as012-return-to-safe}
\lemaszotreturntosafe
\end{lemma}
\begin{proof}
For any $n \in \nat$ define $Z_n \eqdef \left(\states\setminus\reachset\right)^n$.
That is, $Z_n S^\omega$ is the event that the first $n$ visited states are
outside $\reachset$.
For every state $\state \not\in \safesub{\mdp,\safety{\reachset}}{\tau}$
and every strategy $\zstrat$ from $\state$
we have that
$\lim_{n \rightarrow \infty} \probm_{\mdp,\state,\zstrat}(Z_n S^\omega) < \tau
< (1+\tau)/2$ by \cref{def:safeset}.
Let $n(\state) \in \nat$ be the smallest number such that $\probm_{\mdp,\state,\zstrat}(Z_{n(s)} S^\omega) \le (1+\tau)/2$.
Let $L \subseteq \states^*$ be the set of finite sequences
$s_0 s_1 \cdots s_{n-1}$ such that $s_0 \not\in
\safesub{\mdp,\safety{\reachset}}{\tau}$
and $n = n(s_0)$ and
$\forall i < n .\, s_i \in \left(\states\setminus\reachset\right) \setminus \safesub{\mdp,\safety{\reachset}}{\tau}$.

We show for all $s \in \states \setminus \safesub{\mdp,\safety{\reachset}}{\tau}$ and all $k \in \nat$ that $\probm_{\mdp,s,\zstrat}(L^k \states^\omega) \le \left(\frac{1+\tau}{2}\right)^k$.
We proceed by induction on~$k$.
The case $k=0$ is trivial.
For the induction step let $k \ge 0$.
\begin{align*}
\probm_{\mdp,s,\zstrat}(L^{k+1} \states^\omega)
& \le \probm_{\mdp,s,\zstrat}(Z_{n(s)} L^{k} \states^\omega) \\ %
& \le \probm_{\mdp,s,\zstrat}(Z_{n(s)} \states^\omega)\;\;\;\cdot\;\smashoperator{\sup_{{s' \in \states \setminus \safesub{\mdp,\safety{\reachset}}{\tau}}}}~\probm_{\mdp,s',\zstrat}(L^{k} \states^\omega) \\
& \le \probm_{\mdp,s,\zstrat}(Z_{n(s)} \states^\omega) \cdot \left(\frac{1+\tau}{2}\right)^k %
\le \left(\frac{1+\tau}{2}\right)^{k+1}  %
\end{align*}
where the first inequality uses that $L \cap \{s\} S^* \subseteq Z_{n(s)}$,
the third uses the induction hypothesis, and the last the definition of $n(s)$.
This completes the induction proof.

Write $\formula \eqdef \always \neg\safesub{\mdp,\safety{\reachset}}{\tau} \land \always \left(\states\setminus\reachset\right)$.
For all $s \in \states$,
\begin{align*}
\probm_{\mdp,s,\zstrat}(\formula)
& = \probm_{\mdp,s,\zstrat}(L^\omega) && \text{because $\denotationof{\formula}{} = L^\omega$}\\
& = \lim_{k \to \infty} \probm_{\mdp,s,\zstrat}(L^k \states^\omega) &&
\text{by continuity of measures}\\
& \le \lim_{k \to \infty} \left(\frac{1+\tau}{2}\right)^k && \text{as shown above} \\
& = 0 && \text{because $\tau < 1$}
\end{align*}
It follows that
$\probm_{\mdp,s,\zstrat}(\next^j \formula) = 0$,
for all $s \in \states$ and all $j \in \nat$
and therefore that
\begin{align*}
&\probm_{\mdp,\state,\zstrat}(\eventually \always \safesub{\mdp,\safety{\reachset}}{\tau} \land \eventually \always \left(\states\setminus\reachset\right))\\
& \ = \ \probm_{\mdp,\state,\zstrat}(\eventually \formula) \\
& \ = \ \probm_{\mdp,\state,\zstrat}\Big(\bigcup_{j \in \nat} \denotationof{\next^j \formula}{s}\Big) \\
& \ \le \ \sum_{j \in \nat} \probm_{\mdp,\state,\zstrat}(\next^j \formula)\  = \ 0 \qedhere
\end{align*}
\end{proof}

Now we show that if an MDP admits uniformly $\eps$-optimal strategies for all
safety objectives, then optimal strategies for $\cParity{\{0,1,2\}}$ (where
they exist) can be chosen MD.

\begin{lemma}
\label{prop:as012}
Let $\mdp=\mdptuple$ be an MDP such that for every safety objective (given by
some target set $T\subseteq \states$) and $\eps >0$ there exists a uniformly
$\eps$-optimal MD strategy.
Let $\state_0 \in \states$,
$\coloring:\states\to \{0,1,2\}$, $\formula=\Parity{\coloring}$,
and $\zstrat$ a strategy with
$\probm_{\mdp,\state_0,\zstrat}(\formula) = 1$.
Then there is an MD-strategy~$\zstrat'$ with $\probm_{\mdp,\state_0,\zstrat'}(\formula) = 1$.
\end{lemma}
\begin{proof}%
To achieve an almost-sure winning objective, the player must forever remain in states from which the objective can be achieved almost surely.
So we can assume without loss of generality that all states are almost-sure winning, i.e., for all $\state \in \states$ we have $\probm_{\mdp,\state,\zstrat}(\formula) = 1$ for some strategy $\zstrat$.
We will define an MD-strategy~$\zstrat'$ with $\probm_{\mdp,\state,\zstrat'}(\formula) = 1$ for all $\state \in \states$.

Recall that $\colorset{\states}{\neq}{0}\subseteq \states$ denotes the subset of states of color $1$ or $2$.
Let $\reachset\eqdef\colorset{\states}{\neq}{0}$ and let $\epsoptav$ be a
uniformly $\eps$-optimal MD strategy
for $\safety{\reachset}$,
whose existence is guaranteed by our assumption on $\mdp$.
The precise $\eps>0$ is immaterial, we only need that $\eps < \frac13$. %
The MD-strategy $\zstrat'$ will be based on special subsets (\cref{def:safeset}):
\begin{equation}\label{eq:def_safeM}
\safesub{\mdp}{\tau} \eqdef \safesub{\mdp,\epsoptav,\safety{\reachset}}{\tau}\subseteq \states.
\end{equation}
We first define the MD-strategy~$\zstrat'$ partially for the states in $\safesub{\mdp}{\frac13}$ and then extend the definition of~$\zstrat'$ to all states.
For the states in $\safesub{\mdp}{\frac13}$ define $\zstrat' \eqdef \epsoptav$
(which is MD).
Let $\mdp'$ be the MDP obtained from~$\mdp$ by restricting the transition
relation as prescribed by the partial MD-strategy~$\zstrat'$
in $\safesub{\mdp}{\frac13}$ (elsewhere the choices remain free).
We define $\safesub{\mdp'}{\tau}$ for $\mdp'$ as in
Equation~\eqref{eq:def_safeM} for $\mdp$. Thus,
for any $\tau \in [0,1]$, we have $\safesub{\mdp}{\tau} = \safesub{\mdp'}{\tau}$.
Indeed, since $\mdp'$ restricts the options of the player,
we have $\safesub{\mdp}{\tau} \supseteq \safesub{\mdp'}{\tau}$.
Conversely, let $\state \in \safesub{\mdp}{\tau}$.
The strategy $\epsoptav$
attains $\probm_{\mdp,\state,\epsoptav}(\always \colorset \states = 0) \ge \tau$.
Since $\epsoptav$ can be applied in~$\mdp'$, and results in the same Markov chain
as applying it in $\mdp$, we conclude $\state \in \safesub{\mdp'}{\tau}$.
This justifies to write $\safe{\tau}$ for $\safesub{\mdp}{\tau} = \safesub{\mdp'}{\tau}$ in the remainder of the proof.

Next we show that, also in~$\mdp'$, for all states $\state \in \states$ there exists a strategy~$\zstrat_1$ with $\probm_{\mdp',\state,\zstrat_1}(\formula) = 1$.
This strategy~$\zstrat_1$ is defined as follows.
First play according to an almost-surely winning strategy~$\zstrat$ from the statement of the theorem.
If and when the play visits $\safe{\frac13}$, switch to the
MD-strategy~$\epsoptav$.
If and when the play then visits $\colorset \states \ne 0$, switch back to an almost-surely winning strategy~$\zstrat$ from the statement of the theorem, and so forth.

We show that $\zstrat_1$ attains $\probm_{\mdp',\state,\zstrat_1}(\formula) = 1$.
To this end we will use Lemma~\ref{lem-as-partition-scheme}.
We partition the runs of $s S^\omega$ into three events $\playset_0, \playset_1, \playset_2$ as follows:
\begin{itemize}
\item
$\playset_0$ contains the runs where $\zstrat_1$ switches between $\epsoptav$ and~$\zstrat$ infinitely often.
\item
$\playset_1$ contains the runs where $\zstrat_1$ eventually only plays according to~$\epsoptav$.
\item
$\playset_2$ contains the runs where $\zstrat_1$ eventually only plays according to~$\zstrat$.
\end{itemize}
Each time $\zstrat_1$ switches to~$\epsoptav$, there is, by
definition of $\safesub{\mdp}{\frac13}$,
a probability of at least~$\frac13$ of never visiting a color-$\{1,2\}$ state again and thus of never again switching to~$\zstrat$.
It follows that $\probm_{\mdp',\state,\zstrat_1}(\playset_0) = 0$.
By the definition of the switching behavior of $\zstrat_1$,
we have $\playset_1 \subseteq \denotationof{\eventually \always \colorset \states = 0}{} \subseteq \denotationof{\formula}{}$, and hence $\probm_{\mdp',\state,\zstrat_1}(\playset_1 \cap \denotationof{\formula}{}) = \probm_{\mdp',\state,\zstrat_1}(\playset_1)$.
Since $\probm_{\mdp,\state,\zstrat}(\formula) = 1$ and $\formula$ is tail, we have that
$\probm_{\mdp',\state,\zstrat_1}(\playset_2 \cap \denotationof{\formula}{}) = \probm_{\mdp',\state,\zstrat_1}(\playset_2)$.
Using Lemma~\ref{lem-as-partition-scheme}, we obtain $\probm_{\mdp',\state,\zstrat_1}(\formula) = 1$.

Next we show that for all $s \in \states$ the strategy~$\zstrat_1$ defined above achieves $\probm_{\mdp',\state,\zstrat_1}(\eventually \safe{\frac23} \lor \eventually \colorset \states = 2) = 1$.
To this end we will use Lemma~\ref{lem-as-partition-scheme} again.
We partition the runs of $s S^\omega$ into three events $\playset_1', \playset_2', \playset_0'$ as follows:
\begin{itemize}
\setlength\itemsep{.8em}
\item
$\playset_1' = \denotationof{\eventually \always \colorset \states = 0}{s}$
\item
$\playset_2' = \denotationof{\always \eventually \colorset \states = 2}{s}$
\item
$\playset_0' = s S^\omega \setminus \denotationof{\formula}{s}$
\end{itemize}
We have previously shown that $\probm_{\mdp',\state,\zstrat_1}(\formula) = 1$ and hence that $\probm_{\mdp',\state,\zstrat_1}(\playset_0') = 0$.
We now invoke Lemma~\ref{lem:as012-return-to-safe}
with $\tau \eqdef \frac23 +\eps < 1$
and obtain that almost all runs in~$\playset_1'$
satisfy $\always \eventually\safesub{\mdp,\safety{\reachset}}{\tau}$.
Since $\epsoptav$ is uniformly $\eps$-optimal for $\safety{\reachset}$ we have
$\safesub{\mdp,\safety{\reachset}}{\tau} \subseteq \safe{\tau -\eps} = \safe{\frac23}$
and thus
almost all runs in~$\playset_1'$
satisfy
$\always \eventually \safe{\frac23}$.
Since $\denotationof{\always \eventually \safe{\frac23}}{} \subseteq \denotationof{\eventually \safe{\frac23}}{}$,
we observe that
\[
\probm_{\mdp',\state,\zstrat_1}(\playset_1' \cap \denotationof{\eventually \safe{\frac23} \lor \eventually \colorset \states = 2}{}) = \probm_{\mdp',\state,\zstrat_1}(\playset_1').
\]
Since $\playset_2' \subseteq \denotationof{\eventually \colorset \states = 2}{}$, we also have that
\[
    \probm_{\mdp',\state,\zstrat_1}(\playset_2' \cap \denotationof{\eventually \safe{\frac23} \lor \eventually \colorset \states = 2}{}) = \probm_{\mdp',\state,\zstrat_1}(\playset_2').
\]
By \cref{lem-as-partition-scheme} we obtain $\probm_{\mdp',\state,\zstrat_1}(\eventually \safe{\frac23} \lor \eventually \colorset \states = 2) = 1$.

Writing $\reachset' \eqdef \safe{\frac23} \cup \colorset \states = 2$ we have just shown that for all $s \in \states$ there is a strategy~$\zstrat_1$ with $\probm_{\mdp',\state,\zstrat_1}(\eventually \reachset') = 1$.
Since this holds for all $s \in \states$, it follows from \cref{thm:reach-eps}
that there is an MD-strategy~$\hat\zstrat$ for~$\mdp'$ with $\probm_{\mdp',\state,\hat\zstrat}(\eventually \reachset') = 1$ for all $s \in \states$.
We extend the (so far partially defined) strategy~$\zstrat'$ by~$\hat\zstrat$.
Thus we obtain a (fully defined) strategy~$\zstrat'$ for~$\mdp$ such that for all $s \in \states$ we have $\probm_{\mdp,\state,\zstrat'}(\eventually \reachset') = 1$.

It remains to show that
$\probm_{\mdp,\state,\zstrat'}(\formula) = 1$ holds for all $s \in \states$.
To this end we will use Lemma~\ref{lem-as-partition-scheme} again.
We partition the runs of $s S^\omega$ into two events $\playset_1'', \playset_2''$:
\begin{itemize}
\setlength\itemsep{.6em}
\item
$\playset_1'' = \denotationof{\always\eventually \safe{\frac23}}{s}$, i.e., $\playset_1''$ contains the runs that visit $\safe{\frac23}$ infinitely often.
\item
$\playset_2'' = \denotationof{\eventually\always \neg\safe{\frac23}}{s}$, i.e., $\playset_2''$ contains the runs that from some point on never visit $\safe{\frac23}$.
\end{itemize}
Recall that $\zstrat'$ plays like $\epsoptav$ inside of $\safe{\frac13}$, that
$\safe{\frac23} \subseteq \safe{\frac13}$, and that $\epsoptav$ is an
MD-strategy. %
Thus we can invoke Lemma~\ref{lem:LZO-1} with
$\beta_2 \eqdef \frac{2}{3}$ and $\beta_1 \eqdef \frac{1}{3}$ and conclude that
every time a run (according to $\zstrat'$) enters $\safe{\frac23}$,
the probability that the run remains in $\safe{\frac13}$ forever is at least $\frac12$.
It follows that almost all runs in~$\playset_1''$ eventually remain in $\safe{\frac13}$ forever. That is,
$\probm_{\mdp,\state,\zstrat'}(\playset_1'' \cap \denotationof{\eventually \always \safe{\frac13}}{}) = \probm_{\mdp,\state,\zstrat'}(\playset_1'')$.
Since $\safe{\frac13} \subseteq \colorset \states = 0$, we have $\denotationof{\eventually \always \safe{\frac13}}{} \subseteq \denotationof{\eventually \always \colorset \states = 0}{} \subseteq \denotationof{\formula}{}$.
Hence also $\probm_{\mdp,\state,\zstrat'}(\playset_1'' \cap \denotationof{\formula}{}) = \probm_{\mdp,\state,\zstrat'}(\playset_1'')$.

We have previously shown that $\probm_{\mdp,\state,\zstrat'}(\eventually \reachset') = 1$ holds for all $\state \in \states$.
Hence also $\probm_{\mdp,\state,\zstrat'}(\always \eventually \reachset') = 1$ holds for all $\state \in \states$.
In particular, almost all runs in~$\playset_2''$ satisfy $\always \eventually \reachset'$.
By comparing the definitions of $\playset_2''$ and $\reachset'$ we see that almost all runs in~$\playset_2''$ even satisfy $\always \eventually \colorset \states = 2$.
Since $\denotationof{\always \eventually \colorset \states = 2}{} \subseteq \denotationof{\formula}{}$, we obtain $\probm_{\mdp,\state,\zstrat'}(\playset_2'' \cap \denotationof{\formula}{}) = \probm_{\mdp,\state,\zstrat'}(\playset_2'')$.
A final application of Lemma~\ref{lem-as-partition-scheme} yields
$\probm_{\mdp,\state,\zstrat'}(\formula) = 1$ for all $s \in \states$.
\end{proof}

We are ready to prove \cref{thm:012quant}.

\begin{proof}[Proof of \cref{thm:012quant}]
Let $\mdp=\mdptuple$ be an MDP and $\formula$ a $\cParity{\{0,1,2\}}$ objective.
Since $\formula$ is tail, it is possible to define (see \cref{def:conditionedmdp})
the conditioned version
$\pmdp=\tuple{\pstates,\pzstates,\prstates,\ptransition,\pprobp}$
of $\mdp$ wrt.\ $\formula$.
Assume that in $\pmdp$ 
for every safety objective (given by some target $T\subseteq\pstates$)
and $\eps>0$ there exists a uniformly $\eps$-optimal MD strategy.
Let $\states_{\mathit{opt}}$ be the subset of states that have an optimal
strategy for $\formula$ in $\mdp$.

By \cref{thm:reduction-to-as}.1, all states in $\pmdp$ are almost sure winning
for $\formula$.
By our above condition about safety objectives in $\pmdp$,
we can apply \cref{prop:as012} to $\pmdp$ and obtain that for every state in
$\pmdp$ there is an MD strategy that is almost surely winning for $\formula$.
By \cref{thm:reduction-to-as}.2, there is an MD strategy in $\mdp$ that is
optimal for $\formula$ from every state in $\states_{\mathit{opt}}$, as required.
\end{proof}

In order to apply \cref{thm:012quant} to infinitely-branching acyclic MDPs,
we now show that acyclicity guarantees the existence of
uniformly $\eps$-optimal MD strategies for safety objectives.

\thmEpsOptSafety*
\begin{proof}
Let $\safety{\reachset}$ be the safety objective and shortly write
$\valueof{}{s_0}=\valueof{\mdp,\safety{\reachset}}{s_0}$ for the value of a state $\state_0$ w.r.t.~this objective.
Assume w.l.o.g.\ that the target $\reachset\subseteq\states$ is a sink and let $\iota:\states\to\N$ be an enumeration of the state space. %

Let $\zstrat$ be an MD-strategy that, at any state $\state\in\zstates$, picks a
successor $\state'$ such that
$$\valueof{}{\state'}\quad\ge\quad \valueof{}{\state} (1- \eps 2^{-\iota(\state)}).$$
We show that
$\probm_{\mdp,\state_0,\sigma}(\safety{\reachset}) \ge \valueof{}{\state_0}(1-\eps)$
holds for every initial state $\state_0$.

Let's write $\poststar{s}\subseteq \states$ for the set of states reachable from state $s\in\states$
and define $\ferr{s} \eqdef\prod_{s'\in\poststar{s}}(1-\eps2^{-\iota(s')})$.
Further, let $\valueof{s}{n}$ be the random variable denoting the value of the $n$th state of a random run that starts in $s$.
In particular, $\valueof{s}{0}=\valueof{}{\state}$.
An induction on $n$ using our choice of strategy gives, for every $\state_0\in S$, that
\begin{equation}
    \label{eq:exi-1}
    \expectation(\valueof{s_0}{n})
    \ge \valueof{}{s_0}\ferr{s_0}.
\end{equation}
Indeed, this trivially holds for $n=0$. For the induction step there are two cases.

\noindent
\emph{Case 1: $s_0\in \zstates$ and $\sigma(\state_0)=s$.} Then
\begin{align*}
    &\expectation(\valueof{s_0}{n+1})\\
    &=
    \expectation(\valueof{s}{n})\\
    &\ge
    \valueof{}{s}\ferr{s}
    &\text{ind. hyp.}\\
    &\ge
    \valueof{}{\state_0}
    \left(1-\eps2^{-\iota(\state_0)}\right)
    \ferr{s}
    &\text{def. of $\sigma$}\\
    &\ge
    \valueof{}{s_0}\ferr{s_0}
    &\text{acyclicity; def.~of } \ferr{s_0}.
\end{align*}

\noindent
\emph{Case 2: $s_0\in \rstates$.} Then
\begin{align*}
    &\expectation(\valueof{s_0}{n+1})\\
    &=
    \sum_{\state_0\transition{} s}\probp(\state_0)(s) \cdot \expectation(\valueof{s}{n})
    \\
    &\ge
    \sum_{\state_0\transition{} s}\probp(\state_0)(s) \cdot \valueof{}{\state}\ferr{s}
    \quad\quad\text{by ind.~hyp.}\\
    &\ge
    \sum_{\state_0\transition{} s}\probp(\state_0)(s) \cdot \valueof{}{\state}
    \left(1-\eps2^{-\iota(\state_0)}\right)
    \ferr{s}
    \\
    &\ge
    \sum_{\state_0\transition{} s}\probp(\state_0)(s) \cdot \valueof{}{s}
    \ferr{\state_0}
    \quad\text{acyclicity; def.~of } \ferr{\state_0}\\
    &=
    \valueof{}{s_0}\ferr{s_0}.
\end{align*}

Together with the observation that $\ferr{s_0}>(1-\eps)$ for every $\state_0$,
we derive that
    \begin{equation}
        \label{eq:exi-lim}
        \liminf_{n\to\infty} \expectation(\valueof{\state_0}{n}) \ge \valueof{}{\state_0}(1-\eps).
    \end{equation}

    To show the claim, fix $\state_0\in \states$ and
shortly write $\probm$ for $\probm_{\mdp,\state_0,\zstrat}$ here.
Let $\RVNotInT{n} : \states^\omega \to \{0,1\}$ be the random variable that indicates that the $n$th state is not in the target set $\reachset$.
Note that $[\next^n \neg \reachset] \ge \valueof{\state_0}{n}$ because target states have value $0$.
We conclude that
\begin{align*}
    &\probm(\safety{\reachset})\\
    &=\quad\probm\left(\bigcap_{i=0}^\infty{\denotationof{X^i\lnot T}{}}\right)
 && \text{semantics of~$\safety{\reachset}=\always\neg \reachset$} \\
 & =\quad\smashoperator{\lim_{n\to\infty}} \probm\left(\bigcap_{i=0}^n \denotationof{X^i\lnot T}{} \right)
 && \text{cont.~of measures from above} \\
 & =\quad \smashoperator{\lim_{n\to\infty}} \probm\left(\denotationof{X^i\lnot T}{}\right)
 && \text{$\reachset$ is a sink} \\
 & =\quad\smashoperator{\lim_{n\to\infty}} \expectation(\RVNotInT{n})
 && \text{definition of $\RVNotInT{n}$} \\
& \ge\quad \liminf_{n\to\infty} \expectation(\valueof{}{n})
 && \text{as $\RVNotInT{n} \ge \valueof{\state_0}{n}$}\\
 & \ge\quad \valueof{}{\state_0}(1-\eps)
 && \text{by \cref{eq:exi-lim}.}
\qedhere
\end{align*}
\end{proof}

\section{$\eps$-Optimal Strategies for $\cParity{\{0,1\}}$}\label{app:eps01}
\thmcoBuchi*

\begin{proof}
Let $\mdp=\mdptuple$ be an MDP such that 
for all safety objectives,
uniformly 
$\eps$-optimal strategies can be chosen MD.
Let $\coloring:\states\to \{0,1\}$ be a coloring and $\formula=\Parity{\coloring}$ be the resulting
co-B\"uchi objective.

We show that there exist uniformly $\eps$-optimal
MD-strategies for $\formula$.
I.e., for every $\eps >0$ there is an MD-strategy $\zstrat_\eps$
with $\forall_{\state_0 \in \states}\probm_{\mdp,\state_0,\zstrat_\eps}(\formula) \ge
\valueof{\mdp}{\state_0} - \eps$.

To construct this MD-strategy strategy $\zstrat_\eps$, we first need several
auxiliary notions.

Let $\eps_1 >0$ be a suitably small number (to be determined later)
and $\tau_1 \eqdef 1-\eps_1 >0$.
Let $\tau_2 \eqdef 1 - \eps_1/k + \lambda$ for a suitably large $k \ge 1$ (to be
determined later) and let $\lambda < \eps_1/k$ (e.g., $\lambda \eqdef \eps_1/(2k)$).
Thus $\tau_2 <1$.

Let $\reachset \eqdef \coloring^{-1}(\{1\})$ be the set of states with color $1$
and $\safety{\reachset}$ the safety objective.
We have that $\formula = \eventually \always (\states\setminus\reachset)$.

By our assumption on $\mdp$, there exists a uniformly $\lambda$-optimal
MD-strategy $\optav$ for $\safety{\reachset}$.
Let $\states' \eqdef \safesub{\mdp,\optav,\safety{\reachset}}{\tau_1}$
be the set of states where $\optav$ achieves at least value $\tau_1$ for $\safety{\reachset}$
(refer to \cref{def:safeset} for the definition of safe sets).
In particular $\states' \subseteq \states\setminus\reachset$, since $\tau_1 >0$.
From $\mdp$ we obtain a modified MDP $\mdp'$ by fixing all player choices from states in
$\states'$ according to $\optav$.

We show that the value w.r.t.\ objective $\formula$
is only slightly smaller in $\mdp'$, i.e.,
\begin{equation}\label{thm:coBuchi:small-loss}
\valueof{\mdp'}{\state_0} \ge \valueof{\mdp}{\state_0} - \eps_1
\quad\mbox{for every $\state_0 \in \states$.}
\end{equation}
Let $\state_0 \in \states$.
By definition of the value $\valueof{\mdp}{\state_0}$,
for every $\delta>0$ there exists a strategy $\zstrat_\delta$ in $\mdp$ from $\state_0$
s.t.\ $\probm_{\mdp,\state_0,\zstrat_\delta}(\formula) \ge \valueof{\mdp}{\state_0} -\delta$.
We define a strategy $\zstrat_\delta'$ in $\mdp'$ from state $\state_0$ as
follows.
First play like $\zstrat_\delta$. If and when a state in $\states'$ is
reached, then henceforth play like $\optav$. This is possible, since no moves
from states outside $\states'$ have been fixed in $\mdp'$,
and all moves from states inside $\states'$ have been fixed
according to $\optav$.
Then we have:
\[
\begin{aligned}
&\probm_{\mdp',\state_0,\zstrat_\delta'}(\formula) \\[1mm]
&=   \probm_{\mdp,\state_0,\zstrat_\delta}(\formula) \\
&
\quad \;  - \, \probm_{\mdp,\state_0,\zstrat_\delta}(\eventually \states') \cdot
\probm_{\mdp,\state_0,\zstrat_\delta}(\formula \mid \eventually \states') \\
&  \quad \; + \,\probm_{\mdp,\state_0,\zstrat_\delta}(\eventually \states') \cdot
\probm_{\mdp',\state_0,\zstrat_\delta'}(\formula \mid \eventually \states') \\[1mm]
& \ge   \probm_{\mdp,\state_0,\zstrat_\delta}(\formula) \\
&
\quad \;- \,  \probm_{\mdp,\state_0,\zstrat_\delta}(\eventually \states') \cdot
\probm_{\mdp,\state_0,\zstrat_\delta}(\formula \mid \eventually \states')  \\
& \quad \; + \, \probm_{\mdp,\state_0,\zstrat_\delta}(\eventually \states') \cdot \tau_1 \\[1mm]
& \ge  \valueof{\mdp}{\state_0} - \delta - \probm_{\mdp,\state_0,\zstrat_\delta}(\eventually
\states')(1-\tau_1)\\[1mm]
& \ge  \valueof{\mdp}{\state_0} - \delta - \eps_1
\end{aligned}
\]
Since this holds for every $\delta >0$ we have
$\valueof{\mdp'}{\state_0} \ge \valueof{\mdp}{\state_0} - \eps_1$, thus (\ref{thm:coBuchi:small-loss}).

Let $\states'' \eqdef \safesub{\mdp,\optav,\safety{\reachset}}{\tau_2-\lambda}$
as by definition of safe sets in \cref{def:safeset}.
In particular, $\states'' = \safesub{\mdp',\optav,\safety{\reachset}}{\tau_2-\lambda}$,
since $\optav$ is a uniform MD-strategy that has been fixed on the subset $\states'$
in the step from $\mdp$ to $\mdp'$.

Let $\state_0\in \states$ be an arbitrary state.
By definition of the value w.r.t.\ $\formula$ of $\state_0$, for every $\eps' >0$ there exists a
strategy $\zstrat_{\eps'}$ from $\state_0$ in $\mdp'$ with
$\probm_{\mdp',\state_0,\zstrat_{\eps'}}(\formula) \ge
\valueof{\mdp'}{\state_0} -\eps'$.

Since $\tau_2 <1$, can we apply Lemma~\ref{lem:as012-return-to-safe} and obtain
$\probm_{\mdp',\state_0,\hat{\zstrat}}(\eventually \safesub{\mdp',\safety{\reachset}}{\tau_2}) \ge
\probm_{\mdp',\state_0,\hat{\zstrat}}(\formula)$ for every strategy
$\hat{\zstrat}$ from $\state_0$ and
thus in particular for $\zstrat_{\eps'}$.

Therefore, $\probm_{\mdp',\state_0,\zstrat_{\eps'}}(\eventually \safesub{\mdp',\safety{\reachset}}{\tau_2}) \ge
\valueof{\mdp'}{\state_0} -\eps'$.
Since this holds for every $\eps' >0$,
in $\mdp'$ the value of $\state_0$ w.r.t.\ the reachability objective
$\eventually \safesub{\mdp',\safety{\reachset}}{\tau_2}$ is
$\ge \valueof{\mdp'}{\state_0}$ for every state $\state_0$.

By Theorem~\ref{thm:reach-eps}, for every $\eps_2 >0$
there exists a uniformly $\eps_2$-optimal MD-strategy $\zstrat'$
in $\mdp'$ for this reachability objective.
So we obtain
\begin{equation}\label{thm:coBuchi:reachcore}
\probm_{\mdp',\state_0,\zstrat'}(\eventually \safesub{\mdp',\safety{\reachset}}{\tau_2}) \ge
\valueof{\mdp'}{\state_0} - \eps_2\quad\mbox{for every state $\state_0$}.
\end{equation}

In particular, $\zstrat'$ must coincide with $\optav$
at all states in $\states'$, since in $\mdp'$ these choices are
already fixed.

Since $\sigma$ is a uniformly $\lambda$-optimal MD-strategy for
$\safety{\reachset}$ in $\mdp$ and $\mdp'$,
we have
$
\safesub{\mdp',\safety{\reachset}}{\tau_2} \subseteq
\safesub{\mdp',\sigma,\safety{\reachset}}{\tau_2-\lambda} = \states''
$
and thus by (\ref{thm:coBuchi:reachcore}) we get
\begin{equation}\label{thm:coBuchi:reachcore2}
\probm_{\mdp',\state_0,\zstrat'}(\eventually \states'') \ge
\probm_{\mdp',\state_0,\zstrat'}(\eventually
\safesub{\mdp',\safety{\reachset}}{\tau_2}) \ge
\valueof{\mdp'}{\state_0} - \eps_2.
\end{equation}

We obtain the MD-strategy $\zstrat_\eps$ in $\mdp$ by
combining the MD-strategies $\zstrat'$ and $\optav$.
The strategy $\zstrat_\eps$ plays like
$\optav$ at all states inside $\states'$ and like
$\zstrat'$ at all states outside $\states'$ (i.e., at $\states \setminus \states'$).

In order to show that $\zstrat_\eps$ has the required property
$\probm_{\mdp,\state_0,\zstrat_\eps}(\formula) \ge \valueof{\mdp}{\state_0} - \eps$, we
first estimate the probability that a play according to $\zstrat_\eps$ will never
leave the set $\states'$ after having visited a state in
$\states''$.

Let $\state \in \states''$.
Then, by applying \cref{lem:LZO-1} to the Markov chain obtained from applying the MD-strategy $\zstrat_\eps$ to~$\mdp$, we obtain
\begin{equation}\label{thm:coBuchi:stayonion}
\begin{aligned}
\probm_{\mdp,\state,\optav}(\always \states') &\quad\ge \quad \frac{(\tau_2 -\lambda) - \tau_1}{1 - \tau_1}\\[1mm]
	& \quad= \quad \frac{(1-\eps_1/k) - (1-\eps_1)}{\eps_1}\\
	& \quad= \quad  1-\frac{1}{k}\,.
\end{aligned}
\end{equation}
In particular we also have
$\probm_{\mdp,\state,\zstrat_\eps}(\always \states') \ge
1-\frac{1}{k}$,
since $\zstrat_\eps$ coincides with $\optav$ inside the set
$\states'$.
Finally we obtain for every $\state_0 \in \states$
\[
\begin{array}{lcll}
\probm_{\mdp,\state_0,\zstrat_\eps}(\formula)
& = & \probm_{\mdp,\state_0,\zstrat_\eps}(\eventually \always
(\states\setminus\reachset)) & \mbox{by def. of $\formula$}\\
& \ge &
\probm_{\mdp,\state_0,\zstrat_\eps}(\eventually \states'') %
{}\cdot
\probm_{\mdp,\state_0,\zstrat_\eps}(\eventually\always \states' \mid
\eventually \states'') & \mbox{since $\states' \subseteq \states\setminus\reachset$}\\[1mm]
& \ge &
\probm_{\mdp',\state_0,\zstrat'}(\eventually \states'')
\cdot (1-1/k) & \mbox{by (\ref{thm:coBuchi:stayonion})}
\\[1mm]
& \ge &
(\valueof{\mdp'}{\state_0} - \eps_2)
\cdot
(1-1/k) & \mbox{by (\ref{thm:coBuchi:reachcore2})} \\[1mm]
& \ge & (\valueof{\mdp}{\state_0} - \eps_1 - \eps_2)
\cdot (1-1/k) & \mbox{by (\ref{thm:coBuchi:small-loss})}
\end{array}
\]
This holds for every $1 > \eps_1, \eps_2 >0$ and
every $k \ge 1$, and moreover $\valueof{\mdp}{\state_0} \le 1$.
Thus we can set $\eps_1 = \eps_2 \eqdef \eps/4$ and
$k \eqdef \frac{2}{\eps}$
and obtain $\probm_{\mdp,\state_0,\zstrat_\eps}(\formula)
\ge \valueof{\mdp}{\state_0} -\eps$ for every $\state_0 \in \states$ as required.
\end{proof}

\ignore{
The precondition of \cref{thm:coBuchi} (uniformly $\eps$-optimal
MD-strategies for safety objectives) is satisfied by many classes of MDPs,
e.g., the following.

\begin{corollary}\label{cor:coBuchi}
\
\begin{enumerate}
\item
Acyclic MDPs and finitely branching MDPs admit uniformly $\eps$-optimal
MD-strategies for co-B\"uchi objectives.
\item
Every MDP admits $\eps$-optimal deterministic Markov-strategies for co-B\"uchi
objectives.
\end{enumerate}
\end{corollary}
\begin{proof}
Towards (1), acyclic MDPs admit uniformly $\eps$-optimal
MD-strategies for safety by \Cref{thm:eps-optimal-safety}, and
finitely branching MDPs trivially admit even uniformly optimal MD-strategies for safety.
Then apply \cref{thm:coBuchi}.

Towards (2), this follows directly from (1) for acyclic MDPs and
\cref{lem:acyclic-Markov}.
\end{proof}
}

\end{document}